\documentclass[12pt]{article}
\usepackage[utf8]{inputenc}
\usepackage{tempora}
\usepackage{csquotes}
\usepackage{hyperref}
\usepackage{etex}

\usepackage[T1]{fontenc}
\usepackage{stmaryrd}			
\usepackage{lmodern}					% harmonie des polices
\usepackage{textcomp}				% table de symboles
\usepackage{verbatim}			% pour faire des commentaires par block avec begin{comment}...end{comment}
\usepackage[english]{babel}
\usepackage{appendix} 
\usepackage{amsmath,					% packages mathmatiques
			amsthm,
			amsfonts,
			amssymb, 
			mathtools}
\usepackage{slashed}
\usepackage{enumitem} 
\usepackage{todonotes}
\usepackage{ytableau}
\usepackage{youngtab}

\usepackage{amssymb}

\usepackage{authblk}

\usepackage[all]{xy}
\usepackage{tikz-cd} 
\usetikzlibrary{tikzmark}
\usetikzlibrary{arrows}

\usepackage{graphicx}	% pour inclure des figures (.pdf, .jpg)
\usepackage{tabularx}
\def\gg{\mathfrak g}

\def\be{\begin{eqnarray}}

\def\ee{\end{eqnarray}}

\newcommand{\bigslant}[2]{{\left.\raisebox{.2em}{$#1$}\middle/\raisebox{-.2em}{$#2$}\right.}}

\newtheorem{definition}{Definition}
\newtheorem{lemma}[definition]{Lemma}
\newtheorem{theorem}[definition]{Theorem}

\newtheorem{example}{Example}

\newtheorem{prop}[definition]{Proposition}
\newtheorem{conjecture}{Conjecture}
\newtheorem{remark}{Remark}

\newtheorem{innercustomthm}{Proposition}
\newenvironment{customthm}[1]
  {\renewcommand\theinnercustomthm{#1}\innercustomthm}
  {\endinnercustomthm}

\usepackage[vmargin=0.2cm,hmargin=1cm,head=16pt,includeheadfoot]{geometry}
\usepackage{fancyhdr}
\title{From Lie algebra crossed modules to tensor hierarchies}

\author[a,b]{Sylvain Lavau\thanks{Corresponding author: lavau@math.univ-lyon1.fr}}
\author[c]{Jim Stasheff}
\affil[a]{\small \emph{Euler Institute (EIMI) \& Steklov Institute, Fontanka River Embankment 27,  191023,
St-Petersburg, Russia.}} 
\affil[b]{\small \emph{Aristotle University of Thessaloniki, Department of Mathematics, 541 24, Thessaloniki, Greece.}} %\emph{Euler Institute (EIMI) \& Steklov Institute, Fontanka, 27, 191023 St Petersburg, Russian Federation.}}
\affil[c]{\small \emph{109 Holly Dr,
Lansdale, PA 19446, USA.}}

\def\gg{\mathfrak g}

\def\be{\begin{equation}}

\def\ee{\end{equation}}

\def\ddo{\end{document}}

\date{}

\makeatletter
\def\blfootnote{\xdef\@thefnmark{}\@footnotetext}
\makeatother

\begin{document}

\maketitle

%\trd{I have inserted alternative treatments of some subjects, just for consideration and to see what their influence might be further on; ditto for the Pirashvili quotes. In both cases, the scars of the surgery need to be smoothed.}

\begin{abstract}

%The tradition of gauge theories based on Lie algebras may not be sufficient. Gauging procedures in supergravity involve a gauge algebra that is a subalgebra of the Lie algebra of symmetries. In such a cause, the field strength associated to the
%gauge fields might not transform covariantly. Then one is compelled to add a two-form field that compensates for
%the discrepancy, but then the cycle repeats, leading to a hierarchy/tower of tensor fields; a good example is found in supergravity theories.

% On the other hand, of increasing use are $L_\infty$ algebras. More recently, other  natural candidates for
%generalizing Yang-Mills theory are Leibniz algebras.
%Choosing a Leibniz algebra as gauge algebra leads to a Tensor
%Hierarchy, which in turn leads to an  $L_\infty$ algebra. Thus
%$L_\infty$-algebras form a natural, if not necessary, feature of Leibniz
%gauge theories.

 %Leibniz algebras have been increasingly used in gauging procedures in supergravity. Their relationships with $L_\infty$-algebras and tensor hierarchies have been explored \trh{have begun to be explored} in \cite{Bonezzi:2019ygf, LavauPalmkvist}. 
%Inspired by the results presented in [Lav19;LP19], 

The present paper, though inspired by the use of tensor hierarchies in theoretical physics, establishes their mathematical credentials, especially as genetically related to Lie algebra crossed modules. Gauging procedures in supergravity rely on a pairing -- the \emph{embedding tensor} -- between a Leibniz algebra and a Lie algebra. Two such algebras, together with their embedding tensor, form a triple called a \emph{Lie-Leibniz triple}, of which Lie algebra crossed modules are particular cases.
This paper is devoted to showing that any Lie-Leibniz triple induces a differential graded Lie algebra -- its associated \emph{tensor hierarchy} -- whose restriction to the category of Lie algebra crossed modules is the canonical assignment associating to any Lie algebra crossed module its corresponding unique 2-term differential graded Lie algebra. This shows that Lie-Leibniz triples form natural generalizations of Lie algebra crossed modules and that their associated tensor hierarchies can be considered as some kind of `lie-ization' of the former. %We show that this assignment is functorial, when restricted to a particular subcategory of that of Lie-Leibniz triples.
We deem the present construction of such tensor hierarchies clearer and more straightforward than previous derivations. 
%should use 
%, though many of them do. However, w
We stress that  such a construction suggests
%exists, hence justifying that there is room 
the existence of further  well-defined Leibniz gauge theories.
\end{abstract}

\noindent \textbf{Keywords:} Differential graded Lie (super)-algebra,  Leibniz algebra, Lie algebra crossed module.\newline
\noindent 2020 Mathematics Subject Classification: 17B70, 17A32, 18G45.
\tableofcontents

\section{Introduction}

%\trh{Remark for later : I define a $L_\infty$-algebra extension of a Leibniz algebra $V$ as a non-positively graded $L_\infty$-algebra $L=L_0\oplus L_{-1}\oplus\ldots$ such that \begin{itemize}
%    \item $L_0=V$,
%\item $l_2|_{L_0\wedge L_0}=[\,.\,,.\,]$, the skew-symmetric part of the Leibniz product, and
%\item $L_{-i}=0\quad \forall\, i\geq1$ whenever $V$ is a Lie algebra.
%\end{itemize} }
%\p
% \Loo-structures have gradually become a part of the `standard tool kit' in gauge field theory \cite{jds:moskva, Sati:2009ic}. At first,  most results involved technical details of special cases in string field theory and conformal field theory.  %\cite{z:csft}. 
% In the past few years, these higher structures emerged in  gauging procedures in supergravity, bringing  much interest to this topic from physicists. %\cite{Lavau:2014iva, Palmer:2013pka, Hohm:2017pnh, Cederwall:2018aab, Cagnacci:2018buk}. 

 Gauged models of supergravity involve a Lie algebra of symmetries $\mathfrak{g}$, of which only a sub-algebra $\mathfrak{h}$ is promoted to become a gauge algebra. To preserve the symmetries manifest when performing their computations, physicists stick to a formulation of the theory that involves the bigger Lie algebra $\mathfrak{g}$. One of the consequences is that the gauge theory needs the addition of several fields of higher form degree--and their associated field strengths, leading to what physicists call a \emph{tensor hierarchy} \cite{deWit:2008ta, dewit-sam, deWit-sam:endofhierarchy, Trigiante:2016mnt}. Mathematically, it is a (possibly infinite) tower of $\mathfrak{g}$-modules -- that by convention physicists take to be positively graded but that, in the present paper, we take to be negatively graded  %The justification of this choice is that when one goes from tensor hierarchies to $L_\infty$-algebras, the degree of the $k$-brackets is indeed $2-k$ (and not $k-2$, as it would be if the algebra had been positively graded).} 
 -- such that the space at degree $k>1$ depends exclusively on the spaces of lower degree. These spaces contain the redundant information carried by the fields of higher degree; they are necessary to preserve the covariance of the theory. The need for a unified perspective on gauging procedures in supergravity, as well as in Double and Exceptional field theories, has been salient in theoretical physics for some years now \cite{Bonezzi:2019bek,  Cagnacci:2018buk, Cederwall:2017fjm, Hohm:2019wql}.

On the mathematical side, another approach to tensor hierarchies has been recently explored: that of \emph{Leibniz algebras} \cite{Bonezzi:2019ygf, Kotov2019, lavau:TH-Leibniz}. A Leibniz algebra is a vector space $V$ equipped with a bilinear product $\circ$ which is a derivation of itself, i.e. it satisfies the Leibniz rule, hence the denomination %\footnote{They are also called \emph{Loday algebras}, from Jean-Louis Loday (1946-2012), who first identified them, cf the Appendix.} 
\cite{blokhGeneralizationConceptLie1965, loday-pirashvili}:
\begin{equation}
    x\circ(y\circ z)=(x\circ y)\circ z+y\circ(x\circ z)
\end{equation}
%In particular, Lie algebras are Leibniz algebras where the Leibniz product is fully skew-symmetric. 
The construction of tensor hierarchies relies on the existence of \emph{embedding tensors}.
If $\mathfrak{g}$ is a Lie algebra and if a Leibniz algebra $V$ is a $\mathfrak{g}$-module, an embedding tensor is a map $\Theta:V\to \mathfrak{g}$ satisfying some consistency conditions -- called the linear and quadratic constraints, respectively. The quadratic constraint implies that $\mathfrak{h}:=\mathrm{Im}(\Theta)$ is a sub-Lie algebra of $\mathfrak{g}$, whereas what we call the linear constraint %\footnote{Notice that what we call `linear constraint' actually slightly differs from what physicists call `linear constraint' in this context (see footnote \ref{footnoteconstraint}).} 
requires that the Leibniz product on $V$ is generated by the action of $\mathfrak{h}$ on $V$:
\begin{equation}\label{latterconstraint}
x\circ y = \Theta(x)\cdot y 
\end{equation}
In other words, the embedding tensor lifts the adjoint action to $\mathfrak{g}$:
\begin{center}
\begin{tikzpicture}
\matrix(a)[matrix of math nodes, 
row sep=4em, column sep=6em, 
text height=1.5ex, text depth=0.25ex] 
{&\mathfrak{g}\\
V&\mathrm{End}(V)\\}; 
\path[->](a-1-2) edge node[right]{$\rho$} (a-2-2); 
\path[->](a-2-1) edge node[above]{$\mathrm{ad}$} (a-2-2);
\path[->](a-2-1) edge node[above left]{$\Theta$} (a-1-2);
\end{tikzpicture}
\end{center}
where $\rho$ represents the action of $\mathfrak{g}$ on $V$. We call such compatible data $(\mathfrak{g},V,\Theta)$ \emph{Lie-Leibniz triples}.
%Some authors do not put much emphasis on this latter constraint, or consider that the Leibniz algebra structure is a mere by-product of the definition of $\Theta$. On the contrary, we believe that 
Equation~\eqref{latterconstraint} is an important aspect of these algebraic structures, because it draws a parallel with the notion of \emph{Lie algebra crossed module}. The latter, also called \emph{differential crossed module}  \cite{baez:Lie2alg}, appear to be precisely those Lie-Leibniz triples such that  $V$ is a Lie algebra and $\Theta$ is $\mathfrak{g}$-equivariant. The introduction of the notion of Lie algebra crossed module was historically preceded by the study of crossed modules of discrete groups, see \cite{MR4322147} for an historical account.  %Notice that what we call `linear constraint' actually slightly differs from what physicists call `linear constraint' in this context (see footnote \ref{footnoteconstraint}).
%We call the compatible data $(\mathfrak{g},V,\Theta)$ \emph{Lie-Leibniz triples}, and we indeed notice that Lie algebra crossed modules are precisely those Lie-Leibniz triples such that  $V$ is a Lie algebra and $\Theta$ is $\mathfrak{g}$-equivariant.
%of which Lie algebra crossed modules form particular cases.

It has been proved that any Lie-Leibniz triple $(\mathfrak{g},V,\Theta)$ gives rise canonically to a non-positively graded differential graded Lie algebra \cite{lavau:TH-Leibniz}, which might coincide with the one defined from Borcherds algebras in \cite{Palmkvist-THA}. 
Unfortunately, the construction presented in \cite{lavau:TH-Leibniz} is not very satisfying, mathematically speaking, because the role of $\mathfrak{g}$ -- although central in the Lie-Leibniz triple -- is not salient in the resulting differential graded Lie algebra. In the present paper we provide an alternative, simpler construction, that generalizes the results in~\cite{lavau:TH-Leibniz} and that additionally reinstates the prominent role of $\mathfrak{g}$. The result is a differential graded Lie algebra $\mathbb{T}=\bigoplus_{i=-1}^{\infty}T_{-i}$ concentrated in degrees lesser than or equal to $+1$:
\begin{center}
\begin{tikzcd}[column sep=1cm,row sep=0.4cm]
\ldots\ar[r]&T_{-3}\ar[r, "\partial_{-2}"]&T_{-2}\ar[r, "\partial_{-1}"]&V[1]\ar[r,"\partial_{0}=\Theta"]&\mathfrak{g}\ar[r, "\partial_{+1}"]& R_\Theta[-1]
\end{tikzcd}
\end{center}
where $R_\Theta[-1]$ is the cyclic $\mathfrak{g}$-submodule of $\mathrm{Hom}(V,\mathfrak{g})$ generated by $\Theta$, the degree having been shifted by $+1$.
We call this differential graded Lie algebra $\mathbb{T}$ the \emph{tensor hierarchy associated to $(\mathfrak{g},V,\Theta)$}. The negatively graded part $T_\bullet=\bigoplus_{i=1}^{\infty}T_{-i}$ is obtained as a quotient of the free graded Lie algebra of $V[1]$ by a particular graded ideal, as was proposed by earlier work in theoretical physics \cite{Cederwall:2015oua, Gomis:2018xmo}. A priori, it is not an exact sequence, although we raise the question that under some circumstances it might be (see Conjecture~\ref{conjecturr}).

Although it is injective on objects, the function %$G:\textbf{Lie-Leib}\to\textbf{DGLie}_{\leq1}$ 
that assigns to each Lie-Leibniz triple its associated tensor hierarchy is not functorial on the category \textbf{Lie-Leib} of Lie-Leibniz triples. However, we note that its restriction to the full subcategory \textbf{Lie$\times$Mod} of Lie algebra crossed modules
has a very particular image. More precisely, 
  any Lie algebra crossed module $\mathfrak{c}\overset{\Theta}{\longrightarrow}\mathfrak{g}$ is sent by this function to its corresponding canonical 2-term differential graded Lie algebra \cite{baez:Lie2alg} or, more precisely, to the following 3-term differential graded Lie algebra:
\begin{equation}\label{materialization}
\begin{tikzcd}[column sep=1cm,row sep=0.4cm]
\mathfrak{c}\ar[r,"\Theta"]&\mathfrak{g}\ar[r,"0"]&\mathbb{R}[-1]
\end{tikzcd}
\end{equation}
Here, $\mathbb{R}[-1]$ is the irreducible trivial representation of $\mathfrak{g}$, to which  the $\mathfrak{g}$-submodule $R_\Theta\subset\mathrm{Hom}(\mathfrak{c},\mathfrak{g})$ is isomorphic, since $\Theta$ is $\mathfrak{g}$-equivariant by definition of Lie algebra crossed modules. We observe that this assignment defines a faithful functor on \textbf{Lie$\times$Mod}. It is thus reasonable to look for the biggest subcategory of \textbf{Lie-Leib} on which the restriction of the function assigning to each Lie-Leibniz triple its associated tensor hierarchy is functorial.

By construction, the tensor hierarchy associated to a Lie-Leibniz triple encodes the symmetric part of the Leibniz product:
\begin{equation}\label{braque}
    \{x,y\}=\frac{1}{2}\big(x\circ y +y\circ x\big)
\end{equation}
This bracket vanishes when $V$ is a Lie algebra, so in particular when $(\mathfrak{g},V,\Theta)$ is a Lie algebra crossed module, resulting in the cochain complex \eqref{materialization}. The prominence of the symmetric bracket \eqref{braque} and the $\mathfrak{g}$-action in the construction of the tensor hierarchy associated to a given Lie-Leibniz triple $(\mathfrak{g},V,\Theta)$ can be seen more clearly in the fact that the first space of the hierarchy, namely $T_{-2}$, is the quotient of $S^2(V)$ by the biggest $\mathfrak{g}$-submodule of $\mathrm{Ker}\big(\{\,.\,,.\,\}\big)\subset S^2(V)$, denoted $K$.
We say that a morphism of Lie-Leibniz triples $(\varphi, \chi):(\mathfrak{g},V,\Theta)\to(\mathfrak{g}',V',\Theta')$ is \emph{compatible} if it sends $K\subset \mathrm{Ker}\big(\{\,.\,,.\,\}\big)\subset S^2(V)$ to $K'\subset \mathrm{Ker}\big(\{\,.\,,.\,\}'\big)\subset S^2(V')$. 
Then, Lie-Leibniz triples together with compatible morphisms form a wide subcategory \textbf{compLie-Leib} of the category \textbf{Lie-Leib}, and \textbf{Lie$\times$Mod} is a full subcategory of both: \textbf{Lie$\times$Mod}$\ \subset\ $\textbf{compLie-Leib}$\ \subset\ $ \textbf{Lie-Leib}. %{lodayexample2}
%(which is $\mathrm{Rep}(\Theta)$ because $\Theta$ is $\mathfrak{g}$-equivariant). 
%In order to refine our analysis, we rely on the importance of the relationship between the embedding tensor and  the adjoint action of $V$ on itself. We call \emph{semi-strict} (resp. \emph{strict}) Lie-Leibniz triples those Lie-Leibniz triples in which the adjoint map (resp. the embedding tensor) is $\mathfrak{g}$-equivariant. Then we have a sequence of inclusions of full subcategories:
%\begin{equation*}
 %   \textbf{Lie$\times$Mod}\subset\textbf{strLie-Leib}\subset\textbf{semLie-Leib}\subset\textbf{compLie-Leib}\subset\textbf{Lie-Leib}
%\end{equation*}
%where \textbf{semLie-Leib} (resp. \textbf{strLie-Leib}) is the category of semi-strict (resp. strict) Lie-Leibniz triples, and \textbf{compLie-Leib} denotes that of stringent ones.
We now state the main result of this paper\footnote{After several conversations with Jakob Palmkvist at MITP (Mainz) in January 2022, it appeared that Theorem \ref{centraltheorem} was erroneous, as $\textbf{DGLie}_{\leq1}$ should be $\textbf{DGLie}$. We corrected this mistake and provided more details about the reasons why this is so in a Corrigendum \cite{lavauCorrigendumLieAlgebra2023}, which can be found as Appendix \ref{corrige} of this preprint.\label{footnote1}}: %, which can be found as Theorem \ref{centraltheorem}:
\begin{theorem}\label{centraltheorem}
The functor
\begin{align*}
\overline{G}:\hspace{0.2cm}\emph{\textbf{Lie$\times$Mod}}&\xrightarrow{\hspace*{1.2cm}} \hspace{0.4cm}\emph{\textbf{DGLie}}_{\leq1}\\
	\big(\mathfrak{c}\xrightarrow{\Theta}\mathfrak{g}\big)&\xmapsto{\hspace*{1.2cm}} \big(\mathfrak{c}\xrightarrow{\Theta}\mathfrak{g}\xrightarrow{0}\mathbb{R}[-1]\big)
\end{align*}
can be canonically extended to an injective-on-objects function $G:\emph{\textbf{Lie-Leib}}\to\emph{\textbf{DGLie}}_{\leq1}$ such that the restriction of this function to the wide sub-category \textbf{\emph{compLie-Leib}} is a faithful functor. Moreover, \textbf{\emph{compLie-Leib}} is the biggest wide subcategory of \textbf{\emph{Lie-Leib}} (with respect to inclusion) such that this functorial property holds. % which coincides with $\overline{G}$ on \emph{\textbf{Lie$\times$Mod}}.
\end{theorem}

This theorem shows that, while Lie-Leibniz triples form a natural generalization of  Lie algebra crossed modules,  tensor hierarchies  materialize as a generalization of the differential graded Lie algebra formulation of such Lie algebra crossed modules, as given in the cochain complex \eqref{materialization}. The paper is divided into two sections:
Section \ref{section2} presents the main mathematical concepts and ideas,   while Section \ref{section3} is mostly technical as it contains the details of the proof of Theorem \ref{centraltheorem}. To begin with, Lie-Leibniz triples are introduced and studied in subsection~\ref{subsecc}. Then, the relationship between Lie algebra crossed modules and Lie-Leibniz triples is established in subsection \ref{relatio}, which concludes with an explanation of Theorem \ref{centraltheorem}. Section~\ref{section3} 
is subdivided into four parts: building the graded vector space $T_\bullet$ and equipping it with a graded Lie bracket in subsection~\ref{construction}, showing that it moreover admits a non-trivial differential in subsection~\ref{structure}, proving that the differential and the graded Lie bracket are compatible in subsection
\ref{patching}, and finally proving Theorem \ref{centraltheorem} in the final subsection~\ref{restriction}.
Though the construction(s) and proofs are technical, they proceed by a sequence of steps, each of which is meaningful and  fairly transparent. In order for the reader to fully appreciate the first steps of the construction of the tensor hierarchy, we have introduced a detailed example at the end of subsection~\ref{patching} (see Example~\ref{lodayexample2}). The paper concludes on a discussion of the algebraic structures involved and open questions. 
Clarifying the relationship between Lie algebra crossed modules and Lie-Leibniz triples opens the possibility to better understand higher gauge theories. Indeed, the content of this paper suggests that a unified perspective on standard gauge theories and Leibniz gauge theories \cite{Bonezzi:2019ygf, Hohm:2019wql, Strobl:2019mgf} can emerge through the use of Lie-Leibniz triples and tensor hierarchies. %We leave this exciting topic to further investigations. %Applications to Double and Exceptional field theories may certainly be

\section{Lie-Leibniz triples and Lie algebra crossed modules}\label{section2}
%\subsection{Leibniz algebras and Lie-Leibniz triples}

\subsection{Leibniz algebras and Lie-Leibniz triples}\label{subsecc}

The beginning of the section involves well-known facts about Leibniz algebras \cite{blokhGeneralizationConceptLie1965, loday-pirashvili} but then we introduce novel notions that will be thoroughly used in the paper. The field of scalars will be $\mathbb{R}$ throughout the paper.
%and $L_\infty$-algebras \cite{lada-markl}.
\begin{definition}\label{def:leibniz}
A  (left) \emph{Leibniz algebra} is a vector space $V$ together with a bilinear operation $\circ:V\otimes V\to V$ satisfying the relation:
\begin{equation}\label{leibnizproduct}
    x\circ (y\circ z) = (x\circ y)\circ z + y\circ(x\circ z)
    \end{equation}
     A \emph{Leibniz subalgebra} of $V$ is a subspace $U$ which is stable under the Leibniz product.  
    
A morphism of Leibniz algebras, or \emph{Leibniz algebra morphism}, between $(V,\circ)$ and $(V',\circ')$ is a linear map $\chi:V\to V'$ that commutes with the respective Leibniz products:
\begin{equation}
\chi(x\circ y)=\chi(x)\circ'\chi(y)
\end{equation}
\noindent %The space of Leibniz algebra (endo)morphisms from $V$ to itself is noted $\mathrm{End}(V)$. 
We call \emph{\textbf{Leib}} the category of Leibniz algebras, together with Leibniz algebra morphisms.
\end{definition}

%\begin{remark}
%Given a Leibniz algebra $V$, the vector space of Leibniz algebra morphisms from $V$ to itself should not be confused with the space $\mathrm{End}(V)$ of vector spaces endomorphisms.
%\end{remark}

Let $(V,\circ)$ be a Leibniz algebra. The Leibniz product $\circ$ can be split into  its skew-symmetric part, denoted $[\,.\, ,.\, ]:\wedge^2V\to V$, and its symmetric part, denoted 
 $\{\,.\, ,.\, \}: S^2(V) \to V$, defined by:
%We can  split the product $\circ$ of a Leibniz algebra $V$ into its symmetric part $\{.\,,.\}$ and its skew-symmetric part $[\,.\,,.\,]$:

\begin{equation*}
[x,y]:=\frac{1}{2}\big(x\circ  y-y\circ  x\big)%\label{skewsym}
\hspace{0.7cm}\text{and}\hspace{0.7cm}
\{x,y\}:=\frac{1}{2}\big(x\circ  y+y\circ  x\big)%\label{sym}
\end{equation*}
for any $x,y\in V$, so that:
\begin{equation}\label{splitting}
x\circ  y=[x,y]+\{x,y\}
\end{equation}

\begin{example}
A Lie algebra is a Leibniz algebra whose product is skew-symmetric, i.e. such that $\{\,.\,,.\,\}=0$. The category \emph{\textbf{Lie}} of Lie algebras with Lie algebra morphisms is a full subcategory of \emph{\textbf{Leib}}.  %The Leibniz morphisms are then Lie algebra morphisms.
\end{example}

\begin{definition}\label{definitiontoutou}
A \emph{left ideal of $V$} (or \emph{ideal} for short) is a subspace $\mathcal{K}$ of $V$ that satisfies the following condition:
\begin{equation*}
    V\circ \mathcal{K}\subset\mathcal{K}
\end{equation*}
A left ideal $\mathcal{K}$ of $V$ whose %action 
left $\circ$-product with $V$ is trivial, i.e. such that $\mathcal{K}\circ V=0$, is called \emph{central}.
\end{definition}
\begin{example}
The sub-space $\mathcal{I}$ of $V$ generated by elements of the form $\{x,x\}$ is an ideal called the \emph{ideal of squares of V}. A direct application of Equation \eqref{leibnizproduct} shows that it is central. The set $\mathcal{Z}$ of all \emph{central} elements of $V$, defined by:
\begin{equation*}
\mathcal{Z}=\Big\{x\in V\ \big|\ x\circ  y=0\ \ \text{for all}\ \ y\in V\Big\}
\end{equation*}
is a central ideal of $V$, and is the biggest such. It is called the \emph{centreof $V$}.
\end{example}
%By Equation \eqref{eq:deriv}, the centreis stable by every derivations of $V$.

In a general Leibniz algebra, the symmetric part of the bracket is not associative nor does the skew-symmetric part of the 
bracket  satisfy the Jacobi identity (otherwise it would be a Lie bracket), but  using Equation \eqref{leibnizproduct}, we have:
\begin{equation}\label{jacobiator0}
\big[x,[y,z]\big]+\big[y,[z,x]\big]+\big[z,[x,y]\big]=-\frac{1}{3}\Big(\big\{x,[y,z]\big\}+\big\{y,[z,x]\big\}+\big\{z,[x,y]\big\}\Big)
\end{equation}
for every $x,y,z\in V$. 
Since the right hand side lies in the ideal of squares $\mathcal{I}$, we deduce that for any central ideal $\mathcal{K}$ satisfying the following inclusions:
\begin{equation*}
    \mathcal{I}\subset \mathcal{K}\subset\mathcal{Z},
\end{equation*}
the Leibniz product canonically induces a Lie algebra structure on the quotient $\bigslant{V}{\mathcal{K}}$. The Lie algebra $\bigslant{V}{\mathcal{I}}$ is the biggest such, whereas $\bigslant{V}{\mathcal{Z}}$ is the smallest. The quotient map  $V\to\bigslant{V}{\mathcal{I}}$ is universal with respect to morphisms of Leibniz algebras from $V$ to any Lie algebra \cite{loday-pirashvili}.

A (symmetric) representation of a (left) Leibniz algebra is the natural generalization of the notion of Lie algebra representation: %\footnote{The correct notion of representation in the category of Leibniz algebras is  the data of a left and a right actions that satisfy some natural consistency conditions \cite{loday-pirashvili}. It is only when the right action equals minus the left action that one says that the representation is symmetric.}: 

\begin{definition} A (symmetric) representation of a Leibniz algebra $V$ is the data of a vector space $E$ and a morphism of Leibniz algebras $\rho:V\to\mathrm{End}(E)$.
\end{definition}
 In particular, since $\mathrm{End}(E)$ is a Lie algebra, the linear map $\rho$ satisfies the following identity:
\begin{equation}
    \rho(x\circ y)=\big[\rho(x),\rho(y)\big]
\end{equation}
for every $x,y\in V$.
Since the right-hand side is fully skew-symmetric, so should be the left-hand side. Then, the ideal of squares $\mathcal{I}$ is necessarily in the kernel of the map $\rho$. This is true for every symmetric
representation of $V$. In particular, this implies that such representations correspond to representations of the Lie algebra $\bigslant{V}{\mathcal{I}}$ \cite{loday-pirashvili}.

 As an example, one may consider the adjoint action of $V$ on itself.
 Indeed, the Leibniz algebra $V$ is a representation of $V$ through the \emph{adjoint map} -- denoted $\mathrm{ad}$ -- and defined by:
\begin{align*}
\mathrm{ad}:\hspace{0.2cm}V&\xrightarrow{\hspace*{1.2cm}} \hspace{0.4cm}\mathrm{End}(V)\\
	x&\xmapsto{\hspace*{1.2cm}}\mathrm{ad}_x:y\longmapsto x\circ y
\end{align*}
One can rewrite the Leibniz identity \eqref{leibnizproduct} as:
\begin{equation}\label{adleibniz}
     \mathrm{ad}_{\mathrm{ad}_x(y)}(z)=\mathrm{ad}_x\big(\mathrm{ad}_y(z)\big)-  \mathrm{ad}_y\big(\mathrm{ad}_x(z)\big)
\end{equation}
which, in turn, can be compactly summarized as:
\begin{equation}\label{adleibniz2}
 \mathrm{ad}_{\mathrm{ad}_x (y)}=\big[\mathrm{ad}_x,\mathrm{ad}_y\big]
 \end{equation}
 where the bracket on the right hand side is the Lie bracket on $\mathrm{End}(V)$. Since $\mathrm{ad}_x(y)=x\circ y$ on the left hand side, this equation proves that the adjoint action indeed induces a representation of $V$ on itself. 
 
 Moreover, the Leibniz identity \eqref{leibnizproduct} shows that the map $\mathrm{ad}$ lands in the derivations of $V$.
%\begin{definition}
A \emph{derivation} of a Leibniz algebra $(V,\circ)$ is a linear map $\delta:V\to V$ such that:
\begin{equation}
    \delta(x\circ y)=\delta(x)\circ y+x\circ \delta(y)\label{eq:deriv}
\end{equation}
%\end{definition}
The space of %\emph{
derivations of $V$, denoted  $\mathrm{Der}(V)$, is a Lie subalgebra of $\mathrm{End}(V)$.
%whose elements satisfy the following condition:
%\begin{equation}
 %   f(x\circ y)=f(x)\circ y+x\circ f(y)\label{eq:deriv}
%\end{equation}
%for every $f\in\mathrm{Der}(V)$ and $x,y\in V$. 
If $\delta$ is a derivation of $V$, then it is a derivation of the symmetric and the skew-symmetric parts of the Leibniz product. These results imply that both the ideal of squares $\mathcal{I}$ and the centre$\mathcal{Z}$ are stable by every derivation of $V$. %As in the Lie algebra case, one can define an adjoint action on V.
% is the set of derivations of the product $\circ$. 
The image of the adjoint action $\mathrm{ad}:V\to\mathrm{Der}(V)$ is a Lie subalgebra of $\mathrm{Der}(V)$ called the \emph{inner derivations of $V$}, and denoted $\mathfrak{inn}(V)$. The kernel of this map is the centre$\mathcal{Z}$ of the Leibniz algebra. Hence, the adjoint map  defines a Lie algebra isomorphism between $\bigslant{V}{\mathcal{Z}}$ and $\mathfrak{inn}(V)$.

%Since the ideal of squares is in the kernel of the linear map $\mathrm{ad}$, then the leftt hand side equates $\mathrm{ad}_{[x,y]}$, so this formula says that $\mathrm{ad}$ induces a representation of the Lie algebra $\bigslant{V}{\mathcal{I}}$ on $V$, which is equivalent to a symmetric representation of $V$ on itself.

We are now interested in the relationship between the adjoint action of a Leibniz algebra and the embedding tensor of gauging procedures in supergravity. In these theories, the (super)-symmetries are controlled by a Lie algebra $\mathfrak{g}$ and the field content of the physical model lies in various $\mathfrak{g}$-modules \cite{deWit:2008ta, deWit-sam:endofhierarchy}. It turns out that the gauge fields $A_\mu$ do not necessarily take values in the adjoint representation of $\mathfrak{g}$. %the fundamental representation of $\mathfrak{g}$ %\footnote{The fundamental representation of a Lie algebra is defined as its smallest-dimensional faithful representation.}. 
 This is rather unusual, compared to classical Yang-Mills gauge theories where such fields take values in the adjoint representation of $\mathfrak{g}$.
Since gauge transformations are local symmetries generated by the action of gauge fields on the Lagrangian, physicists have developed a natural set-up to couple the representation in which such fields take values with the Lie algebra $\mathfrak{g}$ of global symmetries, see e.g. \cite{Hohm:2019wql}.
%end{remark}

More precisely, the basic data consist of a Lie algebra $\mathfrak{g}$  and a $\mathfrak{g}$-module $V$, % (usually the defining representation of $\mathfrak{g}$), 
in which gauge fields take values. Then physicists observe that under some conditions the Lie algebra structure on $\mathfrak{g}$ can be lifted to a Leibniz algebra structure on $V$ \cite{Hohm:2019wql}.  Although the elements of  $V$ would play the role of gauge fields,  % under some sort of yet not well understood "Leibniz gauge theory". 
their associated field strength -- as given by the usual definition $F=dA + A\wedge A$ -- are not covariant under gauge transformations, due precisely to the absence of skew-symmetry of the Leibniz algebra structure.
This observation forces physicists to add higher fields to the model, taking values in additional $\mathfrak{g}$-modules, hence ending up with a tower of spaces that physicists call \emph{tensor hierarchy} \cite{deWit:2008ta, deWit-sam:endofhierarchy}. This is a new manifestation of higher gauge theories in theoretical physics.

Physicists lift the Lie algebra structure of $\mathfrak{g}$ to a Leibniz algebra structure on $V$ by using what they call an \emph{embedding tensor}. It is a linear mapping $\Theta:V\to \mathfrak{g}$ satisfying some conditions which, if satisfied, ensure the compatibility between the Lie algebra structure on $\mathfrak{g}$ and the Leibniz algebra structure on $V$ in a sense that resembles what occurs in a Lie algebra crossed module.
Due to their importance in supergravity theories and in the present paper, triples of elements $(\mathfrak{g},V,\Theta)$ satisfying these conditions deserve their own name:

\begin{definition}\label{def:lieleibniz}
A \emph{Lie-Leibniz triple} is a triple $(\mathfrak{g},V,\Theta)$ where:
\begin{itemize}
\item $\mathfrak{g}$ is a Lie algebra,
\item $V$ is a $\mathfrak{g}$-module equipped with a Leibniz algebra structure $\circ:V\otimes V\to V$,
\item $\Theta: V\to \mathfrak{g}$ is a linear mapping called the \emph{embedding tensor}, that satisfies two compatibility conditions. 
\end{itemize} 
 The first one is the \emph{linear constraint}:
\begin{equation}\label{eq:compat}
x\circ y = \Theta(x)\cdot y  
\end{equation}
%where $\rho$ denotes the action of $\mathfrak{g}$ on $V$.
The second one is called the \emph{quadratic constraint}:
\begin{equation}\label{eq:equiv}
\Theta(x\circ  y)=[\Theta(x),\Theta(y)]_\mathfrak{g}
\end{equation}
where $[\,.\,,.\,]_\mathfrak{g}$ is the Lie bracket on $\mathfrak{g}$.

A \emph{sub-Lie-Leibniz triple} of $(\mathfrak{g},V,\Theta)$ is a Lie-Leibniz triple $(\mathfrak{t}, U,\Theta|_U)$ such that $\mathfrak{t}$ (resp. U) is a Lie subalgebra of $\mathfrak{g}$ (resp. Leibniz subalgebra of $V$).
%For a fixed Leibniz algebra $(V,\circ)$, every Lie-Leibniz triple $(\mathfrak{g},V,\Theta)$ whose Leibniz product generated by Equation \eqref{eq:compat} coincides with $\circ$, is said \emph{modeled over $V$}. %We say that the Lie-Leibniz triple is \emph{strict} if $\rho:\mathfrak{g}\to\mathrm{End}(V)$ takes values in $\mathrm{Der}(V)$.
%As a map $\Theta$
%If $\Theta:V\to \mathfrak{g}$ is surjective, we say that $(\mathfrak{g},V,\Theta)$ is a \emph{replete Lie-Leibniz triple}. 
\end{definition}
%A Lie-Leibniz triple is said \emph{faithful} when $V$ is a faithful representation of $\mathfrak{g}$. %A Lie-Leibniz triple is said \emph{irreducible} when $T_\Theta$ $-$ the representation to which $\Theta$ belongs $-$ is an irreducible representation of $\mathfrak{g}$.

%\begin{remark}
%In the present paper, to represent the action of an element  $a\in\mathfrak{g}$ on a vector $x$ of any $\mathfrak{g}$-module $R$, we use a map $\rho:\mathfrak{g}\to \mathrm{End}(R)$ or we use a dot, and we write:
%\begin{equation}
%a \cdot x=\rho(a;x)
%\end{equation}
%\end{remark}

\begin{remark}\label{rempalmo}
What physicists call \emph{linear constraint}, or \emph{representation constraint}, is a condition on $\Theta$ that singles out the $\mathfrak{g}$ sub-module of $\mathrm{Hom}(V,\mathfrak{g})$ to which $\Theta$ belongs. In the physics literature, the linear constraint sometimes appears under the form that the symmetric bracket $\{\,.\,,.\,\}$ factors through a $\mathfrak{g}$ sub-module of $S^2(V)$, see e.g. Equation (6) in \cite{Palmkvist2013}. This linear constraint can be interpreted as a weakening of one of the axioms of the notion of Lie algebra crossed module (Equation \eqref{loof1} in Lemma \ref{lemmacapcom}).
\end{remark}

The Leibniz algebra structure on $V$ is uniquely defined from the action of the Lie algebra $\mathfrak{g}$ on $V$ and Equation \eqref{eq:compat}. This linear constraint shows that the embedding tensor is a \emph{lift} of the adjoint action along the representation~$\rho$:\begin{center}
\begin{tikzpicture}
\matrix(a)[matrix of math nodes, 
row sep=5em, column sep=7em, 
text height=1.5ex, text depth=0.25ex] 
{&\mathfrak{g}\\
V&\mathrm{End}(V)\\}; 
\path[->](a-1-2) edge node[right]{$\rho$} (a-2-2); 
\path[->](a-2-1) edge node[above]{$\mathrm{ad}$} (a-2-2);
\path[->](a-2-1) edge node[above left]{$\Theta$} (a-1-2);
\end{tikzpicture}
\end{center}
It is thus redundant to assume in Definition \ref{def:lieleibniz} that $V$ is a Leibniz algebra beforehand. 
 %In practice, most Lie-Leibniz triples arise in this way.
 The same redundancy occurs in the definition of Lie algebra crossed modules. Moreover, Equations \eqref{eq:compat} and \eqref{eq:equiv} are similar to those defining Lie algebra crossed modules. We will investigate in subsection \ref{relatio} the extent of the precise relationship 
between the former and the latter notions.

%The notion of Lie-Leibniz triples originates in gauging procedures in supergravity.
In supergravity models, %the vector space $V$ happens to be a representation of the Lie algebra $\mathfrak{g}$ of symmetries of the system. More precisely, physicists consider that
 $\mathfrak{g}$ is often taken to be the  split real form of a semi-simple Lie algebra and $V$ is the smallest-dimensional faithful representation of $\mathfrak{g}$ \cite{Trigiante:2016mnt}. % (the so-called \emph{fundamental} representation). 
 % Physicists call $V$ the \emph{fundamental} -- or \emph{defining} -- representation of $\mathfrak{g}$ \cite{Trigiante:2016mnt}. 
  The term \emph{embedding tensor} has been introduced in the 2000s and refers to the fact that the image of the map $\Theta$ defines a Lie sub-algebra of $\mathfrak{g}$: this can be straightforwardly read from Equation \eqref{eq:equiv}. %In the situation where this map is onto $\mathfrak{inn}(V)$, one sees that the bilinear product $\circ$ can be derived from the map $\rho$, i.e. for every $x\in V$, there exists $a\in\mathfrak{g}$ such that $\rho(a;-)=\mathrm{ad}_x$, so that:
%\begin{equation}\label{brouillon}
 %   x\circ y=\rho(a;y)
%\end{equation}
%for every $y\in V$. The (non-unique) choice of an assignment of an element $a$ of $\mathfrak{g}$ for every element $x$ of $V$ such that $\rho(a;-)=\mathrm{ad}_x$ gives rise to a Leibniz algebra morphism $\Theta:V\to \mathfrak{g}$ that lifts the adjoint map $\mathrm{ad}:V\to\mathrm{End(V)}$. In other words, the following diagram is commutative in the category of Leibniz algebras:
%on $V$ induces a bilinear product $\bullet$ defining a Leibniz algebra structure on $V$, that may be different than the original one:
%Triples $(\mathfrak{g},V,\Theta)$ making such diagrams commutative  are objects of great interest since they allow more flexibility than a mere Leibniz algebra, by introducing another Lie algebra $\mathfrak{g}$ which is represented in $\mathrm{End}(V)$. 
 %The quadratic constraint, Equation \eqref{eq:equiv}, moreover implies that
%$\mathrm{Im}(\Theta)$ is a Lie subalgebra of \gg, hence the term `embedding'. 
This Lie sub-algebra $\mathrm{Im}(\Theta)\subset\mathfrak{g}$  plays the role of the gauge algebra in supergravity theories. For the particular role that this Lie subalgebra plays in this paper, % that the Lie subalgebra  $\mathrm{Im}(\Theta)\subset\mathfrak{g}$ plays in this paper and in supergravity theories,
 we denote it by $\mathfrak{h}$. 
Throughout this section and the next one, major emphasis will be on the respective roles of 
$\gg$\  and of $\mathfrak{h}$ in a general Lie-Leibniz triple, as well as on the relationship between $\mathrm{ad}$ and $\Theta$. %We will then focus on more specialized Lie-Leibniz triples in Section \ref{section6}, to define the desired functor $ G $, and thus $G$. 
The quadratic constraint~\eqref{eq:equiv} additionally implies two things: that $\mathrm{Ker}(\Theta)$ is an ideal of $V$ and that $\mathcal{I}\subset \mathrm{Ker}(\Theta)$. The linear constraint then implies that this ideal is central. Thus we have the following successive inclusions:
\begin{equation}\label{inclusion}
    \mathcal{I}\subset \mathrm{Ker}(\Theta)\subset \mathcal{Z}
\end{equation}
The embedding tensor $\Theta$ then induces a Lie algebra isomorphism between $\bigslant{V}{\mathrm{Ker}(\Theta)}$ and $\mathfrak{h}=\mathrm{Im}(\Theta)$.
The equality $\mathrm{Ker}(\Theta)=\mathcal{Z}$ is satisfied only when $V$ is a faithful representation of $\mathfrak{g}$. %(which is usually the case in gauging procedures in supergravity).

%\begin{remark}
%In gauging procedures in supergravity, physicists define the \emph{linear constraint} (also called the \emph{representation constraint}) as the symmetrization of Equation \eqref{eq:compat}, i.e. they only require the following identity to hold:
%\begin{equation}
 %   \{x,y\}=\frac{1}{2}\big(\rho(\Theta(x);y)+\rho(\Theta(y);x)\big)
%\end{equation}
%This emphasizes the role that physicists give to the embedding tensor: it contains information on the ideal of squares $\mathcal{I}$, i.e. the part of the Leibniz algebra that does not appear in classical Yang-Mills gauge theory. The Lie algebra structure on $\mathrm{Im}(\Theta)$ defined by the quadratic constraint usually contains all the physical relevant information.
%\end{remark}

\begin{example}
Let $(V,\circ)$ be a Leibniz algebra and take $\mathfrak{g}=\mathrm{End}(V)$.
Let $\Theta=\mathrm{ad}$ be the map that sends any element $x\in V$ to the adjoint endomorphism $\mathrm{ad}_x:y\mapsto x\circ y$. The gauge algebra $\mathfrak{h} = \mathrm{Im}(\Theta)$ consists of the inner derivations of $V$. %, denoted $\mathfrak{inn}(V)$.
Then $\big(\mathrm{End}(V),V,\mathrm{ad}\big)$ is a Lie-Leibniz triple.
\end{example}

Another way of understanding the quadratic constraint  is to note that it implies that $\Theta$ is $\mathrm{Im}(\Theta)$-equivariant (but not necessarily $\mathfrak{g}$-equivariant). 
%More precisely, as a linear map from $V$ to $\mathfrak{g}$, the embedding tensor can be seen as an element of a submodule -- say $ R_\Theta$ -- of the $\mathfrak{g}$-module $V^*\otimes\mathfrak{g}$, that is acted upon through the action of $\mathfrak{g}$ on $V^*$ induced by $\rho$ on the one hand, and through the adjoint action of $\mathfrak{g}$ on itself on the other hand. This defines a map:
%\begin{align*}
%\eta:\hspace{0.2cm}\mathfrak{g}&\xrightarrow{\hspace*{1.2cm}} \hspace{0.4cm}\mathrm{End}(R_\Theta)\\
%	a&\xmapsto{\hspace*{1.2cm}}\eta(a;-):\Xi\longmapsto \big(x\mapsto[a,\Xi(x)]_{\mathfrak{g}}-\Xi(\rho(a;x))\big)
%\end{align*}
%For any Lie subalgebra $\mathfrak{t}$ of the Lie algebra $\mathfrak{g}$, we say that $\Theta$ is \emph{$\mathfrak{t}$-equivariant} if 
%\begin{equation}
 %   \forall\ a\in\mathfrak{t}\hspace{1cm}\eta(a;\Theta)=0\label{eq:equivariance}
%\end{equation}
%In particular, if $\mathfrak{t}=\mathfrak{g}$, the condition of $\mathfrak{g}$-equivariance is equivalent to saying that $ R_\Theta$ is the trivial representation of $\mathfrak{g}$. 
This can be shown as follows:  $\mathrm{Hom}(V,\mathfrak{g})$ is equipped with a $\mathfrak{g}$-module structure through the induced  action of $\mathfrak{g}$ on $V^*$ on the one hand and through the adjoint action of $\mathfrak{g}$ on itself on the other hand. This defines  a map:
\begin{align*}
\eta:\hspace{0.2cm}\mathfrak{g}&\xrightarrow{\hspace*{1.2cm}} \hspace{0.4cm}\mathrm{End}\big(\mathrm{Hom}(V,\mathfrak{g})\big)\\
	a&\xmapsto{\hspace*{1.2cm}}\eta(a;-):\Xi\longmapsto \big(x\mapsto[a,\Xi(x)]_{\mathfrak{g}}-\Xi(a\cdot x)\big)
\end{align*}
We also write $a\cdot \Theta$ for $\eta(a;\Theta)$, and more generally in the present paper,  in order to represent the action of an element  $a\in\mathfrak{g}$ on a vector $r$ of any $\mathfrak{g}$-module $R$, we write $a\cdot r$.
%\begin{equation}
%a \cdot x=\rho(a;x)
%\end{equation}
 %$a\cdot x$ represents the action of an element  $a\in\mathfrak{g}$ on a vector $x\in V$.  .
% -- of the $\mathfrak{g}$-module $\mathrm{Hom}(V,\mathfrak{g})\simeq V^*\otimes\mathfrak{g}$, 
%\Theta_{a_1a_2\ldots a_n}=\eta(a_1;\eta(a_2;\ldots(\eta(a_n;\Theta))\ldots))$ for some $a_1,a_2, \ldots, a_n\in\mathfrak{g}$
%for $\Xi \in R_\Theta$.
Then, for any Lie subalgebra $\mathfrak{t}$ of the Lie algebra $\mathfrak{g}$, we say that $\Theta$ is \emph{$\mathfrak{t}$-equivariant} if 
\begin{equation}
    \forall\ a\in\mathfrak{t}\hspace{1cm} a\cdot \Theta=0\label{eq:equivariance0}
\end{equation}
 Then, %taking $\mathfrak{t}=\mathrm{Im}(\Theta)$, and 
 considering how $\mathfrak{g}$ acts on $V$ and on itself,  Equation \eqref{eq:equiv} can be rewritten:
\begin{equation}\label{eq:equiv2}
    \forall\ x\in V\hspace{1cm}\Theta(x) \cdot\Theta=0
\end{equation}
In other words, the embedding tensor is $\mathfrak{h}$-equivariant since $\mathfrak{h}=\mathrm{Im}(\Theta)$.
We emphasize that $\Theta$ needs not be $\mathfrak{g}$-equivariant.
 The difference between the $\mathfrak{h}$-equivariance %-- recall that we set $\mathfrak{h}=\mathrm{Im}(\Theta)$ --
  versus the $\mathfrak{g}$-equivariance of the embedding tensor $\Theta$ plays a central role in gauging procedures in supergravity. %This difference will actually be central to characterize the relationship between a Lie-Leibniz triple and Lie algebra crossed modules, that we will clarify later. 

Let us now investigate the relationship  between the adjoint map  and the embedding tensor $\Theta$.  Given a Lie-Leibniz triple $(\mathfrak{g},V,\Theta)$, the adjoint map induced by the Leibniz product is an element of $\mathrm{Hom}(V,\mathrm{End}(V))$. Since $V$ is a representation of $\mathfrak{g}$, the spaces $\mathrm{End}(V)$ and $\mathrm{Hom}(V,\mathrm{End}(V))$  canonically inherit a structure of  $\mathfrak{g}$-module. %, so    let us write $ \overline{\eta}$ (resp. $ \widehat{\eta}$) to symbolize this action. 
   Then, since Equation \eqref{eq:compat} establishes that the embedding tensor is the lift of the adjoint action, Equation \eqref{adleibniz} implies that:
\begin{equation}
    \Theta(x)\cdot\mathrm{ad}_y-\mathrm{ad}_{\Theta(x)\cdot y}=0
\end{equation}
for every $x,y\in V$. In turn, this is equivalent to:
\begin{equation}\label{eqad}
    a\cdot\mathrm{ad}=0
\end{equation}
for every $a\in \mathfrak{h}$. That is to say: the map $\mathrm{ad}$ is 
%by definition 
$\mathfrak{h}$-equivariant, as is the embedding tensor. However, it is not necessarily $\mathfrak{g}$-equivariant, i.e. the action of $\mathfrak{g}$ on $\mathrm{ad}$ may not be trivial.
This discussion justifies the following definition:

%the representation $\rho:\mathfrak{g}\to\mathrm{End}(V)$ induces a  %Usually, we denote by  $\mathfrak{h}$ the Lie subalgebra $\mathrm{Im}(\Theta)\subset\mathfrak{g}$.
 %This difference will actually be central to characterize the relationship between a Lie-Leibniz triple and Lie algebra crossed module

%This, together with the relationship between the adjoint map  and the embedding tensor $\Theta$, justifies the following definition:
 \begin{definition}
We say that a Lie-Leibniz triple $(\mathfrak{g},V,\Theta)$ is \emph{semi-strict} if the adjoint map $\mathrm{ad}:V\to\mathrm{End}(V)$ is $\mathfrak{g}$-equivariant, % the Lie algebra morphism  $\rho:\mathfrak{g}\to\mathrm{End}(V)$ defining the $\mathfrak{g}$-module structure on $V$ takes values in $\mathrm{Der}(V)$, 
and we say that it is \emph{strict} if the embedding tensor $\Theta$ is $\mathfrak{g}$-equivariant.
 \end{definition}

The  semi-strictness and the strictness conditions are equivariance properties satisfied by, respectively, the adjoint action and its lift, the embedding tensor. The relationship between these two conditions relies on the relationship between these two maps. %As can be seen on the diagram below, strictness is a refinement of the condition of semi-strictness. 
%The restriction of the representation $\rho:\mathfrak{g}\to\mathrm{End}(V)$
%to $\mathfrak{h}=\mathrm{Im}(\Theta)$
%takes values in $\mathfrak{inn}(V)$ and is a surjective Lie algebra morphism that makes the following diagram commute:
%\begin{center}
%\begin{tikzpicture}
%\matrix(a)[matrix of math nodes, 
%row sep=5em, column sep=7em, 
%text height=1.5ex, text depth=0.25ex] 
%{&\mathfrak{h}\\
%V&\mathfrak{inn}(V)\\}; 
%\path[->>](a-1-2) edge node[right]{$\rho|_{\mathfrak{h}}$} (a-2-2); 
%\path[->>](a-2-1) edge node[above]{$\mathrm{ad}$} (a-2-2);
%\path[->>](a-2-1) edge node[above left]{$\Theta$} (a-1-2);
%\end{tikzpicture}
%\end{center}
As the name indicates, a strict Lie-Leibniz triple is semi-strict but the converse is not necessarily true, as we will now show using another equivalent characterization of semi-strictness:

\begin{lemma}\label{rhoder}
The Lie-Leibniz triple $(\mathfrak{g},V,\Theta)$ is semi-strict if and only if the Lie algebra morphism  $\rho:\mathfrak{g}\to\mathrm{End}(V)$ defining the $\mathfrak{g}$-module structure on $V$ takes values in $\mathrm{Der}(V)$.
\end{lemma}

\begin{proof}
%Let $(\mathfrak{g},V,\Theta)$ be a semi-strict Lie-Leibniz triple. Denote by $ \widehat{\eta}$ the action of $\mathfrak{g}$ on $\mathrm{Hom}(V,\mathrm{End}(V))$, to which belongs $\mathrm{ad}$. 
The fact that the map $\mathrm{ad}\in\mathrm{Hom}(V,\mathrm{End}(V))$ is $\mathfrak{g}$-equivariant translates as Equation \eqref{eqad} for every $a\in\mathfrak{g}$. Then, following the reverse reasoning than led to Equation \eqref{eqad}, one finds that 
%\trd{This equation is the same as the previous one} \textcolor{blue}{here $a$ takes values in $\mathfrak{g}$ !}
this equation can be equivalently written as:
%\begin{equation}
    % \overline{\eta}(a;\mathrm{ad}_x)-\mathrm{ad}_{\rho(a;x)}=0
%\end{equation}
%for all $x\in V$ or, when applying this formula to any element $y\in V$, as:
\begin{equation}\label{hello}
a\cdot(x\circ y)-x\circ (a\cdot y )-(a\cdot x)\circ y=0
\end{equation}
for every $a\in\mathfrak{g}$ and $x,y\in V$. This equation characterizes the fact that $\rho:\mathfrak{g}\to\mathrm{End}(V)$ takes values in $\mathrm{Der}(V)$,
hence the result.
\end{proof}

%Since $\rho$ is by definition $\mathfrak{g}$-equivariant 
%with respect to the $\mathfrak{g}$-module structure on $\mathfrak{g}^*\otimes \mathrm{End}(V)$ (this is precisely the condition that $V$ is a $\mathfrak{g}$-module), 
%the above diagram shows
%Equation \eqref{eq:compat} shows that the condition that $\Theta$ is $\mathfrak{g}$-equivariant implies that $\mathrm{ad}$ is, but the converse is not true unless $\rho\big|_\mathfrak{h}:\mathfrak{h}\to\mathfrak{inn}(V)$ is bijective hence injective, i.e. unless  $V$ is faithful. Indeed, 

Going back to our discussion: using the linear constraint \eqref{eq:compat}, one can rewrite Equation \eqref{hello} as:
\begin{equation}
  a\cdot(\Theta(x) \cdot y)-\Theta(x)\cdot(a\cdot y)-\Theta(a\cdot x)\cdot y=0
\end{equation}
The left-hand side can be rewritten as $(a\cdot\Theta)(x)\cdot y$, hence the property that $(\mathfrak{g},V,\Theta)$ is semi-strict is given by the following condition:
\begin{equation}
      \forall\ a\in\mathfrak{g},\quad\forall\ x,y\in V\hspace{1cm}(a\cdot\Theta)(x)\cdot y=0\label{eq:semistrict}
\end{equation}
Thus, if the Lie-Leibniz triple $(\mathfrak{g},V,\Theta)$ is strict, i.e. if $a\cdot\Theta=0$ for every $a\in\mathfrak{g}$, then it is semi-strict, but the converse is not true, for Equation \eqref{eq:semistrict} does not necessarily imply that $\Theta$ is $\mathfrak{g}$-equivariant, unless $\rho:\mathfrak{g}\to\mathrm{End}(V)$ is injective, i.e. unless the representation $V$ is faithful. %\footnote{In supergravity theories, Lie-Leibniz triples are not necessarily semi-strict, let alone strict, but the representation $V$ is often taken to be the defining representation of $\mathfrak{g}$, which is faithful.}. 

%\begin{remark}
%A shortest proof that strictness induces semi-strictness relies on noticing that strictness implies that $R_\Theta$ is isomorphic to the irreducible (one-dimensional) trivial representation of $\mathfrak{g}$. Then, the identity $1\leq\mathrm{dim}(R_{\mathrm{ad}})\leq\mathrm{dim}(R_\Theta)$ concludes the argument.
%\end{remark}

\begin{remark}
The notion of strict Lie-Leibniz triple  had already been introduced by Loday and Pirashvili in Proposition 3.2 of \cite{LodayPiramod0}, as well as by Kinyon and Weinstein in Definition 2.3. of \cite{kinyonLeibnizAlgebrasCourant2001}, although it had not been given a name at the time. In \cite{LodayPiramod}, it is called a \emph{Loday-Pirashvili module}.
\end{remark}

 \begin{example}
 Let $(\mathfrak{g},V,\Theta)$ be a Lie-Leibniz triple, and set $\mathfrak{h}=\mathrm{Im}(\Theta)$. Then, by Equation \eqref{eq:equiv}, the Lie-Leibniz triple $(\mathfrak{h},V,\Theta)$ is strict. 
 \end{example}
 
 \begin{example}\label{example5}
By Lemma \ref{rhoder}, the Lie-Leibniz triple $\big(\mathrm{Der}(V),V,\mathrm{ad}\big)$ is semi-strict. Since the equality $f(\mathrm{ad}_x(-))=\mathrm{ad}_{f(x)}(-)+\mathrm{ad}_x( f(-))$ holds for every derivation $f$ of $V$ and $x\in V$, we deduce that the embedding tensor $\mathrm{ad}$ is $\mathrm{Der}(V)$-equivariant. This means that $\big(\mathrm{Der}(V),V,\mathrm{ad}\big)$ is actually a strict Lie-Leibniz triple.
\end{example}

\begin{example}\label{lodayexample}
We take the canonical example of \cite{loday-pirashvili} generalizing the construction of a Lie algebra associated to any associative algebra. Let $(A,\cdot)$ be a $\mathbb{R}$ associative algebra equipped with an endomorphism $D:A\to A$ which is not necessarily an algebra homomorphism but which satisfies the following identities:
\begin{equation}\label{propD}
D(x\cdot Dy)=Dx\cdot Dy=D(Dx\cdot y)
\end{equation}
Then the bilinear product $\circ:A\times A\to A$ defined as:
\begin{equation}\label{majora}
x\circ y:=Dx\cdot y-y\cdot Dx
\end{equation}
is a Leibniz product on $A$. If $D=\mathrm{id}_A$, then we obtain the canonical Lie algebra structure associated to $A$. Let us denote the corresponding Lie bracket by $[\,.\,,.\,]_A$ -- not to be confused with the skew-symmetric part of the Leibniz product \eqref{majora} -- so that Equation \eqref{majora} then reads:
 \begin{equation}\label{majora2}
x\circ y=[Dx,y]_A
\end{equation}
and we say that $\circ$ is a \emph{derived bracket} \cite{yks:derived}.

With these data we can define a canonical Lie-Leibniz triple associated to any associative algebra $A$: $\mathfrak{g}$ is the Lie algebra $(A,[\,.\,,.\,]_A)$, the Leibniz algebra $V$ is $(A,\circ)$ and the embedding tensor $\Theta$ is the endomorphism~$D$. 
 The action of the Lie algebra $\mathfrak{g}$ on $V$ is mediated through the adjoint action of the Lie bracket $[\,.\,,.\,]_A$: $x\cdot y=[x,y]_A$.  The linear constraint is precisely Equation \eqref{majora2}, while the quadratic constraint is a consequence of properties \eqref{propD}:
 \begin{equation}\label{eqhalo}
 D(x\circ y)=D(Dx\cdot y)- D(y\cdot Dx)=Dx\cdot Dy-Dy\cdot Dx=[Dx,Dy]_A
 \end{equation}
 Moreover, the gauge algebra $\mathfrak{h}$ is precisely the image of the endomorphism $D$, itself a subspace of $A$.
Now, since the Leibniz product is a derived bracket, the Jacobi identity on $\mathfrak{g}$ together with properties \eqref{propD} imply that we have, for any $x,y,z\in A$:
\begin{align}
[x,y\circ z]_A&=[x,[Dy,z]_A]_A\\
&=[[x,Dy]_A,z]_A+[Dy,[x,z]_A]_A\\
&=-[y\circ x, z]_A+y\circ [x,z]_A\\
&=[x,y]_A\circ z+y\circ [x,z]_A  -\big[y\circ x-D([y,x]_A),z\big]_A %[y,x]_A\circ z-[y\circ x, z]_A
\end{align}

So,  the (semi)-strictness of the Lie-Leibniz triple $\big((A,[\,.\,,.\,]_A),(A,\circ),D\big)$ associated to the associative algebra $A$ depends on the endomorphism $D$. The Lie-Leibniz triple is semi-strict if the bracket $ \big[y\circ x-D([y,x]_A),z\big]_A $ vanishes for every $x,y,z$ while the strictness condition -- the $\mathfrak{g}$-equivariance of the embedding tensor -- here reads $[x,Dy]_A=D[x,y]_A$  and is precisely equivalent to the vanishing of $y\circ x-D([y,x]_A)$.
Let us apply these observations to the following two cases, that were introduced in \cite{loday-pirashvili}:
\begin{enumerate}
\item $D$ is an algebra homomorphism and it is idempotent: $D^2=D$;
%\item $A$ is a superalgebra so that any $x\in A$ can be written $x=x_++x_-$, and $D(x)=x_+$;
\item $D$ is a differential: $D(x\cdot y)=Dx\cdot y+x\cdot Dy$ and $D^2=0$.
\end{enumerate}

In the first case, we first deduce from the homomorphism property that $D([x,y]_A)=[Dx,Dy]_A$. By  Equation \eqref{eqhalo} we then conclude that $D([x,y]_A)=D(x\circ y)$ so the Lie-Leibniz triple associated to the associative algebra $A$ and the endomorphism $D$ is strict if and only if $x\circ y=D(x\circ y)$. In the second case, we deduce from the derivation property that $[Dx,y]_A=D([x,y]_A)+[Dy,x]_A$ so we obtain the skew-symmetric bracket of the Leibniz product:
\begin{equation}
[x,y]=\frac{1}{2}D([x,y]_A)
\end{equation}
This has the following consequences: first, $D(x\circ y)=D([x,y])=0$ because $D^2=0$, and second: $y\circ x-D([y,x]_A)=x \circ y$ so the corresponding Lie-Leibniz triple is strict if and only if $x\circ y=0$, i.e. the Leibniz algebra structure is zero.

 % Then $[D(y\circ x),z]=(y\circ x)\circ z=y \circ (x\circ z)-x\circ (y\circ z)$
\end{example}

Let us now observe some consequences of semi-strictness on the 
 symmetric bracket $\{\,.\,,.\,\}$ entailed by the Leibniz algebra structure $\circ$ on $V$. This bracket defines a map from $S^2(V)$ to $V$, also denoted $\{\,.\,,.\,\}$. This map, as an element of $\mathrm{Hom}(S^2(V),V)$ is subject to the action of $\mathfrak{g}$ induced from the $\mathfrak{g}$-module structure on~$V$:
\begin{equation}\label{zeus}
a\cdot\big(\{\,.\,,.\,\}\big)(x,y)=a\cdot\{x,y\}-\{a\cdot x,y\}-\{x,a\cdot y\}
\end{equation}
for every $a\in\mathfrak{g}$ and $x,y\in V$. The $\mathfrak{g}$-equivariance of the symmetric bracket is equivalent to the vanishing of the right hand side, which is a priori not granted. However, thanks to the linear constraint \eqref{eq:compat}, the left-hand side of Equation \eqref{zeus} can be written as the following:
\begin{equation}\label{hera0}
a\cdot\big(\{\,.\,,.\,\}\big)(x,y)=\frac{1}{2}\big((a\cdot \Theta)(x)\big)\cdot y+\frac{1}{2}\big((a\cdot \Theta)(y)\big)\cdot x %=a \cdot \{x, x\}-2\{x, a\cdot x\}
\end{equation}
%for every $a\in \mathfrak{g}$ and $x\in V$. %, and where $\eta:\mathfrak{g}\to\mathrm{End}(R_\Theta)$ encodes the representation of $\mathfrak{g}$ on $R_\Theta$. 
Notice that, due to Equation \eqref{eq:equiv2}, this term is always 0 when $a\in\mathfrak{h}$, thus implying that the symmetric bracket is always $\mathfrak{h}$-equivariant, but there is no guarantee that it is $\mathfrak{g}$-equivariant.   However we deduce from Equation \eqref{eq:semistrict} that a sufficient condition for the symmetric bracket to be $\mathfrak{g}$-equivariant is that the Lie-Leibniz triple $(\mathfrak{g},V,\Theta)$ is semi-strict, so that the right-hand side of Equation \eqref{hera0} vanishes.
 These observations imply that, for any Lie-Leibniz triple, the space  $\mathrm{Ker}\big(\{\,.\,,.\,\}\big)\subset S^2(V)$ is always a $\mathfrak{h}$-module but not necessarily a $\mathfrak{g}$-module. This latter fact is central in the construction of the tensor hierarchy associated to a Lie-Leibniz triple. %Given that this fact is of crucial importance in the construction of the tensor hierarchy associated to $(\mathfrak{g},V,\Theta)$, we give a name to Lie-Leibniz triples for which such a case occur:

%\begin{align}
%[y,x]_A\circ z-[y\circ x, z]_A&=D(yx)\cdot z-D(xy)\cdot z -z\cdot D(yx)+z\cdot D(xy)-D(y)\cdot x\cdot z+x\cdot Dy\cdot z + z\cdot D(y)\cdot x-z\cdot x\cdot Dy\\
%&=\Big(D([y,x]_A)-y\circ x\Big)\cdot z-z\cdot \Big(D([y,x]_A)-y\circ x\Big)\\
%&=\Big[D([y,x]_A)-y\circ x,z\Big]_A
%\end{align}

\begin{comment}

\begin{definition}
A Lie-Leibniz triple $(\mathfrak{g},V,\Theta)$ is called \emph{stringent} when $\mathrm{Ker}\big(\{\,.\,,.\,\}\big)\subset S^2(V)$ is a $\mathfrak{g}$-module.
\end{definition}

\begin{remark}
Although $\mathfrak{g}$-equivariance of the symmetric bracket $\{\,.\,,.\,\}$ entails that the Lie-Leibniz triple is stringent, the converse does not hold.
This also shows that being stringent is strictly weaker than being semi-strict, for a Lie-Leibniz triple may have a $\mathfrak{g}$-equivariant symmetric bracket $\{\,.\,,.\,\}$ -- i.e. the right-hand side of Equation \ref{hera0} vanishes -- without satisfying Equation \eqref{eq:semistrict}. 
\end{remark}

\begin{example}
Lie-Leibniz triples such that $V$ is a Lie algebra, are stringent, because in the latter case $\mathrm{Ker}\big(\{\,.\,,.\,\}\big)=S^2(V)$, which is a $\mathfrak{g}$-module by definition. 
\end{example}
\end{comment}

\subsection{A relationship between Lie-Leibniz triples and Lie algebra crossed modules}\label{relatio}

In order to study in a more consistent way the notion of Lie-Leibniz triple with respect to what is already known for Lie algebra crossed modules, we make use of categories.
That is why we now introduce
%to introduce 
the correct notion for morphisms of Lie-Leibniz triples:
%; it is necessary in order to set up a more categorical point of view:
\begin{definition}\label{defmorphism}
A \emph{morphism} between  two Lie-Leibniz triples $(\mathfrak{g},V,\Theta)$ and $(\mathfrak{g}',V',\Theta')$  is a pair $(\varphi,\chi)$ consisting of a Lie algebra morphism $\varphi:\mathfrak{g}\to \mathfrak{g}'$ and a Leibniz algebra morphism $\chi: V\to V'$, satisfying the following compatibility conditions:
\begin{align}
\Theta'(\chi(x))&=\varphi(\Theta(x))\label{eq:jfk}\\
\chi(a\cdot x)&=\varphi(a)\cdot\chi(x)\label{eq:jfk2}
\end{align}
for every $a\in \mathfrak{g}$ and $x\in V$. %where $\rho$ (resp. $\rho'$) denotes the action of $\mathfrak{g}$ (resp. $\mathfrak{g}'$) on $V$ (resp. $V'$). %We say that $(\varphi,\chi)$ is \emph{confined} if:
%\begin{equation}\label{confined}
%    (\chi\odot\chi)(K)\subset K'
%\end{equation}
%where $K$ (resp. $K'$) is the biggest $\mathfrak{g}$-submodule of $\mathrm{Ker}\big(\{\,.\,,.\,\}\big)$ (resp. the biggest $\mathfrak{g}'$-submodule of $\mathrm{Ker}\big(\{\,.\,,.\,\}'\big)$). 
%Equations \eqref{eq:compat}, \eqref{eq:jfk} and \eqref{eq:jfk2}, together with the Leibniz morphism property of $\chi$ and $\varphi$, form a commutative prism:
%\begin{center}
%\begin{tikzcd}[row sep=4em, column sep=5em, 
%text height=1.5ex, text depth=0.25ex]
%V\otimes V\arrow[drr,"\Theta\otimes %\mathrm{id}"]\arrow[rrr,"\circ_V"]\arrow[dd,"\chi\otimes\chi"]  && & %V\arrow[dd,"\chi"] \\
% && \mathfrak{g}\otimes V\arrow[ur,"\rho"] &   \\
%V'\otimes V'\arrow[drr,"\Theta'\otimes %\mathrm{id}"]\arrow[rrr,"\circ_{V'}"]  && &V'\\
 %&& \mathfrak{g}'\otimes V'\arrow[ur,"\rho'"] \arrow[from=uu, crossing %over, "\varphi\otimes\chi" above right]&\\
% \end{tikzcd}
%\end{center}
%\begin{center}
That is to say, the following prism is commutative:
\begin{center}
\begin{tikzpicture}[on top/.style={preaction={draw=white,-,line width=#1}},
on top/.default=5pt]
\matrix(a)[matrix of math nodes, 
row sep=2em, column sep=2em, 
text height=1ex, text depth=0.25ex] 
{ &\mathfrak{g}\otimes V&&&\mathfrak{g}'\otimes V'\\
&&&&\\
V\otimes V&&&V'\otimes V'&\\ 
&V&&&V'\\}; 
\path[->](a-3-1) edge node[above right]{$\chi\otimes\chi$} (a-3-4); %node[below right]{$\varphi$}  (a-1-2); 
\path[->](a-3-1) edge node[above left]{$\Theta\otimes\mathrm{id}$} (a-1-2);
\path[->](a-3-1) edge node[below left]{$\circ$} (a-4-2);
\path[->](a-3-4) edge node[above right]{$\circ'$} (a-4-5);
\path[->](a-3-4) edge node[above left]{$\Theta'\otimes\mathrm{id}$} (a-1-5);
\path[->](a-1-2) edge node[above]{$\varphi\otimes\chi$} (a-1-5);
\path[->](a-4-2) edge node[above]{$\chi$} (a-4-5);
\path[->](a-1-2) edge[on top=5pt] node[above left]{$\rho$} (a-4-2);
\path[->](a-1-5) edge node[above right]{$\rho'$} (a-4-5);
\end{tikzpicture}
\end{center}
%\begin{center}
%\begin{tikzpicture}[on top/.style={preaction={draw=white,-,line width=#1}},
%on top/.default=5pt]
%\matrix(a)[matrix of math nodes, 
%row sep=2em, column sep=2em, 
%text height=1ex, text depth=0.25ex] 
%{ &\mathfrak{g}&&&\mathfrak{g}'\\
%&&&&\\
%V&&&V'&\\ 
%&f(Ker(\chi))\subset Ker (\chi)&&&\mathrm{Hom}(\chi(V),V')\\}; 
%\path[->](a-3-1) edge node[above right]{$\chi$} (a-3-4); %node[below right]{$\varphi$}  (a-1-2); 
%\path[->](a-3-1) edge node[above left]{$\Theta$} (a-1-2);
%\path[->](a-3-1) edge node[below left]{$\mathrm{ad}$} (a-4-2);
%\path[->](a-3-4) edge node[above right]{$\mathrm{ad}'$} (a-4-5);
%\path[->](a-3-4) edge node[above left]{$\Theta'$} (a-1-5);
%\path[->](a-1-2) edge node[above]{$\varphi$} (a-1-5);
%\path[->](a-4-2) edge node[above]{$\chi$} (a-4-5);
%\path[->](a-1-2) edge[on top=5pt] node[above left]{$\rho$} (a-4-2);
%\path[->](a-1-5) edge node[above right]{$\rho'$} (a-4-5);
%\end{tikzpicture}
%\end{center}
where $\rho$ (resp. $\rho'$) denotes the action of $\mathfrak{g}$ (resp. $\mathfrak{g}'$) on $V$ (resp. $V'$). 
We say that $(\varphi,\chi)$ is an \emph{isomorphism of Lie-Leibniz triples} when both $\varphi$ and $\chi$ are isomorphisms in their respective categories.
\end{definition}

\begin{remark}\label{ex:semi}
Let $(V,\circ)$ be a Leibniz algebra, then any Lie-Leibniz triple $(\mathfrak{g}, V,\Theta)$ whose embedding tensor generates the Leibniz product $\circ$ as given in Equation \eqref{eq:compat} induces a canonical Lie-Leibniz triple morphism $(\rho,\mathrm{id}_V)$ from $(\mathfrak{g}, V,\Theta)$ to $(\mathrm{End}(V),V,\mathrm{ad})$. Moreover, if $(\mathfrak{g}, V,\Theta)$  is semi-strict, the image of the Lie-Leibniz triple morphism $(\rho,\mathrm{id}_V)$ is a sub-Lie-Leibniz triple of $(\mathrm{Der}(V),V,\mathrm{ad})$.
\end{remark}

A particular case of Lie-Leibniz triple morphisms are those for which the linear map $\chi\odot \chi:S^2V\to S^2V'$ -- here, $\odot$ symbolizes the symmetric product -- preserves a particular submodule of $\mathrm{Ker}\big(\{\,.\,,.\,\}\big)\subset S^2(V)$. Assume that we have a morphism  $(\varphi,\chi)$ between two Lie-Leibniz triples $(\mathfrak{g},V,\Theta)$ and $(\mathfrak{g}',V',\Theta')$. The kernel of the symmetric bracket on $V$ is a subspace $\mathrm{Ker}\big(\{\,.\,,.\,\}\big)$ of $S^2(V)$ and a $\mathfrak{h}$-module but not necessarily a $\mathfrak{g}$-module. However, it admits $\mathfrak{g}$-submodules, and among them one is the biggest, that we denote by $K$. The same occurs for the symmetric bracket on $V'$ and the biggest $\mathfrak{g}'$-submodule of $\mathrm{Ker}\big(\{\,.\,,.\,\}'\big)\subset S^2(V')$ is denoted $K'$. 
Since $\chi$ is a morphism of Leibniz algebras, we have that:
\begin{equation}\label{inclusionker}
\chi\odot \chi(K)\subset \mathrm{Ker}\big(\{\,.\,,.\,\}'\big)
\end{equation} By Equation \eqref{eq:jfk2}, the vector space $\chi\odot \chi(K)$ is only a $\varphi(\mathfrak{g})$-module, and not a $\mathfrak{g}'$-module.
So a priori, there is no reason that $\chi\odot \chi(K)$ would be a subspace of $K'$. We decide to give a name to morphisms that precisely have this property:  %It is a priori not guaranteed that $\chi\odot \chi:S^2V\to S^2V'$ sends $K$ into $K'$ and when it is the case

\begin{definition}\label{defcomp}
We say that a morphism of Lie-Leibniz triples    $(\varphi,\chi):(\mathfrak{g},V,\Theta)\to(\mathfrak{g}',V',\Theta')$ is \emph{compatible} if \begin{equation}\label{comstring}\chi\odot \chi(K)\subset K'\end{equation}
\end{definition}

This notion of compatible morphism  allows us to define several categories: we call \textbf{Lie-Leib}  the category of Lie-Leibniz triples with their associated morphisms given in Definition \ref{defmorphism}.  We call \textbf{compLie-Leib} the wide subcategory %\footnote{A \emph{wide} subcategory of a category $\mathcal{C}$ is a category $\mathcal{B}$ whose collection of objects is the same as that of $\mathcal{C}$ -- $\mathrm{Ob}(\mathcal{B})=\mathrm{Ob}(\mathcal{C})$ -- but whose collection of morphisms does not necessarily coincides with that of $\mathcal{C}$. }
 of \textbf{Lie-Leib} whose morphisms are the compatible morphisms of Lie-Leibniz triples. In other words, Ob(\textbf{compLie-Leib})=Ob(\textbf{Lie-Leib}), while Mor(\textbf{compLie-Leib}) is contained in   Mor(\textbf{Lie-Leib}).  The semi-strict (resp. strict)  Lie-Leibniz triples form a full subcategory %\footnote{A \emph{full} subcategory of a category $\mathcal{C}$ is a category $\mathcal{B}$ whose collection of objects is  included in that of $\mathcal{C}$ -- $\mathrm{Ob}(\mathcal{B})\subset\mathrm{Ob}(\mathcal{C})$ -- but whose morphisms between elements are the same as that of  $\mathcal{C}$, in the following sense: for every pair of objects $X,Y\in \mathrm{Ob}(\mathcal{B})$, we have the identity $\mathrm{Hom}_{\mathcal{B}}(X,Y)=\mathrm{Hom}_{\mathcal{C}}(X,Y)$.} 
 of both \textbf{compLie-Leib} and \textbf{Lie-Leib}, denoted \textbf{semLie-Leib} (resp. \textbf{strLie-Leib}). This can be seen from the fact that in semi-strict Lie-Leibniz triples, the biggest $\mathfrak{g}$-submodule of $\mathrm{Ker}\big(\{\,.\,,.\,\}\big)\subset S^2(V)$ is the kernel itself. So,  for $K=\mathrm{Ker}\big(\{\,.\,,.\,\}\big)$ and $K'=\mathrm{Ker}\big(\{\,.\,,.\,\}'\big)$, inclusion \eqref{inclusionker}, which is valid for any morphism of Lie-Leibniz triples, becomes $\chi\odot \chi(K)\subset K'$. Hence any morphism of semi-strict Lie-Leibniz triples is compatible.
 
%Indeed, the subcategory of stringent Lie-Leibniz triple is a \emph{wide} subcategory and not a full one: it contains all Lie-Leibniz triples, i.e. all objects of \textbf{Lie-Leib}, but 
These inclusions of subcategories will allow us  to be more precise when studying the relationship between Lie-Leibniz triples and Lie algebra crossed modules.
Recall the definition of the latter notion: %\footnote{We chose the convention that Lie algebra crossed modules are concentrated in degrees $-1$ and $0$ so that this grading is consistent with the grading of tensor hierarchies, of which Lie algebra crossed modules are a particular case.}:

\begin{definition}\label{defdiff}
A \emph{Lie algebra crossed module} consists of the following data: two Lie algebras $(\mathfrak{g},[\,.\,,.\,]_{\mathfrak{g}})$ and $(\mathfrak{c},[\,.\,,.\,]_{\mathfrak{c}})$, a linear map $\Theta:\mathfrak{c}\to \mathfrak{g}$ and a Lie algebra morphism $\rho:\mathfrak{g}\to \mathrm{Der}(\mathfrak{c})$, satisfying the following axioms:
\begin{align}
\rho(\Theta(x); y)&=[x,y]_{\mathfrak{c}}
\label{equatdiff1}\\
    \Theta(\rho(a;x))&=[a,\Theta(x)]_{\mathfrak{g}}\label{equatdiff2}
\end{align}
%\trd{OR :denoting $\rho(a;x)$ as $a\cdot x$ \begin{align}
%\Theta(x)\cdot y &=[x,y]_{1}
%\label{equatdiff3}\\
%\Theta(a\cdot x)&=[a,\Theta(x)]_{\mathfrak{g}}
%%end{align}}
for every $x,y\in\mathfrak{c}$ and $a\in\mathfrak{g}$. We denote by \emph{\textbf{Lie$\times$Mod}} the category of Lie algebra crossed modules. \end{definition} %\footnote{Morphisms in the category of Lie algebra crossed modules are pairs of Lie algebra morphisms $(\varphi:\mathfrak{g}\to\mathfrak{g}'_0,\chi:\mathfrak{c}\to\mathfrak{g}'_{-1})$ satisfying  Equations \eqref{eq:jfk} and \eqref{eq:jfk2}.}.

%\begin{remark} Equivalently, a Lie algebra crossed module can be thought of as  a differential graded Lie algebra concentrated in degree 0 and $-1$ with differential $\Theta$, also known as a \emph{strict} Lie 2-algebra,
%  \emph{i.e.} 2-term $L_\infty$-algebras where the 3-brackets and higher are zero \cite{baez:Lie2alg}. This aspect and its straightforward generalization to Lie-Leibniz triples is the heart of the paper and will be discussed in more details in Section \ref{section3}. %Alternatively, they can be considered 
 % \end{remark}
  %, for $\mathfrak{g}$ and $V$, respectively. 

%Although Lie algebra crossed modules were originally defined in terms of just Lie algebras and Lie modules, we %use the extension to 
%extend to a case involving Leibniz algebras.
%derived from crossed products of Lie groups, one may want to take some distance from this historical origin and study Lie algebra crossed modules as mathematical objects \emph{per se}.

%The first observation is that some of the conditions appearing in the definition of Lie algebra crossed modules may be too stringent. 

%\trd{what is being relaxed? also why not have the first definition in terms of $g_0,g_1$ and save g and V for the following?} \textcolor{blue}{great idea!}

%Let us rewrite Definition \ref{defdiff} as follows:

Notice that the definition of $\rho$ implies that $\mathfrak{c}$ is a  $\mathfrak{g}$-module, equipped with a Lie algebra structure entirely defined by Equation \eqref{equatdiff1}. Moreover, Equation \eqref{equatdiff2} for $a=\Theta(z)$ implies that $\Theta$ is a Lie algebra morphism as well: $\Theta([z,x]_{\mathfrak{c}})=[\Theta(z),\Theta(x)]_{\mathfrak{g}}$.  These observations imply that we have an alternative, equivalent definition of Lie algebra crossed modules that relies on the following data:
% \begin{remark}
%The equation:
% is not a tautology if we had only assumed that 1)
\begin{lemma}\label{lemmacapcom} A Lie algebra crossed module consists of the following data:
 \begin{enumerate}
 \item a Lie algebra $\mathfrak{g}$,
 \item a $\mathfrak{g}$-module $\mathfrak{c}$, and 
 \item a $\mathfrak{g}$-equivariant linear map $\Theta:\mathfrak{c}\to \mathfrak{g}$ satisfying:
 \begin{equation}
\Theta(x)\cdot y=-\Theta(y)\cdot x\label{loof1}
\end{equation}
\end{enumerate}
 \end{lemma}
 
 \begin{proof}
  Because of Equation \eqref{loof1}, the vector space $\mathfrak{c}$ can be equipped with a skew-symmetric bracket $[x,y]_{\mathfrak{c}}=\Theta(x)\cdot y$. This definition, together with the fact that $\mathfrak{c}$ is a $\mathfrak{g}$-module has two consequences: on the one hand, the bracket $[\,.\,,.\,]_{\mathfrak{c}}$ satisfies the Jacobi identity and on the other hand,  the representation map $\rho:\mathfrak{g}\to \mathrm{End}(\mathfrak{c})$ takes values in $\mathrm{Der}(\mathfrak{c})$. Eventually, Equation \eqref{equatdiff1} is the definition of the bracket on $\mathfrak{c}$ while Equation \eqref{equatdiff2} corresponds to the $\mathfrak{g}$-equivariance of $\Theta$.
  \end{proof}

 The criteria established in Lemma \ref{lemmacapcom} are reminiscent of the definition of a Lie-Leibniz triples. Items 1. and 2. establish that one can equip $\mathfrak{c}$ with a Leibniz product, as in Equation \eqref{eq:compat}, but Equation \eqref{loof1} implies that this product is skew-symmetric, hence a Lie bracket.
 Moreover, the $\mathfrak{g}$-equivariance of $\Theta$ that is constitutive of Lie algebra crossed modules straightforwardly implies the quadratic constraint \eqref{eq:equiv}. We then have the following relationship between Lie algebra crossed modules and Lie-Leibniz triples: %The condition that $\rho:\mathfrak{g_0}\to\mathrm{End}(\mathfrak{c})$ takes values in $\mathrm{Der}(\mathfrak{c})$ is actually a consequence 
\begin{lemma}\label{categdiffmod}
%Every Lie algebra crossed module is a strict Lie-Leibniz triple in which the Leibniz algebra structure is a Lie algebra 
The category \emph{\textbf{Lie$\times$Mod}} forms a full subcategory of \emph{\textbf{strLie-Leib}}.%, in which the Leibniz product on $V$ is a skew-symmetric i.e. it 
\end{lemma}
\begin{proof}
Let $\mathfrak{c}\overset{\Theta}{\longrightarrow}\mathfrak{g}$ be a Lie algebra crossed module. Then, Equations \eqref{equatdiff1} and \eqref{equatdiff2} imply Equations \eqref{eq:compat} and \eqref{eq:equiv}, under the following assumptions:
\begin{enumerate}
\item  $\mathfrak{g}=\mathfrak{g}$ and $V=\mathfrak{c}$;
    \item the Leibniz product $\circ$ on $\mathfrak{c}$ corresponds to the Lie algebra structure $[\,.\,,.\,]_{\mathfrak{c}}$.
    \end{enumerate}
    This makes $(\mathfrak{g},\mathfrak{c},\Theta)$ a Lie-Leibniz triple. Moreover, Equation \eqref{equatdiff2} tells us that the embedding tensor $\Theta$ is $\mathfrak{g}$-equivariant, i.e. that the Lie-Leibniz triple is strict.  
    %Then, the Lie algebra crossed module $\mathfrak{c}\overset{\Theta}{\longrightarrow}\mathfrak{g}$ can be seen as a strict Lie-Leibniz triple $(\mathfrak{g},V,\Theta)$. %,  i.e. the Lie-Leibniz triple $(\mathfrak{g},\mathfrak{c}, \Theta)$ is strict. 
    %Conversely, any strict Lie-Leibniz triple $(\mathfrak{g},V,\Theta)$ is a Lie algebra crossed module when $V$ is a Lie algebra.
  %  Notice that this is consistent with the fact that the representation $\rho:\mathfrak{g}\to\mathrm{End}(\mathfrak{c})$ takes values in $\mathrm{Der}(\mathfrak{c})$, which is equivalent, by Lemma \ref{rhoder}, to being semi-strict. 
    Morphisms of Lie algebra crossed modules are precisely the morphisms of the corresponding Lie-Leibniz triples, hence the result.
\end{proof}

This lemma characterizes Lie algebra crossed modules as the strict Lie-Leibniz triples $(\mathfrak{g},V,\Theta)$ for which $V$ is a Lie algebra and for which the embedding tensor is $\mathfrak{g}$-equivariant. It moreover implies that we have the following inclusions of categories:

\begin{equation*}
    \textbf{Lie$\times$Mod}\subset\textbf{strLie-Leib}\subset\textbf{semLie-Leib}\subset\textbf{compLie-Leib}\subset\textbf{Lie-Leib}
\end{equation*}
where each inclusion is full in the sense of categories, except on the right-most inclusion where \textbf{compLie-Leib} is a wide subcategory of \textbf{Lie-Leib}.
The similarity between Lie algebra crossed modules and Lie-Leibniz triples is striking: %Equations \eqref{equatdiff1} and \eqref{equatdiff2} are particular cases of Equations \eqref{eq:compat} and \eqref{eq:equiv}, respectively. 
%Conversely, 
one can see the latter as natural generalizations of the former, in which some stringent conditions (a Lie algebra structure on $\mathfrak{c}$ and $\mathfrak{g}$-equivariance of $\Theta$) have been relaxed.

Let us now push the analogy further by analyzing the differential graded Lie algebra structure of Lie algebra crossed modules and its natural extension to Lie-Leibniz triples. %The next section is devoted to showing that the relationship between Lie algebra crossed modules and Lie-Leibniz triples actually goes further.
  %Since Lie algebra crossed modules are particular cases of strict Lie-Leibniz triples in which the Leibniz product $\circ$ induced on $V$ by the action of $\mathfrak{g}$ is skew-symmetric -- hence defining a Lie algebra structure on $V$ --  the following inclusions of categories:
%In the last section, we have shown that Lie algebra crossed modules are particular cases of (strict) Lie-Leibniz triples in which the Leibniz product $\circ$ induced on $V$ by the action of $\mathfrak{g}$ is skew-symmetric -- hence defining a Lie algebra structure on $V$. 
Indeed, an important property of Lie algebra crossed modules is that they can be considered as 2-term differential graded Lie algebras concentrated in degrees 0 and 1 \cite{baez:Lie2alg} or, equivalently, in degrees 0 and $-1$ (our convention). This correspondence is derived from the equivalence of categories between the category of crossed modules of Lie groups and that of strict Lie 2-groups \cite{MR2841564, janelidzeInternalCrossedModules2003}.  We recall here how the identification works, since it will be a cornerstone of the present section.

Given a Lie algebra crossed module $\mathfrak{c}\xrightarrow{\Theta}\mathfrak{g}$, assign degree $-1$ to $\mathfrak{c}$ and degree 0 to $\mathfrak{g}$. Then, define the graded Lie bracket on $\mathfrak{c}\oplus \mathfrak{g}$ as follows:
\begin{align}
    \llbracket a,b\rrbracket&=[a,b]_{\mathfrak{g}}\label{bronze}\\
        \llbracket x,y\rrbracket&=0\label{gold}\\
    \llbracket a,x \rrbracket&=a\cdot x\qquad \text{and}\qquad  \llbracket x,a \rrbracket=-a\cdot x\label{silver}
\end{align}
for every $a,b\in\mathfrak{g}$ and $x,y\in \mathfrak{c}$. This bracket and the differential $\Theta$ are indeed compatible, in the sense that together they form a differential graded Lie algebra structure. The first non-trivial compatibility condition between the graded Lie bracket $\llbracket  \,.\,,.\,\rrbracket$  and the differential $\Theta$ is:
  \begin{equation}\label{compatdiff}
      \Theta\big(\llbracket a, x\rrbracket\big)=\llbracket a, \Theta(x)\rrbracket
\end{equation}
This is nothing but Equation \eqref{equatdiff2}.
The other non-trivial compatibility condition %would be:
is:% Equation \eqref{equatdiff1} can be recasted as:
\begin{equation}\label{compatdiff2}
\Theta\big(\llbracket  x,y\rrbracket\big)=\llbracket \Theta(x), y\rrbracket-\llbracket x, \Theta(y)\rrbracket
  %  \Theta\big(\llbracket a, x\rrbracket\big)&=\llbracket a, \Theta(x)\rrbracket
\end{equation}
Notice that the left hand side  vanishes identically because of Equation \eqref{gold}, and then the right hand side just tells us that $\Theta(x)\cdot y=-\Theta(y)\cdot x$ (Equation \eqref{loof1}). However this is a tautology because by Equation \eqref{equatdiff1}, $\Theta(x)\cdot y=[x,y]_{\mathfrak{c}}$ and we know that the Lie bracket on $\mathfrak{c}$ is skew-symmetric. Eventually, the bracket defined in Equations \eqref{bronze}-\eqref{silver} satisfies the Jacobi identity as the only non-obvious identity to satisfy is the following one:
\begin{equation}\label{compatdiff3}
\big\llbracket a,\llbracket b,x\rrbracket\big\rrbracket-\big\llbracket b,\llbracket a,x\rrbracket\big\rrbracket=\big\llbracket \llbracket a,b\rrbracket,x\big\rrbracket
\end{equation}
This is an alternative, equivalent rewriting of the fact that $\rho:\mathfrak{g}\to\mathrm{Der}(\mathfrak{c})$ is a Lie algebra morphism. 
 %In turn, this map satisfies the Jacobi identity  because  $\rho:\mathfrak{g}\to\mathrm{Der}(\mathfrak{c})$ is a Lie algebra morphism (one uses Equations \eqref{silver} and \eqref{compatdiff} in the proof). Hence, it defines a Lie bracket on $\mathfrak{c}$, which precisely coincides with the Lie bracket defined in Equation \eqref{equatdiff1}.
%
%
 %and hence defines a Lie bracket on $\mathfrak{c}$ because $\rho:\mathfrak{g}\to\mathrm{Der}(\mathfrak{c})$
 Thus, a Lie algebra crossed module is equivalent to a 2-term differential graded Lie algebra concentrated in degrees $0$ and $-1$.

Let us now present an alternative, equivalent formulation of this observation by using the cyclic representation generated by $\Theta$. Let $(\mathfrak{g},V,\Theta)$ be a Lie-Leibniz triple and let $\eta:\mathfrak{g}\to\mathrm{End}\big(\mathrm{Hom}(V,\mathfrak{g})\big)$ denote the induced action of $\mathfrak{g}$ on $\mathrm{Hom}(V,\mathfrak{g})$ as defined in Section \ref{subsecc}.
 As an element of $\mathrm{Hom}(V,\mathfrak{g})$, the embedding tensor $\Theta$ generates a cyclic $\mathfrak{g}$-submodule of the $\mathfrak{g}$-module $\mathrm{Hom}(V,\mathfrak{g})$; we call it
 $ R_\Theta$. More precisely, noting 
 \begin{equation}
 \Theta_{a_1a_2\ldots a_n}:=a_1\cdot(a_2\cdot (\ldots (a_n\cdot\Theta)\ldots))=\eta(a_1;\eta(a_2;\ldots(\eta(a_n;\Theta))\ldots))
 \end{equation}
  for any $a_1,a_2, \ldots, a_n\in\mathfrak{g}$, we have:
 \begin{equation}\label{defrep}
R_\Theta:=\mathrm{Span}(\Theta, \Theta_{a_1a_2\ldots a_m}\,|\, a_1,a_2,\ldots,a_m\in\mathfrak{g})\subset \mathrm{Hom}(V, \mathfrak{g})
\end{equation}
In the case where the Lie-Leibniz triple $(\mathfrak{g},V,\Theta)$ is a Lie algebra crossed module -- i.e. when $\mathfrak{g}=\mathfrak{g}$ and $V=\mathfrak{c}$ -- then $\Theta$ being $\mathfrak{g}$-equivariant means that $R_\Theta$ is the  one-dimensional irreducible trivial representation of $\mathfrak{g}$ and can be identified with $\mathbb{R}$.
Then, the differential graded Lie algebra structure on $\mathfrak{c}\oplus \mathfrak{g}$ presented earlier %Lie algebra crossed module $\mathfrak{c}\xrightarrow{\Theta}\mathfrak{g}$ 
can be straightforwardly extended to a 3-term differential graded Lie algebra structure on the following cochain complex:
\begin{center}
\begin{tikzcd}[column sep=1cm,row sep=0.4cm]
\mathfrak{c}\ar[r,"\Theta"]&\mathfrak{g}\ar[r,"0"]&\mathbb{R}
\end{tikzcd}
\end{center}
%where, at the right extremity, $\mathbb{R}$ denotes the reals and is subject to the trivial action of $\mathfrak{g}$. 
 One consider that elements of $\mathbb{R}$ carry a degree $+1$, so that the grading on the complex is increasing from left to right.
%\trd{OR: where, at the right extremity, $\mathbb{R}$ is the $\mathfrak{g}$-module $\mathrm{Rep}(\Theta)$ generated by $\Theta$ in $\mathrm{Hom}(V,\mathfrak{g})$. }. 
%The linear map $\Theta:\mathfrak{c}\to \mathfrak{g}$, becoming zero under the action of $\mathfrak{g}$ by Equation \eqref{equatdiff2}, generates a one-dimensional $\mathfrak{g}$-submodule of $\mathrm{Hom}(\mathfrak{c},\mathfrak{g})$, which can be identified with the irreducible trivial representation of $\mathfrak{g}$. 
Actually, the above sequence is a particular case of the following observation:
%Every $\mathfrak{g}$-invariant tensor or operator generates a  can be seen as an element of $\mathbb{R}$ the field of real numbers; in particular this is the case for $\Theta$, since it is $\mathfrak{g}$-invariant by Equation \eqref{equatdiff2}. In other words, $\mathbb{R}$ corresponds to the $\mathfrak{g}$-module $\mathrm{Rep}(\Theta)$ generated by $\Theta$ in $\mathrm{Hom}(V,\mathfrak{g})$. 

\begin{prop}\label{experimental}
To every Lie-Leibniz triple $(\mathfrak{g},V,\Theta)$ in which $V$ is a Lie algebra, there corresponds a unique 3-term differential graded Lie algebra concentrated in degrees $-1$, $0$ and $1$\footnote{Our convention for suspension is the following: for any graded vector space $U_\bullet=\bigoplus_{i\in \mathbb{Z}}U_{i}$, the \emph{suspension operator} $[-1]:U\to U[-1]$ shifts the degree of every elements in $U$ by $1$, i.e. it is such that $(U[-1])_{i}=U_{i-1}$. The \emph{desuspension operator} $[1]$ has the reverse effect.}:
\begin{center}
\begin{tikzcd}[column sep=1cm,row sep=0.4cm]
V[1]\ar[r,"\Theta"]&\mathfrak{g}\ar[r,"-\eta(-;\Theta)"]& R_\Theta[-1]
\end{tikzcd}
\end{center}
which canonically extends the graded Lie algebra structure on $V[1]\oplus \mathfrak{g}$ defined by Equations \eqref{bronze}-\eqref{silver}. %Here, $ R_\Theta\subset \mathrm{Hom} (V,\mathfrak{g})$ is the cyclic $\mathfrak{g}$-submodule generated by the embedding tensor $\Theta$ (see Equation \eqref{defrep}), and $\eta:\mathfrak{g}\to\mathrm{End}( R_\Theta)$ denotes the associated action of $\mathfrak{g}$.
\end{prop}

%\begin{remark}\label{remarksusp}
%Our convention for the \emph{suspension operator} is the following: for any graded vector space $U_\bullet=\bigoplus_{i\in \mathbb{Z}}U_{i}$, the operator $[1]:U\to U[1]$ shifts the degree of every elements in $U$ by $-1$, i.e. it is such that $(U[1])_{i}=U_{i+1}$. The \emph{desuspension operator} $[-1]$ has the reverse effect.
%Recall that, for any graded vector space $U_\bullet=\bigoplus_{i\in \mathbb{Z}}U_{i}$, the operator $[1]:U\to U[1]$ is the \emph{suspension operator} that shifts the degree of every elements in $U$ by $-1$, i.e. it is such that $(U[1])_{i}=U_{i+1}$. It admits an inverse \emph{desuspension operator} $[-1]:U\to U[-1]$, that has the reverse effect. %See \cite{LavauPalmkvist} for more details. %, we should have written $V[1]$ instead of $V$  and $ R_\Theta[-1]$ instead of $ R_\Theta$ to emphasize their degree, respectively $-1$ and $+1$. For the sake of clarity we chose not to write them here, but we will adopt these conventions in the remainder of this section.
%\end{remark}

\begin{proof}
Given a Lie-Leibniz triple $(\mathfrak{g},V,\Theta)$ in which $V$ is assumed to be a Lie algebra (i.e. the Leibniz product $\circ$ is skew-symmetric), then one defines the following brackets (the first four being those of Equations \eqref{bronze}-\eqref{silver}):
\begin{align}
    \llbracket a,b\rrbracket&=[a,b]_{\mathfrak{g}},  &\llbracket x,y\rrbracket&=0,\\ 
    \llbracket a,x \rrbracket&=a\cdot x, &\llbracket x,a \rrbracket&=-a\cdot x,\\
    \llbracket a,\Theta \rrbracket&=a\cdot\Theta,  &\llbracket \Theta,a \rrbracket&=-a\cdot\Theta \label{def:bracket1},\\
    \llbracket \Theta,x\rrbracket&=\Theta(x) \hspace{2cm}\text{and}\hspace{-3cm}
        &\llbracket x,\Theta\rrbracket&=\Theta(x), \label{def:bracket2}
\end{align}
for every $a,b\in\mathfrak{g}$ and $x,y\in V[1]$. 
 Since $R_\Theta$ is generated by the successive actions of $\mathfrak{g}$ on $\Theta$, one extends the above brackets to the whole of $R_\Theta[-1]$ by the following formulas:
\begin{align}
 \llbracket\Theta_{a_1a_2\ldots a_m},u\rrbracket&=a_1\cdot \llbracket \Theta_{a_2\ldots a_m},u\rrbracket-\llbracket \Theta_{a_2\ldots a_m},a_1\cdot u\rrbracket\label{imporfg}\\
  \llbracket u,\Theta_{a_1a_2\ldots a_m}\rrbracket&=- (-1)^{|u|} \llbracket\Theta_{a_1a_2\ldots a_m},u\rrbracket \label{imporfgbis}\\
 0&=\llbracket R_\Theta[-1],R_\Theta[-1]\rrbracket\label{imporfg2}
\end{align}
for any homogeneous element $u\in\mathfrak{g}\oplus V[1]$ of degree $|u|$, and where $\Theta_{a_1a_2\ldots a_m}=a_1\cdot(a_2\cdot (\ldots(a_m\cdot\Theta)\ldots))$ for any $a_1,a_2, \ldots, a_m\in\mathfrak{g}$.
Moreover setting $\partial_{1}=-\eta(-,\Theta)$ and $\partial_0=\Theta$, one has $\partial_1(\partial_0(x))=0$ by Equation \eqref{eq:equiv2}. Actually,  Lines \eqref{def:bracket1} and \eqref{def:bracket2}  imply that the induced differential $\partial$ coincides with the adjoint action of $\Theta$:
\begin{equation}\label{earageag}
\partial=\llbracket \Theta,-\rrbracket
\end{equation}

%the identities on Lines \eqref{def:bracket1} and \eqref{def:bracket2} are sufficient to extend the bracket to any element of $R_\Theta[-1]$.

 Using the properties of Lie-Leibniz triples, one may check that the above brackets and differentials satisfy the graded Jacobi identities and Leibniz rule. We have already shown several of them: the Jacobi identity on $\mathfrak{g}\oplus V[1]$ is Equation \eqref{compatdiff3}, while the Leibniz rule for any pair of element $x\in V[1]$ and $a\in\mathfrak{g}$ can be found in Equations \eqref{compatdiff} and \eqref{compatdiff2}. We will now show that their validity extends to $R_\Theta[-1]$. In the following, we assume that $a,b,a_1,\ldots,a_m\in\mathfrak{g}$, and $x,y\in V$. We first turn to the Leibniz rules.

Using Equation \eqref{imporfg} for $\llbracket\Theta_{a},b\rrbracket$ and reorganizing the terms, one has: %recalling that the linear map $\eta:\mathfrak{g}\to \mathrm{End}(R_\Theta)$ is a morphism of Lie algebras, we first have:
 \begin{align}
%\Theta\big(\llbracket  a,b\rrbracket\big)&= -\eta([a,b]_{\mathfrak{g}},\Theta)=-\eta\big(a,\eta(b,\Theta)\big)+\eta\big(b,\eta(a,\Theta)\big)\\
\partial\big(\llbracket a,b\rrbracket\big)=\llbracket \Theta, [a,b]_{\mathfrak{g}}\rrbracket&=a\cdot\llbracket \Theta,b\rrbracket-\llbracket\Theta_{a},b\rrbracket\\
&=\llbracket a,\llbracket \Theta,b\rrbracket\rrbracket+\llbracket \llbracket \Theta,a\rrbracket,b\rrbracket=\llbracket a,\partial(b)\rrbracket+\llbracket \partial(a),b\rrbracket
 \end{align}
 This proves the Leibniz rule for the bracket $\llbracket a,b\rrbracket$. 
 Next, because of the degree, one knows that any Leibniz rule involving an element of $R_\Theta[-1]$ and an element of $\mathfrak{g}$ is automatically zero. This leaves us with only one Leibniz rule to check: the one involving one element $x\in V[1]$ and every generator of $R_\Theta[-1]$.
The Leibniz rule between $x$ and $\Theta$ should read:
 \begin{equation}
  \partial\big(\llbracket \Theta,x\rrbracket\big)=-\llbracket\Theta,\partial(x)\rrbracket
 \end{equation}
because $\partial(R_\Theta[-1])=0$ by construction. But Equation \eqref{earageag} implies that the left-hand side is zero, while the right-hand side is $\eta(\partial(x),\Theta)=\Theta(x)\cdot\Theta$ which also vanishes by the quadratic constraint \eqref{eq:equiv2}. This proves that the Leibniz rules is satisfied on the pair of elements $x$ and $\Theta$. 

It is now sufficient to prove it on any other choice of generator $\Theta_{a_1\ldots a_m}$ of $R_\Theta$.  While one can deduce from Equation \eqref{imporfg} that $\partial\big(\llbracket \Theta_{a},x\rrbracket\big)=-\llbracket\Theta_{a},\partial(x)\rrbracket$, the same does not hold for generators of $R_\Theta[-1]$ of the form $\Theta_{a_1\ldots a_m}$ for $m\geq2$ (see the discussion following Equation \eqref{eq:imporg}). Then, as a further assumption, one is allowed to additionally require that:
 \begin{equation}\label{cannotbe}
 \partial\big(\llbracket \Theta_{a_1\ldots a_m},x\rrbracket\big)=-\llbracket\Theta_{a_1\ldots a_m},\partial(x)\rrbracket
 \end{equation}
This is precisely the Leibniz rule for $x$ and $ \Theta_{a_1\ldots a_m}$.  One concludes that the Leibniz rule holds for every pair of elements taken from $V[1]$ and $R_\Theta[-1]$, and thus on the whole of $V[1]\oplus \mathfrak{g}\oplus R_\Theta[-1]$.
% Because of Equations \eqref{earageag} and \eqref{imporfg2}, this is actually the Leibniz rule involving one element of $R_\Theta[-1]$ and one element of $V[1]$.

 We can now turn to proving the remaining Jacobi identities. We already know by Equation \eqref{compatdiff3} and because $\mathfrak{g}$ is a Lie algebra that the graded Jacobi identity is satisfied on $V[1]\oplus\mathfrak{g}$. Moreover it is trivially satisfied on $R_\Theta[-1]$ because the bracket is zero on this space. Then we need only check the cases for which one or two terms belong to $R_\Theta[-1]$ and the other one or two belong to $V[1]\oplus\mathfrak{g}$. The proof is made by induction and, in order to avoid unnecessary technical computation at this stage of the paper, we postpone it to the second part of the proof of Proposition \ref{prop2} (see in particular the discussion below Equation  \eqref{gradedjac}). This concludes the proof of Proposition \ref{experimental}.
\end{proof}

%We first notice that, for any Lie-Leibniz triple $(\mathfrak{g},V,\Theta)$ in which $V$ is a Lie algebra,  the differential graded Lie algebra associated to it by the construction given in the present section will be the following 3-term differential graded Lie algebra:
%\begin{center}
%\begin{tikzcd}[column sep=1cm,row sep=0.4cm]
%V\ar[r,"\Theta"]&\mathfrak{g}\ar[r,"-\eta(-;\Theta)"]& R_\Theta
%\end{tikzcd}
%\end{center}
 %\tjg{why $-\eta$ ?}\textcolor{blue}{so that we have $\partial\llbracket x,y\rrbracket=\llbracket \partial x,y\rrbracket+(-1)^{|x|}\llbracket x,\partial y\rrbracket$}
%where $ R_\Theta\subset V^*\otimes\mathfrak{g}\simeq\mathrm{Hom} (V,\mathfrak{g})$ is the $\mathfrak{g}$-module to which the embedding tensor $\Theta$ belongs, and where $\eta:\mathfrak{g}\to\mathrm{End}(R_\Theta)$ is the corresponding action of $\mathfrak{g}$. 

This result shows that the assignment to a Lie algebra crossed module
of a differential graded Lie algebra can be extended to specific cases of Lie-Leibniz triples: those for which the Leibniz algebra is a Lie algebra, without further assumption on the $\mathfrak{g}$-equivariance of the embedding tensor. 
One objective of the paper is to show that one can %would like to 
generalize Proposition \ref{experimental} to every Lie-Leibniz triples. %the above observation that Lie algebra crossed modules are 3-term differential graded Lie algebras concentrated in degrees $-1$, $0$ and $+1$ (with additional properties that will be clarified later).
 However,  since Lie-Leibniz triples involve Leibniz algebras and not just Lie algebras, we expect new intricate structures to emerge.  In particular, one would expect to 
associate a bigger differential graded Lie algebra to every Lie-Leibniz triple $(\mathfrak{g},V,\Theta)$ whose  Leibniz algebra $V$ is not a Lie algebra: %, which wreduces to the above construction for Lie algebra crossed modules. % This differential graded Lie algebra %(dgLa for short)% would certainly be concentrated in degrees $\leq0$.
%Moreover, we will show that our construction  provides us with an injective function $G:\textbf{Lie-Leib}\to \textbf{DGLie}\leq1$, which becomes a functor when restricted to $\textbf{semLie-Leib}$.
%In the presence of a symmetric component in the Leibniz product, we will show that the cochain complex of Proposition \ref{experimental} extends to the left:
\begin{center}
\begin{tikzcd}[column sep=1cm,row sep=0.4cm]
\ldots\ar[r]&V[1]\ar[r,"\Theta"]&\mathfrak{g}\ar[r,"-\eta(-;\Theta)"]& R_\Theta[-1]
\end{tikzcd}
\end{center}
More precisely, let $\overline{G}$ be the functor associating to each Lie algebra crossed module the corresponding unique 3-term differential graded Lie algebra:
\begin{align*}
\overline{G}:\hspace{0.2cm}\textbf{Lie$\times$Mod}&\xrightarrow{\hspace*{1.2cm}} \hspace{0.4cm}\textbf{DGLie}_{\leq1}\\
	\big(\mathfrak{c}\xrightarrow{\Theta}\mathfrak{g}\big)&\xmapsto{\hspace*{1.2cm}} \big(\mathfrak{c}\xrightarrow{\Theta}\mathfrak{g}\xrightarrow{0}\mathbb{R}[-1]\big)
\end{align*}
%that corresponds to it 
%(where $\mathbb{R}$ is considered to be in degree $+1$), 
This functor canonically extends to the category of Lie-Leibniz triples for which $V$ is a Lie algebra -- this is the content of Proposition \ref{experimental}. Moreover, it turns out that this functor satisfies Theorem \ref{centraltheorem}, %has the following important property,
 which the rest of the paper is dedicated to proving.

 We will show that, in the presence of a symmetric component in the Leibniz product of a Lie-Leibniz triple $(\mathfrak{g},V,\Theta)$, the cochain complex of Proposition \ref{experimental} actually extends to the left:
  %one can extend the above complex to the right as a cochain complex of $\gg$-modules, so that no information is lost:
\begin{equation}\label{diagram2}
\begin{tikzcd}[column sep=1cm,row sep=0.4cm]
\ldots\ar[r,"\partial_{-3}"]&T_{-3}\ar[r,"\partial_{-2}"]&T_{-2}\ar[r,"\partial_{-1}"]&V[1]\ar[r,"\Theta"]&\mathfrak{g}\ar[r,"-\eta(-;\Theta)"]& R_\Theta[-1]
\end{tikzcd}
\end{equation}
The existence of a tower of spaces $T_\bullet=\bigoplus_{i\geq1}T_{-i}$ beyond degree $-1$ is thus the sign that the Leibniz algebra $V$ is \emph{not} a Lie algebra.
The $\gg$-module structure of $V$ is key in defining $T_{-2}$, as it is used in its very definition. 
%\emph{e.g.} the first step, when one defines $T_{-2}$.
%Given this cochain complex, the first remark is that the Lie bracket between  $\mathfrak{g}$ and any
%other $T_{-i}$ would certainly reduce to the action of $\mathfrak{g}$ on $T_{-i}$. Thus the cochain complex $T_\bullet$ is a cochain complex of $\mathfrak{g}$-modules.
%This is used in
%will justify
%\emph{e.g.} the first step, when one defines $T_{-2}$. 
The other $T_{-i}$ for $i\geq3$ are built by induction as quotients of sums of products of other $T_{-j}$ for $1\leq j\leq i-1$ so that the property of being a $\mathfrak{g}$-module is inherited by induction. %Proposition \ref{prophierarchy} shows that the graded vector space $T_\bullet=\bigoplus_{i\geq1}T_{-i}$ can be canonically equipped with a graded Lie algebra structure.
%Actually, t
The action of $\mathfrak{g}$ on $T_{-i}$ is moreover necessary to define the differential on $T_\bullet$, showing the centrality of the Lie algebra $\mathfrak{g}$ in the construction. As proved in Proposition \ref{prop2}, the total space $\mathbb{T}=T_\bullet\oplus\mathfrak{g}\oplus R_\Theta[-1]$ then canonically inherits a differential graded Lie algebra structure, which restricts to that of Proposition \ref{experimental} when $V$ is a Lie algebra. %This  Proposition \ref{prop2} appears as a straightforward generalization of Proposition \ref{experimental} to any Lie-Leibniz triple. %We will first show how to construct the graded vector space $T_\bullet$ in Subsection \ref{construction}, and then show how to construct the differential graded Lie algebra structure on $\mathbb{T}=T_\bullet\oplus\mathfrak{g}\oplus R_\Theta[-1]$ in Subsections \ref{structure} and \ref{patching}, thus defining the (injective-on-objects) function $G:\textbf{Lie-Leib}\to\textbf{DGLie}_{\leq1}$. Finally, Subsection \ref{restriction} is devoted to showing that the restriction of $G$ to the subcategory \textbf{compLie-Leib} is a faithful functor.

%Here, the space $T_{+1}=R_\Theta[-1]$ is understood as the space $R_\Theta$, shifted by  degree $+1$.

This construction improves and simplifies the one presented in \cite{lavau:TH-Leibniz}. In that paper, the author associated a differential graded Lie algebra $\big((T_{-i})_{i\geq0}, \partial, [\,.\,,.\,]\big)$ to any Lie-Leibniz triple such that $T_{-1}=V[1]$, $T_0=\mathfrak{h}=\mathrm{Im}(\Theta)$ and $\partial_0=\Theta$:
\begin{equation}\label{diagram1}
\begin{tikzcd}[column sep=1cm,row sep=0.4cm]
\ldots\ar[r,"\partial_{-3}"]&T_{-3}\ar[r,"\partial_{-2}"]&T_{-2}\ar[r,"\partial_{-1}"]&V[1]\ar[r,"\Theta"]&\mathfrak{h}
\end{tikzcd}
\end{equation}
Here, the  spaces $T_{-i}$ and the linear maps $\partial_{-i}$ are the same as in the cochain complex \eqref{diagram2}.
The construction presented in \cite{lavau:TH-Leibniz} used a dual point of view which implied rather cumbersome computations and had the serious drawback that the resulting differential graded Lie algebra lost important information about the Lie algebra $\mathfrak{g}$ (for example it did not appear in the resulting cochain complex \eqref{diagram1}). In the present paper, we propose to offer a novel and direct construction of the cochain complex \eqref{diagram2}, drawing on theoretical work done by physicists, among others \cite{Cederwall:2015oua, Gomis:2018xmo}. Moreover, we shall prove that each map $\partial_{-i}$, as an element of $\mathrm{Hom}(T_{-i-1},T_{-i})$ on which $\mathfrak{g}$ acts, generates a cyclic $\mathfrak{g}$-submodule of $\mathrm{Hom}(T_{-i-1},T_{-i})$ which is isomorphic to (a sub-module of) $R_\Theta$. %belongs to the same representation $ R_\Theta$ of $\Theta$. 
This result, although having been conjectured and used for a long time by physicists, is new and is needed precisely in order to extend the cochain complex \eqref{diagram1} to the right as in diagram \eqref{diagram2}.

%The underlying (differential) graded Lie algebra structure that would equip this cochain complex corresponds to  the \emph{tensor hierarchy algebra} 
%so that this is a cochain complex of $\gg$-modules. 

\section{The tensor hierarchy associated to a Lie-Leibniz triple}\label{section3}
%\gg-modules}

%A construction of the graded differential graded Lie algebra represented by diagram \eqref{diagram1} has been given in \cite{lavau:TH-Leibniz}. However, we  here provide a much  simpler and more straightforward construction that slightly generalize it, allowing us to extend it to the right as in diagram \eqref{diagram2}. This extension is not trivial, as it requires to carefully extend the graded Lie bracket and the differential so that they still satisfy the Leibniz rule. This problem is rather subtle, which is the reason why it was not devised before. 

Subsection \ref{construction} is dedicated to defining the particular graded vector space  $T_{\bullet}$ and equip it with a graded Lie bracket. Subsection \ref{structure} is dedicated to defining the differential, while subsection \ref{patching} proves the compatibility of the bracket with the differential, turning the total space $\mathbb{T}=T_\bullet\oplus\mathfrak{g}\oplus R_\Theta[-1]$ into a differential graded Lie algebra. This algebra is the one defining the function $G$ in Theorem \ref{centraltheorem}.  We refer the reader to Example \ref{lodayexample2} in order to fully appreciate the details of the first steps of the construction of the algebra.
Finally, subsection~\ref{restriction} is devoted to showing  that the restriction of  $G$ to the subcategory \textbf{compLie-Leib} is a faithful functor.
%The present discussion draws on theoretical work done by physicists, among others \cite{Cederwall:2015oua, Gomis:2018xmo}. 

\subsection{Associating a graded Lie algebra to a Lie-Leibniz triple}\label{construction}

% and the mathematics of 
 %Loday and  Pirashvili \cite{Loday-Pirashvili}.
 The main result of the current subsection shows that any Lie-Leibniz triple gives rise to a graded Lie algebra generalizing the $\mathfrak{g}$-module $\mathfrak{c}$ that appears in any Lie algebra crossed module $\mathfrak{c}\xrightarrow{\Theta}\mathfrak{g}$:

\begin{prop}\label{prophierarchy}
To any Lie-Leibniz triple $(\mathfrak{g},V,\Theta)$ one can associate a negatively graded Lie algebra $\big(T_\bullet=(T_{-i})_{i\geq1},\llbracket\,.\,,.\,\rrbracket\big)$, which satisfies the  following three conditions:
\begin{enumerate}
    \item $T_\bullet$ is a graded $\mathfrak{g}$-module, and the corresponding representation $\rho:\mathfrak{g}\to\mathrm{End}(T_\bullet)$ takes values in $\mathrm{Der}(T_\bullet)$, the space of  derivations of $\llbracket\,.\,,.\,\rrbracket$; %, so that the bracket between $\mathfrak{g}$ and $T_{-i}$ corresponds to the action of $\mathfrak{g}$ on $T_{-i}$;
   % \item $T_{-1}=V[1]$;
        \item  $T_{\bullet}=V[1]\oplus \llbracket T_{\bullet}, T_\bullet\rrbracket$;
        %\item $\llbracket\,.\,,.\,\rrbracket=0$ if and only if $V$ is a Lie algebra.
    \item $T_{\bullet}=T_{-1}=V[1]$ if and only if $V$ is a Lie algebra.
\end{enumerate}
\end{prop}

The first item of the proposition means that for every $i\geq1$, $T_{-i}$ is a $\mathfrak{g}$-module under the action of a Lie algebra morphism $\rho_{-i}:\mathfrak{g}\to \mathrm{End}(T_{-i})$, so that $\rho$ is the unique (graded) Lie algebra morphism restricting to $\rho_{-i}$ on $\mathrm{End}(T_{-i})$. Moreover for every $a\in\mathfrak{g},x\in T_{-i}$ and $y\in T_{-j}$, $\rho(a;-)$ is a derivation of the graded Lie bracket~$\llbracket x,y\rrbracket$:
\begin{equation}
\rho_{-i-j}\big(a;\llbracket x,y\rrbracket\big)=\big\llbracket\rho_{-i}(a;x),y\big\rrbracket+\big\llbracket x,\rho_{-j}(a;y)\big\rrbracket
\end{equation}
%In particular, $\rho$ is a degree 0 map.
The second item  can be reformulated as follows: $T_{-1}=V[1]$ and for every $i\geq2$, $T_{-i}=\sum_{j=1}^{i-1}\,\llbracket T_{-j},T_{-i+j}\rrbracket$. Eventually, the last item is the condition that
 $T_{-i}=0$ for every $i\geq2$ if and only if $V$ is a Lie algebra. 
In other words, when $V$ is a Lie algebra, the tower of vector spaces reduces to $T_\bullet=T_{-1}=V[1]$ and the graded Lie bracket $\llbracket\,.\,,.\,\rrbracket$ is identically zero, whereas when $V$ is a Leibniz algebra whose Leibniz product is not fully skew-symmetric, then the spaces of lower degrees exist and satisfy item 2. of Proposition \ref{prophierarchy}.

The proof is split in two parts.
First, we will define the $T_{-i}$ recursively, starting by setting $T_{-1}=V[1]$, and then by  taking the quotient of the free graded Lie algebra generated by $V[1]$ by a particular graded ideal $K_\bullet=\bigoplus_{i\geq2}K_{-i}$ which is constructed by induction. The quotient then is a (possibly not bounded) graded vector space $T_\bullet=\bigoplus_{i=1}^\infty T_{-i}$ 
which has the following properties:
\begin{enumerate}
    \item Every vector space $T_{-i}$ is a $\mathfrak{g}$-module;
    \item $T_{-1}=V[1]$;
    \item $T_{-i}=0$ for every $i\geq2$ if and only if $V$ is a Lie algebra.
\end{enumerate}
This represents the biggest part of the proof, and it crucially relies on key Lemmas \ref{lemmaexact} and \ref{snake}.
Then, we will define a graded Lie bracket on $T_\bullet$, so that items 1. and 2. of Proposition \ref{prophierarchy} are satisfied by construction.

\begin{remark}
The idea behind the proof (quotienting a graded ideal out of a free graded Lie algebra) in the context of building a tensor hierarchy has first been given in \cite{Cederwall:2015oua}, where the authors use a partition function in order to compute the $\mathfrak{g}$-modules that form the building blocks of the hierarchy. On the other hand, our proof relies on an inductive argument that is algebraic and that enables us to make explicit the definition of the graded Lie bracket. The relationship between the graded Lie algebra structure presented in \cite{Cederwall:2015oua} and ours has still to be investigated, but we expect it to be very close if not identical.
\end{remark}

 Let us first give a brief account of our conventions and a technical reminder of graded algebra.
Given a finite dimensional vector space $X$, one can define the free Lie algebra $\mathrm{Free}(X)$  as the Lie algebra generated by  $X$, imposing only the defining relations of skew-symmetry
%alternating K-bilinearity 
and the Jacobi identity. It can be thought of 
as the vector space generated by the successive commutators of a basis of $X$.  
The vector space $F_\bullet=\mathrm{Free}(X)$ is graded by the length of the commutators, e.g. 
%\footnote{A more natural description (basis independent) is to \emph{filter} $F =\mathrm{Free}(U)$ as $$F\supset [F,F] \supset [F,[F,F]] \supset$$ and pass to the associated graded.}
%\tgr{ which I would like to add}
$[\cdots [x_1, x_2], x_3], \cdots x_i]$, for any $x_1,\ldots,x_i\in X$, has length $i$. There exist various  formulas for a basis of $F_\bullet=\bigoplus_{k\geq1}F_k$ \cite{Hall, lyndon, shirshovBasesFreeLie2009}. % and any basis $(w^k_{l})_{1\leq l\leq \mathrm{dim}(W_{k})}$ of $W_k$ is graded by the number of elements  in $X$ used when the $w^k_{l}$ are expressed as linear combinations of $k$ basis elements.
The construction also works for graded vector spaces $X$, hence defining  free graded Lie algebras, also called free Lie superalgebras when the emphasis is  on the parity \cite{chibrikov, shtern}. %\emph{superalgebra}.

First, let us recall that 
for a graded vector space $X$, %$\Lambda(X)$ denotes 
the free graded skew-commutative algebra generated by  $X$ is denoted by $\Lambda^\bullet(X)$. That is, the algebra spanned by elements $x_1\wedge  x_2\wedge\cdots\wedge x_n$.  
With the degree of $x$ denoted $|x|$, we have:
\begin{equation}\label{eqantisym}
x\wedge y = - (-1)^{|x||y|} y\wedge x
%=  -(-1)^{|u||v|} v\wedge u.
\end{equation}
%The free graded symmetric algebra generated by $U$ is denoted by $S(U)$ and uses the reverse convention:
%\begin{equation}
%u\odot v =  (-1)^{|u||v|} v\odot u.
%=  -(-1)^{|u||v|} v\wedge u.
%\end{equation}
%where $\odot$ symbolizes the symmetric product.
For example, pick  a vector space $V$ and define  $X=V[1]$ to be its desuspension, that is to say: $X$ is isomorphic to $V$ as a vector space, but its elements are considered to carry a grading defined to be $-1$. Then, due to the shifted degree and Equation \eqref{eqantisym},  the vector space $\Lambda^2(X)$ is isomorphic  to $S^2(V)$ (understood in the category of vector spaces).
%\tgr{Need to reconcile notation above}
%denotes the component in ?total? degree $-k$.
%\trd{The following may not be needed.}Let $i:Free( ) \to \mathfrak{T}$ restrict to $i_n: Free( )_n\subset \mathfrak{T}_n$,as in Periashvili. His main observation  is the fact that if $\gg$ is a Leibniz algebra, then $L(\gg)_*$ is closed under the Loday boundary map, and therefore $L(\gg)_*$ is a subcomplex of ${CL}_*(\gg)$. In order to describe the induced boundary map$\dd:L(\gg)_n\to L(\gg)_{n-1}$, we need the linear map $p_n:T(V,1)_n\to L(V,1)_n$ given by$$p_n(x_1\gg\otimes\cdots\gg\otimes x_n)=[[x_1,\cdots ,x_n]].$$
%\begin{Prop}\label{komplexi} Let $\gg$ be a Leibniz algebra. Then, for any element $\omega\in L(\gg)_{n+1}$ one has $$di_{n+1}(\omega)=(-1)^np_n\circ f^n\circ i_{n+1}(\omega),$$ where
%$$f^n: \gg^{\otimes n+1}\to \gg^{ \otimes n}$$
%is the linear map given by
%$$f^n(x_1 \otimes\cdots \otimes x_n \otimes x_{n+1})=x_1\gg\gg\otimes \cdots \gg\gg\otimes x_{n-1}\gg\gg\otimes \{x_n,x_{n+1}\}.$$\end{Prop}
%\end{remark}
%\tgr{Need to say something about $\Lambda^2(F_{-1})$ being a special case to be distinguished from
%$\Lambda^\bullet(F_\bullet)$ as in CE} \textcolor{blue}{$\Lambda^\bullet(F_\bullet)=F_{-1}\oplus \Lambda^2(F_{-1})\oplus F_{-2}\oplus \Lambda^3(F_{-1})\oplus F_{-1}\otimes F_{-2}\oplus F_{-3}\oplus ...$ so one is part of the other!}

Thus, the free graded 
Lie algebra $F_\bullet=\mathrm{Free}(X)$ generated by the degree $-1$ vector space $F_{-1}=X=V[1]$  admits $\Lambda^2(X)\simeq S^2(V)$ as the desired space $F_{-2}$, where $\simeq$ indicates that the two spaces are isomorphic.
Notice that for $x,y\in F_{-1}$, we have $[x,y]= [y,x]$, where here the bracket is the free graded Lie bracket on $F_\bullet$.
%freely generated by the commutators of two elements of $F_{-1}$ is then $\Lambda^2(F_{-1})$.
The next vector space $F_{-3}$ has
a basis 
%freely generated by 
of some elements of the form $[[x,y],z]$ for  $x,y,z\in F_{-1}$, subject to the Jacobi identity.
%would thus be 
Thus $F_{-3}$ is isomorphic to the quotient of $F_{-1}\otimes \Lambda^2(F_{-1})\simeq V\otimes S^2(V)$ by $\Lambda^3(F_{-1})\simeq S^3(V)$. More generally, the negative grading on $F_\bullet=\bigoplus_{k\geq1}F_{-k}$ is defined as follows: the number $k\geq1$ labelling $F_{-k}$ indicates the number of elements $x_1,\ldots, x_k\in X=V[1]$ used in the brackets $[\cdots[x_1, x_2], x_3],\ldots],x_k]$  defining the generators of $F_{-k}$.  %The graded vector space $F_\bullet =\mathrm{Free}(V[1])$ is thus graded by bracket length.

In this section, we will consider the total exterior algebra  $\Lambda(F_\bullet)=\bigoplus_{p=0}^\infty \Lambda^p F_\bullet$ of a free graded Lie algebra of the form $F_\bullet=\mathrm{Free}(V[1])$. The  exterior algebra inherits a grading from $F_\bullet$ which is
%Notice $\Lambda^*(F_{*})$ 
%which is bigraded, for 
such that the component in total degree $-i$ is denoted $\Lambda(F_\bullet)_{-i}= \bigoplus_{p\geq0} \Lambda^p(F_\bullet)|_{-i}$. The space $\Lambda(F_\bullet)_{-i}$ at degree $-i$  is a finite dimensional graded vector space.
This exterior algebra is obtained as the quotient of the tensor algebra $T(F_\bullet)$ by the (graded) ideal generated by elements of the form:
\begin{equation}
x\otimes y +(-1)^{|x||y|}y\otimes x
\end{equation}
for every two homogeneous elements $x,y\in F_\bullet$. The resulting quotient is isomorphic to a graded sub-vector space of the tensor algebra, and is a graded algebra on its own, when equipped with the wedge product.

Let us discuss the first few spaces appearing in the decomposition of $\Lambda^2(F_\bullet)$. As a vector space, $\Lambda^2(F_\bullet)|_{-2}$ is canonically isomorphic to $F_{-1}\wedge F_{-1}$, while $\Lambda^2(F_\bullet)|_{-3}$ is canonically isomorphic to $F_{-1}\otimes F_{-2}$, whereas $\Lambda^2(F_\bullet)|_{-4}$ is canonically isomorphic to $F_{-1}\otimes F_{-3}\oplus F_{-2}\wedge F_{-2}$. More generally, $\Lambda^2(F_\bullet)_{-i}$ (for $i\geq3$) is canonically isomorphic to the following vector spaces depending on the parity of $i$:
\begin{align}
&\text{if $i$ is even} &F_{-\frac{i}{2}}\wedge F_{-\frac{i}{2}}\quad\oplus \bigoplus_{1\leq j\leq\frac{i-2}{2}}F_{-j}\otimes F_{-i+j}\\
&\text{if $i$ is odd} &\hspace{2cm}\bigoplus_{1\leq j\leq\frac{i-1}{2}}F_{-j}\otimes F_{-i+j}
\end{align} When $j\neq \frac{i}{2}$, what we denote by $F_{-j}\wedge F_{-i+j}$ is the vector space isomorphic to $F_{-j}\otimes F_{-i+j}$, understood as a subspace of $\Lambda^2(F_\bullet)_{-i}$. In particular,  if $x\in F_{-j}$ and $y\in F_{-i+j}$ then we will write $x\wedge y$ instead of $x\otimes y$ when we want to emphasize that we work in $F_{-j}\wedge F_{-i+j}$.
Regarding $F_{-j}\wedge F_{-i+j}$ as isomorphic to both $F_{-j}\otimes F_{-i+j}$ and $F_{-i+j}\otimes F_{-j}$, we then have $x\wedge y$ represented by $x\otimes y\in F_{-j}\otimes F_{-i+j}$ or equally by $-(-i)^{j(i-j)}y\otimes x\in F_{-i+j}\otimes F_{-j}$.
% and we will consider that $-(-1)^{j(i-j)}y\wedge x$ represent the same element in $F_{-j}\wedge F_{-i+j}$, that can be faithfully represented as $x\otimes y$ in $F_{-j}\otimes F_{-i+j}$.

The Chevalley-Eilenberg homology of a graded Lie algebra $\gg$  is the homology
of their chain complex $CE_\bullet(\gg)$ with $CE_n(\gg)=\Lambda^n\gg$
and differential $d = \{d_n:CE_n(\gg)\to CE_{n-1}(\gg)\}$ given by (\cite{GARCIAMARTINEZ} p. 471):
\begin{equation}\label{CEdiff}
d_n(x_1\wedge\ldots\wedge x_n)=
 \sum_{1\leq i<j\leq n}\ (-1)^{i+j+|x_j|\,\eta_j+|x_i|\,\eta_{i}-|x_i||x_j|}\ [x_i,x_j]\wedge x_1\wedge \ldots\wedge x_{i-1}\wedge x_{i+1}\wedge \ldots\wedge x_{j-1}\wedge  x_{j+1}\wedge \ldots\wedge x_n
%\sum_{\sigma\in\mathrm{Un}(2,n-2)}\epsilon^\sigma\,[x_{\sigma(1)},x_{\sigma(2)}] \wedge x_{\sigma(3)}\wedge\ldots\wedge  x_{\sigma(n)}   
\end{equation}
where $\eta_m=\sum_{1\leq k\leq m-1} |x_{k}|$.
Degrees of the elements appear in the formula because one has to permute $x_i$ and $x_{j}$ to the first and second positions in order to eventually take the Lie bracket of $x_i$ and $x_j$. There exists another, identical but more compact formula, which uses unshuffles:
\begin{equation}\label{CEdiff2}
  d_n(x_1\wedge\ldots\wedge x_n)=  -\sum_{\sigma\in\mathrm{Un}(2,n-2)}\epsilon^\sigma\,[x_{\sigma(1)},x_{\sigma(2)}] \wedge x_{\sigma(3)}\wedge\ldots\wedge  x_{\sigma(n)} 
\end{equation}
where $\mathrm{Un}(2,n-2)$ are the $(2,n-2)$-unshuffles, and where $\epsilon^\sigma$ is the sign that is induced by the permutation $\sigma$ on $x_1\wedge\ldots\wedge x_n$. More precisely, a permutation $\sigma\in S_n$ is called a $(p,n-p)$-\emph{unshuffle} if $\sigma(1)<\ldots<\sigma(p)$ and 
$\sigma(p+1)<\ldots<\sigma(n)$, and $\epsilon^\sigma$ is defined through the following formula: 
\begin{equation}\label{eqkoszul}
x_1\wedge\ldots\wedge x_n=\epsilon^\sigma x_{\sigma(1)}\wedge\ldots\wedge x_{\sigma(n)}
\end{equation} % the product of the Koszul sign of the permutation $\sigma$ with its signature, that is: $x_1\wedge\ldots\wedge  x_n=\epsilon^\sigma\,x_{\sigma(1)}\wedge \ldots\wedge x_{\sigma(n)}$. So for example, for  $\sigma=(1\,2)$, we have  $\epsilon^{(1\,2)}=-(-1)^{|x_1||x_2|}$.  
One may check that Equation \eqref{CEdiff} is the same as Equation \eqref{CEdiff2}, and that they correspond to the common definition of the Chevalley-Eilenberg differential (in homology) when $\mathfrak{g}$ is concentrated in degree 0 (see \cite{knappLieGroupsLie1988} p. 284).
%\tgr{Even simpler is the original with
%\begin{equation}
%d_n(x_1\wedge\ldots\wedge x_n)=
%\sum_{1 \leq n}\epsilon^\sigma\,[x_{i},x_{j}] \wedge\ldots \hat x_i \wedge \ldots\wedge \hat x_n 
%\end{equation}}
Then the following result holds:
\begin{lemma}\label{lemmaexact}
For a free graded Lie algebra  $F_\bullet=(F_{-i})_{i\geq1}$, %then $H_1(F_\bullet)=H_2(F_\bullet)=0$, and this translates as 
the following sequences are exact:
\begin{equation}\label{exactsequence1}
\begin{tikzcd}[column sep=0.7cm,row sep=0.4cm]
0\ar[r]&\Lambda^2(F_{-1})\ar[r,"d_2"]&F_{-2}\ar[r]&0
\end{tikzcd}
\end{equation}
\begin{equation}\label{exactsequence2}
\begin{tikzcd}
\Lambda^3(F_{-1})\ar[r,"d_3"]&F_{-1}\wedge F_{-2}\ar[r,"d_2"]&F_{-3}\ar[r]&0%\\
%\Lambda^3(F_{-1}\oplus F_{-2})|_{-4}\ar[r]&F_{-1}\otimes F_{-3}\,\oplus \,\Lambda^2(F_{-2})\ar[r,"\varphi_{-4}"]&F_{-4}\ar[r]&0%\\
%\Lambda^3\big(\bigoplus_{j=1}^3 F_{-j}\big)|_{-5}\ar[r]&\Lambda^2\big(\bigoplus_{j=1}^4 F_{-j}\big)|_{-5}\ar[r, "\varphi_{-5}"]&F_{-5}\ar[r]&0
\end{tikzcd}
\end{equation}
and more generally, for every $i\geq4$:
\begin{equation}
\begin{tikzcd}[column sep=0.7cm,row sep=0.4cm]
\Lambda^3\big(\bigoplus_{j=1}^{i-2} F_{-j}\big)\big|_{-i}\ar[r, "d_3"]&\Lambda^2\big(\bigoplus_{j=1}^{i-1} F_{-j}\big)\big|_{-i}\ar[r, "d_2"]&F_{-i}\ar[r]&0 \label{exactsequence3}
\end{tikzcd}
\end{equation}
\end{lemma}
\begin{proof}
For a free graded Lie algebra $F_\bullet$, one has $H_1(F_\bullet)=0$ and $H_2(F_\bullet)=0$. The first equality is justified because on the one hand, the map $d_1:F_\bullet\to \mathbb{R}$ is the zero map and, on the other hand, $d_2:\Lambda^2(F_\bullet)\to F_{\leq-2}$ is by definition surjective: 
%, i.e. that $H_1(F_\bullet)=0$, 
 $F_\bullet$ is a free graded Lie algebra so each  $F_{-i}$ (for $i\geq2$) is spanned by (iterated) brackets. The second equation, $H_2(F_\bullet)=0$, is deduced from Corollary 6.5 in \cite{GARCIAMARTINEZ} by keeping track of the natural graduation of $F_\bullet$. % induced by that of X
  Then it suffices to write 
%the equations $H_1(F_\bullet)=0$ and $H_2(F_\bullet)=0$
the equations $H_1(F_\bullet)=0$ and $H_2(F_\bullet)=0$ degree-wise, starting at $n=2$, to obtain the exact sequences \eqref{exactsequence1}, \eqref{exactsequence2} and \eqref{exactsequence3}. 
%The rightmost column is the image of
%Note $d_1:F_\bullet\to\mathbb{K}$  is zero by definition.
\end{proof}

The proof of Proposition \ref{prophierarchy} then relies on the use of Lemma \ref{lemmaexact} in conjunction with the following Lemma: %, whose statement and proof has been brought to our knowledge by Jonathan Wise:
%Then we need to introduce the following Lemma:

\begin{lemma}\label{snake}
Let $\mathcal{A}$ be an abelian category and assume that the following diagram in $\mathcal{A}$ is commutative:
\begin{center}
\begin{tikzcd}[column sep=1cm,row sep=1cm]
0\ar[r] &\ar[d,"u"]A\ar[r,"\mathrm{id}"]&\ar[d,"k"]A\ar[r]&0\\
W\ar[r,"a"]&\ar[d]X\ar[r,"b"]&\ar[d]Y\ar[r,"c"]&Z\\
W\ar[r,"\alpha"]&\bigslant{X}{u(A)}\ar[r,"\beta"]&\bigslant{Y}{k(A)}\ar[r,"\gamma"]&Z
\end{tikzcd}
\end{center}
where $\alpha, \beta$ and $\gamma$ are the linear maps canonically induced from $a,b$ and $c$. Then exactness of the second row implies exactness of the third. %, with exact rows. Then the induced sequence:
%\begin{center}
%\begin{tikzcd}[column sep=1cm,row sep=0.4cm]
%V\ar[r,"a"]&\bigslant{W}{u(Z)}\ar[r,"b"]&\bigslant{X}{b(u(Z))}\ar[r,"c"]&Y
%\end{tikzcd}
%\end{center}
%is exact.
\end{lemma}

\begin{remark}
Here, given a morphism $\chi\in\mathrm{Hom}_{\mathcal{A}}(P,Q)$, the quotient notation $\bigslant{P}{\chi(Q)}$ denotes the cokernel of the map $\chi$, that is: $\bigslant{P}{\chi(Q)}:=\mathrm{coker}(\chi:P\to Q)$. In the category of $R$-modules, this reduces to the usual quotient of $R$-modules. %Hence Lemma \ref{snake} amounts to saying that taking cokernels is a right exact functor.
\end{remark}

%\begin{lemma}\label{snake-alt}
%Let $\mathcal{A}$ be an abelian category, and assume that the following diagram in $\mathcal{A}$:
%\begin{center}
%\begin{tikzcd}[column sep=1cm,row sep=1cm]0\ar[r] &\ar[d,"u"]Z\ar[r,"\mathrm{id}"]&\ar[d,"b\circ u"]Z\ar[r]&0\\V\ar[r,"a"]&W\ar[r,"b"]&X\ar[r,"c"]&Y\\V\ar[r,"a"]&\bigslant{W}{u(Z)}\ar[r,"b"]&\bigslant{X}{b(u(Z))}\ar[r,"c"]&Y\end{tikzcd}\end{center}
%is commutative, with the top 2 rows exact. Then the third sequence:
%\begin{center}\begin{tikzcd}[column sep=1cm,row sep=0.4cm]
%V\ar[r,"a"]&\bigslant{W}{u(Z)}\ar[r,"b"]&\bigslant{X}{b(u(Z))}\ar[r,"c"]&Y\end{tikzcd}\end{center}is exact.\end{lemma}

\begin{proof}
Exactness of the sequence:
\begin{center}
\begin{tikzcd}[column sep=1cm,row sep=1cm]
W\ar[r,"a"]&X\ar[r,"b"]&Y\ar[r,"c"]&Z
\end{tikzcd}
\end{center}
is equivalent to the existence of two isomorphisms:
\begin{equation*}
\mathrm{Coker}(a)\simeq\mathrm{Im}(b)\hspace{1cm}\text{and}\hspace{1cm}\mathrm{Coker}(b)\simeq\mathrm{Im}(c)
\end{equation*}
%which in turn is true if and only if there exists an isomorphism between $\mathrm{Coker}(a)$  and $\mathrm{Ker}(c)$.
Now we have the following isomorphisms:
\begin{equation*}
\mathrm{Im}(\gamma)\simeq\mathrm{Im}(c)\simeq \mathrm{Coker}(b)\simeq \mathrm{Coker}(\beta)
\end{equation*}
which implies that the sequence:
\begin{center}
\begin{tikzcd}[column sep=1cm,row sep=1cm]
W\ar[r,"\alpha"]&\bigslant{X}{u(A)}\ar[r,"\beta"]&\bigslant{Y}{k(A)}\ar[r,"\gamma"]&Z
\end{tikzcd}
\end{center}
is exact at $\bigslant{Y}{k(A)}$. Exactness at $\bigslant{X}{u(A)}$ is implied by the following sequences of isomorphisms:
\begin{equation*}
\mathrm{Coker}(\alpha)\simeq\mathrm{Coker}\big(A\to\mathrm{Coker}(a)\big)\simeq\mathrm{Coker}\big(A\to\mathrm{Im}(b)\big)\simeq\mathrm{Im}(\beta)
\end{equation*}
where the first and last equalities are obtained by definition of the maps $\alpha$ and $\beta$. This concludes the proof.
 \end{proof}

 \begin{proof}[Proof of Proposition \ref{prophierarchy}]
The first part of the proof relies on
constructing  the graded vector space $T_\bullet$ %underlying the cochain complex:
%\begin{center}
%\begin{tikzcd}[column sep=1cm,row sep=0.4cm]
%\ldots\ar[r,"\partial_{-3}"]&T_{-3}\ar[r,"\partial_{-2}"]&T_{-2}\ar[r,"\partial_{-1}"]&V[1]\ar[r,"\Theta"]&\mathfrak{g}\ar[r,"-\eta(-;\Theta)"]& R_\Theta[-1]
%\end{tikzcd}
%\end{center}
by taking the quotient of  the free graded Lie algebra $F_\bullet=\mathrm{Free}(V[1])$ by a graded ideal $K_\bullet=\bigoplus_{i\geq2} K_{-i}$. More precisely, for every $i\geq2$:
$$T_{-i}=\bigslant{F_{-i}}{K_{-i}}$$
%and\footnote{Quotienting by a graded ideal is a little more subtle than in the ungraded case \cite{Weibel?}.}
and  $T_{-1}=V[1]$ of course.
This will be done by using Lemma \ref{snake}, where we  justify that the lines are exact by invoking Lemma \ref{lemmaexact}.
%To keep the analogy with Lie algebra crossed modules where $\mathfrak{c}$ is a $\mathfrak{g}$-module, we want $T_\bullet$ to be a graded $\mathfrak{g}$-module.The choice of ideal will be such that the negatively graded vector space $T_\bullet$ obtained from the quotient operation would satisfy the following condition: when $V$ is a Lie algebra, then $T_\bullet=V[1]$, that is to say: all spaces of degree strictly lower than $-1$ would vanish. However, in the case in which $V$ would be a Leibniz algebra with non trivial symmetric bracket, these spaces  certainly would not vanish. 
%The third item implies in turn the following remark: 

Given this strategy, %the first condition in Proposition \ref{prophierarchy} will be satisfied by choosing a graded ideal $K_\bullet$ of $F_\bullet$ that is also a $\mathfrak{g}$-submodule.
let us first define $T_{-2}$ as a quotient of $\Lambda^2(F_{-1})\simeq S^2(V)$. %Denote by $\mathfrak{h}$ the image of $\Theta$ in $\mathfrak{g}$.
The kernel $\mathrm{Ker}\big(\{\,.\,,.\,\}\big)$ of the symmetric bracket on $V$ is a subspace of $S^2(V)$; as seen from Equation \eqref{hera0}, it is a $\mathfrak{h}$-module but not necessarily a $\mathfrak{g}$-module. However, it admits $\mathfrak{g}$-submodules, and among them one is the biggest, that we denote by $K$. %Hence, when the Lie-Leibniz triple is stringent, $K=\mathrm{Ker}\big(\{\,.\,,.\,\}\big)$.
We then set $T_{-2}$ as the following quotient: 
\begin{equation*}
T_{-2}:=\left(\bigslant{S^2(V)}{K}\right)[2]
\end{equation*}
Here, the desuspension operator $[1]$ is applied twice so that elements of $T_{-2}$ are considered to carry   degree $-2$.
 Identifying $\Lambda^2(F_{-1})$ and $S^2(V)[2]$, and setting $K_{-2}=K[2]$, the vector space $T_{-2}$ can moreover be identified with:
\begin{equation*}
T_{-2}=\bigslant{F_{-2}}{K_{-2}}
\end{equation*}
since $\Lambda^2(F_{-1})$ is by construction canonically isomorphic to $F_{-2}$.
This vector space is a $\mathfrak{g}$-module, inheriting this structure from the canonical quotient map $p:F_{-2}\to T_{-2}$. Recalling that $F_{-1}=T_{-1}=V[1]$, let us define $q_{-2}:\Lambda^2(T_{-1})\to T_{-2}$ to be the following composite map, which is by construction $\mathfrak{g}$-equivariant: \begin{center}
\begin{tikzpicture}
\matrix(a)[matrix of math nodes, 
row sep=5em, column sep=6em, 
text height=1ex, text depth=0.25ex] 
{&F_{-2}\\
\Lambda^2(T_{-1})&T_{-2}\\} ; 
\path[->>](a-2-1) edge node[above]{$q_{-2}$} (a-2-2); %node[belowright]{$\varphi$}  (a-1-2); 
\path[->](a-1-2) edge node[right]{$2\,p$} (a-2-2);
\path[->](a-2-1) edge node[above left]{$-d_2$} (a-1-2);
\end{tikzpicture}
\end{center}

\begin{remark}
  This choice  of $T_{-2}$ satisfies all required conditions: it is a $\mathfrak{g}$-module and it is zero if and only if $V$ is a Lie algebra, since then $\{\,.\,,.\,\}=0$ and $K= S^2(V)$. In particular, it is precisely for this last reason that we did not pick up \emph{any} $\mathfrak{g}$-submodule of $\mathrm{Ker}\big(\{\,.\,,.\,\}\big)$, but \emph{only} the biggest such.
\end{remark}

%twice the quotient map defined above. %That is to say:
%\begin{center}
%\begin{tikzcd}[column sep=1cm,row sep=1cm]
%\Lambda^2(T_{-1})\ar[r,"q_{-1}=2p"]&\bigslant{\Lambda^2(T_{-1})}{K_{-2}}\ar[r]&0
%\end{tikzcd}
%\end{center}
%The relevance of the coefficient $2$ in the definition of $q_{-2}$ will be explained in Subsection \ref{structure}.
%The presence of the factor 2 is a convention that makes certain calculations much simpler and consistent with the definition of maps in subsection \ref{structure}.

Let us now use Lemmas \ref{lemmaexact} and \ref{snake} to define $T_{-3}$. %Since $T_{-1}=F_{-1}=V[1]$, t
Since $K_{-2}$ is a submodule of $F_{-2}$, there is a canonical embedding $u_{-3}:F_{-1}\otimes K_{-2}\to F_{-1}\wedge F_{-2}$.
This makes the following diagram:
\begin{center}
\begin{tikzcd}[column sep=1cm,row sep=1cm]
0\ar[r]&\ar[d,"u_{-3}"] F_{-1}\otimes K_{-2}\ar[r,"\mathrm{id}"]&\ar[d, dashed, "k_{-3}"] F_{-1}\otimes K_{-2}\ar[r]&0\\
\Lambda^3(F_{-1})\ar[r,"-d_3"]&F_{-1}\wedge F_{-2}\ar[r,"-d_2"]&F_{-3}\ar[r]&0
\end{tikzcd}
\end{center}
commutative with exact rows (the bottom row is exact by Lemma \ref{lemmaexact}). We appended a minus sign in front of the linear maps $d_2$ and $d_3$ because the definition of the Chevalley-Eilenberg differential $d_2$ is minus the Lie bracket defining $F_{-3}$ (see Equation \eqref{CEdiff2}). 
The map $k_{-3}:F_{-1}\otimes K_{-2}\to F_{-3}$   is the composite $-d_2\circ u_{-3}$, making this diagram commutative. 
As a composition of $\mathfrak{g}$-equivariant morphisms, the map $k_{-3}$ is itself $\mathfrak{g}$-equivariant, so that $\mathrm{Im}(k_{-3})$
%its image 
defines a $\mathfrak{g}$-submodule $K_{-3}$ of $F_{-3}$.

Since $T_{-1}=F_{-1}$ the quotient of $F_{-1}\wedge F_{-2}$ by $F_{-1}\wedge K_{-2}$ is $T_{-1}\wedge T_{-2}$, so by Lemma \ref{snake} we deduce that:
\begin{equation}
\begin{tikzcd}[column sep=1cm,row sep=1cm]
\Lambda^3(F_{-1})\ar[r]&T_{-1}\wedge T_{-2}\ar[r]&\bigslant{F_{-3}}{K_{-3}}\ar[r]&0
\end{tikzcd}
\end{equation}
is an exact sequence. Here, we adopted the same convention for $T_{-1}\wedge T_{-2}$ as we did for $F_{-1}\wedge F_{-2}$, so that it can be identified with $T_{-1}\otimes T_{-2}$.   We set:
\begin{equation*}
T_{-3}:=\bigslant{F_{-3}}{K_{-3}}
\end{equation*}
and we denote by  $q_{-3}:T_{-1}\wedge T_{-2}\to T_{-3}$ the  map induced from $-d_2|_{F_{-1}\wedge F_{-2}}$ by Lemma \ref{snake}, i.e. the unique map such that the following diagram is commutative, with exact rows:
\begin{center}
\begin{tikzcd}[column sep=1cm,row sep=1cm]
0\ar[r]&\ar[d,"u_{-3}"] F_{-1}\otimes K_{-2}\ar[r,"\mathrm{id}"]&\ar[d, "k_{-3}"] F_{-1}\otimes K_{-2}\ar[r]&0\\
\Lambda^3(F_{-1})\ar[r,"-d_3"]&\ar[d]F_{-1}\wedge F_{-2}\ar[r,"-d_2"]&\ar[d]F_{-3}\ar[r]&0\\
\Lambda^3(F_{-1})\ar[r]&T_{-1}\wedge T_{-2}\ar[r,"q_{-3}"]&T_{-3}\ar[r]&0
\end{tikzcd}
\end{center}
The bottom vertical rows are the quotient maps induced by $u_{-3}$ and $k_{-3}$.
 Elements of $T_{-3}$ are considered to carry  homogeneous degree $-3$. By construction, $T_{-3}$ is a $\mathfrak{g}$-module and $q_{-3}$ is $\mathfrak{g}$-equivariant. Finally, if $V$ is a Lie algebra, then $T_{-2}=0$ which implies in turn that $T_{-3}=0$ as well, as required by item 3. of Proposition~\ref{prophierarchy}.

The next step is not as straightforward and illustrates what happens for higher degrees; for this reason, it is worth its own treatment. 
%Given that $T_{-2}=\bigslant{F_{-2}}{K_{-2}}$, we deduce that $\Lambda^2(F_{-2})$ is (non canonically) isomorphic to the direct sum $T_{-2}\wedge T_{-2}\,\oplus\, (T_{-2}\otimes K_{-2})\,\oplus\, K_{-2}\wedge K_{-2}$. 
Since $K_{-2}$ (resp. $K_{-3}$) is a well-defined subspace of $F_{-2}$ (resp. $F_{-3}$), %a choice of a complementary subspace to $K_{-2}$ in $F_{-2}$ -- with which $T_{-2}$ could be identified -- 
one can define a linear map $u_{-4}:F_{-1}\otimes K_{-3}\,\oplus\, F_{-2}\otimes K_{-2}\to \Lambda^2(F_\bullet)|_{-4}$ as the composition of the following two maps:
\begin{equation*}
(F_{-1}\otimes K_{-3})\,\oplus\, (F_{-2}\otimes K_{-2})\xhookrightarrow{\hspace{2cm}}  (F_{-1}\otimes F_{-3})\,\oplus\, (F_{-2}\otimes F_{-2}) \,\oplus\, (F_{-3}\otimes F_{-1})\xrightarrow{\hspace{2cm}} \Lambda^2(F_\bullet)|_{-4}
\end{equation*}
The first arrow is the canonical inclusion, whereas the second arrow is the quotient map defining the exterior algebra. The action of $u_{-4}$ can be explicitly given by the following definition:
\begin{equation*}
u_{-4}(x\otimes y)=x\wedge y\hspace{1cm}\text{for every}\quad x\otimes y\in (F_{-1}\otimes K_{-3})\,\oplus\, (F_{-2}\otimes K_{-2})
\end{equation*} Then, by construction, the map $u_{-4}$ is injective. This implies in turn that it is $\mathfrak{g}$-equivariant because $K_{-2}$ and $K_{-3}$ are $\mathfrak{g}$-modules. %Given that $T_{-2}=\bigslant{F_{-2}}{K_{-2}}$, one can deduce that $\Lambda^2(F_{-2})$ is (non canonically) isomorphic to the direct sum $T_{-2}\wedge T_{-2}\,\oplus\, F_{-2}\otimes K_{-2}$. The isomorphism depends on a choice of complementary subspace to $K_{-2}$ in $F_{-2}$. It means however that
%Even if the choice of complementary subspace to $K_{-2}$ is arbitrary, the cokernel of $u_{-4}$ is always canonically isomorphic to $T_{-2}\wedge T_{-2}\oplus T_{-1}\otimes T_{-3}$ (recall that $T_{-1}=F_{-1}$).
We have the following commutative diagram:
\begin{center}
\begin{tikzcd}[column sep=1cm,row sep=1cm]
0\ar[r]&\ar[d, "u_{-4}"] \parbox{3cm}{\centering $F_{-1}\otimes K_{-3}$\\$\oplus\, F_{-2}\otimes K_{-2}$}\ar[r,"\mathrm{id}"]&\ar[d, dashed, "k_{-4}"] \parbox{3cm}{\centering $F_{-1}\otimes K_{-3}$\\$\oplus\, F_{-2}\otimes K_{-2}$}\ar[r]&0\\
\Lambda^3(F_{-1}\oplus F_{-2})|_{-4}\ar[r, "-d_3"]&\parbox{2cm}{\centering $F_{-1}\wedge F_{-3}$\\ $\oplus\, \Lambda^2(F_{-2})$}\ar[r, "-d_2"]&F_{-4}\ar[r]&0
\end{tikzcd}
\end{center}

Notice that the bottom row is exact by Lemma \ref{lemmaexact}. 
%where $-d_2$ sends the chosen complementary subspace to 0 and hence the map is independent of the choice of complement and restricts to the free Lie algebra structure on    $F_{-1}\otimes F_{-3}$ $\oplus\, \Lambda^2(F_{-2})$.
The dashed vertical arrow $k_{-4}$ is the composite $-d_2\circ u_{-4}$  making this diagram commutative.
As a composition of $\mathfrak{g}$-equivariant morphisms, this map is $\mathfrak{g}$-equivariant. Its image is thus a $\mathfrak{g}$-submodule of $F_{-4}$, denoted $K_{-4}$. By construction, this space is uniquely defined.
%Although the image of the map $u_{-4}$ depends on the choice of complementary subspace to $K_{-2}$ into $F_{-2}$, the image of $k_{-4}$ does not. %cokernel of the map $u_{-4}$ is canonically isomorphic to 
%This can be shown as follows: first, even if the choice of complementary subspace to $K_{-2}$ is arbitrary, the cokernel of $u_{-4}$ is always canonically isomorphic to $T_{-2}\wedge T_{-2}\oplus T_{-1}\otimes T_{-3}$.
Then, recalling that $T_{-1}=F_{-1}$, one can quotient the bottom row by the respective images of $u_{-4}$ and $k_{-4}$ and obtain:
\begin{equation}\label{sequence1}
\begin{tikzcd}[column sep=1cm,row sep=1cm]
\Lambda^3(F_{-1}\oplus F_{-2})|_{-4}\ar[r]&T_{-1}\wedge T_{-3}\,\oplus \,\Lambda^2(T_{-2})\ar[r]&\bigslant{F_{-4}}{K_{-4}}\ar[r]&0
\end{tikzcd}
\end{equation}
%Notice that we have here adopted the same convention for $T_{-1}\wedge T_{-3}$ that we had chosen for $F_{-1}\wedge F_{-3}$. 
We set:
\begin{equation*}
T_{-4}:=\bigslant{F_{-4}}{K_{-4}}
\end{equation*}
so that elements of $T_{-4}$ are considered to carry  homogeneous degree $-4$.
We call $q_{-4}:T_{-1}\wedge T_{-3}\,\oplus \,\Lambda^2(T_{-2})\to T_{-4}$ the  map %canonically 
induced from $-d_2|_{\Lambda^2(F_\bullet)|_{-4}}$ by the quotient.  $T_{-4}$ inherits the quotient $\mathfrak{g}$-module structure and $q_{-4}$ is by construction $\mathfrak{g}$-equivariant.  By Lemma \ref{snake}  the sequence \eqref{sequence1} is exact. It implies in particular that $q_{-4}$ is surjective.
Moreover, if $V$ is a Lie algebra, then  $T_{-2}=0$ and $T_{-3}=0$, which implies that $T_{-4}=0$ as desired in that case. %Notice also that exactness of the sequence \eqref{sequence1} is conserved in the following one:
%\begin{equation}
%\begin{tikzcd}[column sep=1cm,row sep=1cm]
%\Lambda^3(T_{-1}\oplus T_{-2})|_{-4}\ar[r]&T_{-1}\wedge T_{-3}\,\oplus \,\Lambda^2(T_{-2})\ar[r]&\bigslant{F_{-4}}{K_{-4}}\ar[r]&0
%\end{tikzcd}
%\end{equation}
%This fact will be used in subsection \eqref{structure}.

Now, let $n\geq4$ and assume that the spaces $T_{-i}$, for $4\leq i\leq n$, have been defined as  quotients:
\begin{equation*}
T_{-i}=\bigslant{F_{-i}}{K_{-i}}
\end{equation*}
where $K_{-i}$ is the $\mathfrak{g}$-submodule of $F_{-i}$ (of degree $-i$) corresponding to the canonical image of $\bigoplus_{j=1}^{i-2} F_{-j}\otimes K_{-i+j}$ in $F_{-i}$, as the following commutative diagram (with exact rows) shows:

%\tjg{Suggestion on notation: If we change $\bigoplus_{j=1}^{i-2}T_{-j}\oplus\bigoplus_{j=2}^{i-1}K_{-j}\big)$ to  $\bigoplus_{j=1}^{i-2}T_{-j}\oplus\bigoplus_{j=1}^{i-2}K_{-j+1}\big)$ then we could simplify to $\bigoplus_{j=1}^{i-2}(T_{-j}\oplus K_{-j+1}\big)$}
\begin{center}
\begin{tikzcd}[column sep=1cm,row sep=1cm]
0\ar[r]&\ar[d, "u_{-i}"]\bigoplus_{j=1}^{i-2} F_{-j}\otimes K_{-i+j} \ar[r,"\mathrm{id}"]&\ar[d, "k_{-i}"] \bigoplus_{j=1}^{i-2} F_{-j}\otimes K_{-i+j} \ar[r]&0\\
\Lambda^3\big(\bigoplus_{j=1}^{i-2} F_{-j}\big)\big|_{-i}\ar[r, "-d_3"]&\Lambda^2\big(\bigoplus_{j=1}^{i-1} F_{-j}\big)\big|_{-i}\ar[r, "-d_2"]&F_{-i}\ar[r]&0
\end{tikzcd}
\end{center}
where the (injective) map $u_{-i}$ is the composition of the two following maps:
\begin{equation*}
\bigoplus_{j=1}^{i-2} F_{-j}\otimes K_{-i+j} 
\xhookrightarrow{\hspace{2cm}}  \bigoplus_{j=1}^{i-1} F_{-j}\otimes F_{-i+j} 
\xrightarrow{\hspace{2cm}} \Lambda^2(F_\bullet)|_{-i}
\end{equation*}
and where $k_{-i}$ is the unique linear making the diagram commutative.  We can assume that both $u_{-i}$ and $k_{-i}$ are $\mathfrak{g}$-equivariant. Moreover, if $V$ is a Lie algebra, assume that every $T_{-i}=0$, for $2\leq i\leq n$. %The vertical morphism on the left is not unique because it depends on the choice of an embedding of $\Lambda^2\big(\bigoplus_{j=2}^{i-2}T_{-j}\oplus\bigoplus_{j=2}^{i-2}K_{-j}\big)|_{-i}$ into $\Lambda^2\big(\bigoplus_{j=2}^{i-2} F_{-j}\big)|_{-i}$. However, whatever choice of such embedding we make, the morphism $k_{-i}$ is well defined/unique.

Let us show that these data uniquely define a vector space $T_{-n-1}$ of degree $-(n+1)$ satisfying the same conditions, one level higher. We will then use this induction to conclude the proof of Proposition \eqref{prophierarchy}. %Since for every $2\leq i\leq n$, the submodule $K_{-i}$ canonically injects itself into $F_{-i}$:
%First, one has to choose an embedding of $\Lambda^2\big(\bigoplus_{i=1}^{n-1}T_{-i}\oplus \bigoplus_{i=2}^{n}K_{-i}\big)|_{-n-1}$ into $\Lambda^2\big(\bigoplus_{i=1}^{n} F_{-i}\big)|_{-n-1}$, say $u_{-n-1}$. It is non-canonical because one has to pick up a complementary subspace in each $F_{-i}$ to which $T_{-i}$ would be sent. 
The finite direct sum $\bigoplus_{j=1}^{n-1} F_{-j}\otimes K_{-(n+1)+j}$ is a subspace of $\bigoplus_{j=1}^{n-1} F_{-j}\otimes F_{-(n+1)+j}$ and then, this induces a unique linear map $u_{-n-1}:\bigoplus_{j=1}^{n-1} F_{-j}\otimes K_{-(n+1)+j}\to\Lambda^2(F_\bullet)|_{-n-1}$ through the composition:
\begin{equation*}
\bigoplus_{j=1}^{n-1} F_{-j}\otimes K_{-(n+1)+j}\xhookrightarrow{\hspace{2cm}} 
\bigoplus_{j=1}^{n-1} F_{-j}\otimes F_{-(n+1)+j} \xrightarrow{\hspace{2cm}}
\Lambda^2\big(\bigoplus_{i=1}^{n} F_{-i}\big)\big|_{-n-1}
\end{equation*}
The map $u_{-n-1}$ is injective and a morphism of $\mathfrak{g}$-modules.
One can then define a commutative diagram:
\begin{center}
\begin{tikzcd}[column sep=1cm,row sep=1cm]
0\ar[r]&\ar[d, "u_{-n-1}"] \bigoplus_{j=1}^{n-1} F_{-j}\otimes K_{-(n+1)+j}\ar[r,"\mathrm{id}"]&\ar[d, dashed, "k_{-n-1}"] \bigoplus_{j=1}^{n-1} F_{-j}\otimes K_{-(n+1)+j}\ar[r]&0\\
\Lambda^3\big(\bigoplus_{i=1}^{n-1} F_{-i}\big)\big|_{-n-1}\ar[r, "-d_3"]&\Lambda^2\big(\bigoplus_{i=1}^{n} F_{-i}\big)\big|_{-n-1}\ar[r, "-d_2"]&F_{-n-1}\ar[r]&0
\end{tikzcd}
\end{center}
The dashed vertical line on the right is defined as the unique linear map -- denoted $k_{-n-1}$ -- making the diagram commutative. %We call $k_{-n-1}$ %:\Lambda^2(\bigoplus_{i=1}^{n}K_{-i})|_{-n-1}\to F_{-n-1}$ 
%this map.
 As a composition of two $\mathfrak{g}$-equivariant morphisms, it is $\mathfrak{g}$-equivariant. %Although the choice of $u_{-n-1}$ is not unique, the induced morphism $k_{-n-1}$ is.
The image of this map  uniquely defines a $\mathfrak{g}$-submodule  of $F_{-n-1}$ that we call $K_{-n-1}$. Since by Lemma \ref{lemmaexact}  the bottom row of the above diagram is exact, Lemma \ref{snake} implies that its quotient by the above row is an exact sequence:
\begin{equation}\label{sequence2}
\begin{tikzcd}[column sep=1cm,row sep=1cm]
\Lambda^3\big(\bigoplus_{i=1}^{n-1} F_{-i}\big)\big|_{-n-1}\ar[r]&\Lambda^2\big(\bigoplus_{i=1}^{n} T_{-i}\big)\big|_{-n-1}\ar[r,"q_{-(n+1)}"]&\bigslant{F_{-n-1}}{K_{-n-1}}\ar[r]&0
\end{tikzcd}
\end{equation}
 We set:
\begin{equation*}
T_{-n-1}:=\bigslant{F_{-n-1}}{K_{-n-1}}
\end{equation*}
and elements of $T_{-n-1}$ are considered to carry  homogeneous degree $-(n+1)$. We call $q_{-n-1}:\Lambda^2(\bigoplus_{i=1}^{n} T_{-i})|_{-n-1}\to T_{-n-1}$ the  map %canonically 
induced from $-d_2|_{\Lambda^2(F_\bullet)|_{-n-1}}$ by taking the quotient. By construction, $T_{-n-1}$ is a $\mathfrak{g}$-module and $q_{-n-1}$ is $\mathfrak{g}$-equivariant. Exactness of the sequence \eqref{sequence2} implies that $q_{-n-1}$ is surjective.
Thus, if $V$ is a  Lie algebra, since by induction $T_{-i}=0$ for every $2\leq i\leq n$,  we obtain that $T_{-n-1}=0$ as well. We then have proved that the induction hypothesis is true at level $n+1$.

\begin{remark}\label{remarquexacte} Notice that exactness of the sequence \eqref{sequence2} is preserved if one changes each $F_{-i}$ by $T_{-i}$ in the term at the extreme left, so that the following sequence is also exact:
\begin{equation}\label{equationcentrale}
\begin{tikzcd}[column sep=1cm,row sep=1cm]
\Lambda^3\big(\bigoplus_{i=1}^{n-1} T_{-i}\big)\big|_{-n-1}\ar[r]&\Lambda^2\big(\bigoplus_{i=1}^{n} T_{-i}\big)\big|_{-n-1}\ar[r]& T_{-n-1}\ar[r]&0
\end{tikzcd}
\end{equation}
This fact will be used in the proof of Lemma \ref{lemmadiff}. %subsection \eqref{structure}.
\end{remark}

Continuing the induction provides us with a (possibly infinite) graded vector space $T_\bullet=\bigoplus_{i=1}^\infty T_{-i}$ 
which has the following properties:
\begin{enumerate}
    \item Every vector space $T_{-i}$ is a $\mathfrak{g}$-module;
    \item $T_{-1}=V[1]$;
    \item $T_{-i}=0$ for every $i\geq2$ if and only if $V$ is a Lie algebra.
\end{enumerate}
%This is half of what needs to be shown.
So, items 1. and 2. of Proposition \ref{prophierarchy} are partially shown, whereas item 3. is proved. Let us now define a graded Lie bracket on $T_\bullet$ so that $T_\bullet=V[1]\oplus \llbracket T_\bullet, T_\bullet\rrbracket$.
%In particular, if $V$ is a Lie algebra then the graded Lie bracket on $T_\bullet = V[1]$ is zero. This is not contradictory with the fact that $V$ is a Lie algebra since the Lie bracket on $V$ is skew symmetric, whereas the expected graded Lie bracket on $V[1]$ would be symmetric, because element of $V[1]$ have odd degree. 
%and, since $T_{-1}=R_{-1}=V[1]$, the quotient map $\phi:T_{-1}\otimes T_{-2}\to T_{-1}\otimes\Lambda^2(T_{-1})$ is a morphism of $\mathfrak{g}$-modules.
The graded vector space $T_\bullet$ can be seen  as a quotient of  the free graded Lie algebra $F_\bullet=\mathrm{Free}(V[1])$ by the graded $\mathfrak{g}$-submodule $K_\bullet=\bigoplus_{i=2}^\infty K_{-i}$. 
Additionally, by construction, this graded module is a Lie ideal in the free graded Lie algebra $F_\bullet$. This implies in particular that the graded vector space $T_\bullet$ inherits a graded Lie algebra structure that descends from the free graded Lie algebra structure on $F_\bullet$.

The corresponding graded Lie bracket $\llbracket\,.\,,.\,\rrbracket$ can be made explicit: let $q:\Lambda^2(T_\bullet)\to T_{\bullet}$ be the degree 0 map whose restriction to $\Lambda^2(T_\bullet)_{-i}$ coincides with $q_{-i}$ for every $i\geq2$. This condition uniquely defines the map $q$. This bilinear map, being induced from the Chevalley-Eilenberg differential, satisfies the cohomological condition:
\begin{equation}\label{chevalleyeilen}
q^2=0
\end{equation}
Additionally, it is surjective and $\mathfrak{g}$-equivariant, by construction.
%where $\mathrm{Un}(2,p-2)$ are the $(2,p-2)$-unshuffles, and where $\epsilon^\sigma$ is the product of the Koszul sign of the permutation $\sigma$ with its signature, that is: $x_1\wedge\ldots\wedge  x_p=\epsilon^\sigma\,x_{\sigma(1)}\wedge \ldots\wedge x_{\sigma(p)}$. So for example, for  $p=2$ and $\sigma=(1\,2)$, we have  $\epsilon^{(1\,2)}=-(-1)^{|x_1||x_2|}$. %\textbf{Give some details why on $\Lambda^3(T_\bullet)$ the map $q$ is the map induced by quotienting $\Lambda^3(F_\bullet)\longrightarrow\Lambda^2(T_\bullet)$ by the ideal $K_\bullet$ and hence have the same image, making the following sequence exact:}
%\begin{center}
%\begin{tikzcd}[column sep=1cm,row sep=1cm]
%\Lambda^3(T_\bullet)|_{-i}\ar[r]&\Lambda^2(T_\bullet)|_{-i}\ar[r]&T_{-i}\ar[r]&0
%\end{tikzcd}
%\end{center}
%\tgr{I'll take a look at it also}
Then, for every two homogeneous elements $x,y\in T_\bullet$, one has:
\begin{equation}\label{crochet1}
    \llbracket x,y\rrbracket=q(x,y)
\end{equation}
%More precisely, it means that $q(x,y)\in T_{|x|+|y|}$, %$\llbracket x,y\rrbracket=q_{|x|+|y|}(x,y)$,
%where $|x|,|y|\leq-1$ are the respective degrees of the homogeneous elements $x$ and $y$.
This bracket is a (degree 0) graded bracket -- i.e. % it is graded:
$\llbracket x,y\rrbracket \in F_{|x|+|y|}$, 
and it is graded skew symmetric:
\begin{equation}
    \llbracket x,y\rrbracket=-(-1)^{|x||y|}\llbracket y,x\rrbracket
\end{equation}
Moreover, it satisfies the Jacobi identity -- which is equivalent to $q$ satisfying Equation \eqref{chevalleyeilen}. %is induced by the Lie bracket of $F_\bullet$.
Since $q$ is surjective on $T_{-i}$ for every $i\geq2$, the condition $T_\bullet=V[1]\oplus \llbracket T_\bullet, T_\bullet\rrbracket$ is automatically satisfied.  Item 2. of Proposition \ref{prophierarchy} is proved, there remains the last part of item 1.

Let us call $\rho_{-i}:\mathfrak{g}\otimes T_{-i}\to T_{-i}$ the map representing the inherited action of $\mathfrak{g}$ on $T_{-i}$,
and let $\rho:\mathfrak{g}\otimes T_\bullet \to T_\bullet$ be the degree 0 map restricting to $\rho_{-i}$ on $\mathfrak{g}\otimes T_{-i}$. Alternatively, $\rho$ can of course be seen as a degree 0  graded Lie algebra morphism  from $\mathfrak{g}$ to $\mathrm{End}(T_\bullet)$.
Then, since $q$ is $\mathfrak{g}$-equivariant, it means that the action of $\mathfrak{g}$ on $T_\bullet$ is a derivation of the Lie bracket, i.e. the representation $\rho:\mathfrak{g}\to\mathrm{End}(T_\bullet)$  actually takes values in $\mathrm{Der}(T_\bullet)$. This concludes the proof of Proposition \ref{prophierarchy}.\end{proof}

\begin{remark}
The graded Lie algebra structure defined on $T_\bullet$ is not independent of the Leibniz algebra structure on $V$ since we defined $T_{-2}$ as a quotient of $S^2(V)$ by the biggest $\mathfrak{g}$-submodule $K$ of $\mathrm{Ker}\big(\{\,.\,,.\,\}\big)\subset S^2(V)$. By induction, the dependence on the Leibniz structure has thus been  propagated implicitly along the whole structure of $T_\bullet$, in particular through the maps $q_{-i}$. In particular, semi-strict Lie-Leibniz triples satisfy: $$T_{-2}=\left(\bigslant{S^2(V)}{\mathrm{Ker}\big(\{\,.\,,.\,\}\big)}\right)[2]$$
\end{remark}

\begin{remark}
The graded Lie algebra obtained in Proposition \ref{prophierarchy} coincides with the one defined in Lemma 3.13 in \cite{lavau:TH-Leibniz}. This follows from: 1. the choice for $T_{-1}$ and $T_{-2}$ are identical in both constructions and 2. the map $q$ in the present section is dual to the map $\pi$ defined in \cite{lavau:TH-Leibniz}. The surjectivity of $q$ hence corresponds to the injectivity of $\pi$, then the spaces of both hierarchies coincide. Notice that a typo has unfortunately been left in the final, published version of the afore mentioned paper \cite{lavau:TH-Leibniz}: in Definition 2.13, one should read $H^2_{\mathrm{CE}}(\mathfrak{g})=\bigslant{\wedge^2(\mathfrak{g}_{-1})^*}{d_{\mathrm{CE}}(\mathfrak{g}_{-2})^*}$. \end{remark}

\subsection{Existence and uniqueness of the differential}\label{structure}

From now on, in order to avoid any misunderstanding, we will denote by $T_{\leq -1}$ the graded vector space $T_\bullet$ constructed in Proposition \ref{prophierarchy}, and more generally we will write $T_{\leq-i}=\bigoplus_{k\geq i}T_{-k}$. The direct sum $T_\bullet\oplus \mathfrak{g}$ is thus denoted  $T_{\leq0}$ while $T_\bullet\oplus \mathfrak{g}\oplus R_{\Theta}[-1]$ is denoted $T_{\leq+1}$, where the space $T_{+1}=R_\Theta[-1]$ is understood as the space $R_\Theta$, shifted by  degree $+1$.
 Additionally, for brevity sometimes we will write $\mathbb{T}$ in order  to denote $T_{\leq+1}$.
 
This subsection is dedicated to building a (non-trivial) differential $\partial=(\partial_{-i})_{i\geq1}$ on  the graded vector space $T_{\leq-1}$ so that it eventually extends to the right as a cochain complex:
\begin{center}
\begin{tikzcd}[column sep=1cm,row sep=0.4cm]
%\ldots\ar[r,"\partial_{-4}"]&T_{-4}\ar[r,"\partial_{-3}"]&T_{-3}\ar[r,"\partial_{-2}"]&T_{-2}\ar[r,"\partial_{-1}"]&V[1]\ar[r,"\partial_0=\Theta"]&\mathfrak{g}
\ldots\ar[r,"\partial_{-3}"]&T_{-3}\ar[r,"\partial_{-2}"]&T_{-2}\ar[r,"\partial_{-1}"]&V[1]\ar[r,"\Theta"]&\mathfrak{g}\ar[r,"-\eta(-;\Theta)"]& R_\Theta[-1]
\end{tikzcd}
\end{center}
The differential will be constructed precisely so that it is compatible with the graded Lie algebra structure on $T_{\leq-1}$ defined in Proposition \ref{prophierarchy}. However we postpone to  subsection \ref{patching} the proof that such objects equip the graded vector space $\mathbb{T}$ with a differential graded Lie algebra structure.
 The following proposition then provides existence and unicity of the desired differential: %Denote by $\mathbb{T}$ the graded vector space concentrated in degrees lower than or equal to $+1$ defined by $\mathbb{T}=  T_{\leq0}\oplus R_\Theta[-1]$. 
  %The grading of $T_{\leq0}$ is left unchanged, hence $T_0=\mathfrak{g}$ and $T_{-1}=V[1]$, etc. 
% This would emphasize its negative grading and moreover allows extension to higher degrees.  
%As the following proposition shows, we can define a unique differential on $T_{\leq -1}$ which  is \emph{almost} compatible with the graded Lie algebra structure on $T_{\leq-1}$ defined in Proposition \ref{prophierarchy}:
%The following proposition is the main result of t

%Now let us turn to the proof of the existence and uniqueness of a well-defined family of linear maps $$\partial_{-i}:T_{-i-1}\to T_{-i}$$ satisfying
%item 6. of Definition \ref{wooo} as well as the condition
%$\partial_{-i}\circ \partial_{-i-1}=0$, for every $i\geq1$:

% and being (almost) compatible with the graded Lie bracket.

\begin{prop}\label{propdiff}
Let $(\mathfrak{g},V,\Theta)$ be a Lie-Leibniz triple, and let $(T_{\leq-1},\llbracket\,.\,,.\,\rrbracket)$ be the negatively graded Lie algebra associated to $(\mathfrak{g},V,\Theta)$ by Proposition \ref{prophierarchy}. Then there exists a unique family of degree $+1$  $\mathfrak{h}$-equivariant linear maps $\partial=\big(\partial_{-i}:T_{-i-1}\to T_{-i}\big)_{i\geq1}$ %extending as a derivation of $\Lambda^\bullet(T_{\leq-1})$
 and satisfying the following conditions:
 \begin{align}
    \partial\big(\llbracket u,v\rrbracket\big)&=  2\{u,v\}\label{commutator0}\\
   \partial\big(\llbracket u,x\rrbracket\big)&=  \rho(\Theta(u);(x)) - \llbracket u, \partial(x)\rrbracket\label{commutator}\\
     \partial\big(\llbracket x,y\rrbracket\big) &= \llbracket \partial(x), y\rrbracket+(-1)^{|x|}\llbracket x,\partial(y)\rrbracket\label{commutatorbis}
 \end{align}
 for every $u,v\in T_{-1}$ and $x,y\in T_{\leq-2}$.
 Moreover, the degree $+1$ linear map $\partial=\big(\partial_{-i}:T_{-i-1}\to T_{-i}\big)_{i\geq1}$ is a differential, i.e. for every $i\geq1$,
\begin{equation}\partial_{-i}\circ \partial_{-i-1}=0.\label{differentialzero}
\end{equation}

\begin{remark}
The differential defined in Proposition \ref{propdiff} does not turn $T_{\leq-1}$ into a differential graded Lie algebra, but when one appends $T_0=\mathfrak{g}$ and $T_{+1}=R_\Theta[-1]$ to $T_{\leq-1}$ with adequate graded Lie brackets, the total graded vector space $\mathbb{T}=T_{\leq+1}$ becomes a differential graded Lie algebra, as is shown in Proposition \ref{prop2}.
\end{remark}

 %Moreover, for every $i\geq1$, the linear map $\partial_{-i}$, as an element of $\mathrm{Hom}(T_{-i-1},T_{-i})$ on which $\mathfrak{g}$ acts, generates a $\mathfrak{g}$-submodule $R_{-i}$ of $\mathrm{Hom}(T_{-i-1},T_{-i})$ which is isomorphic to a quotient of $R_\Theta$.
  %this linear map transforms in the same representation as the embedding tensor $\Theta$ \textbf{Eventually improve the formulation of this sentence}.
\end{prop}

\begin{proof} The proof is build on three lemmas  \ref{lemmacommut}, \ref{lemmadiff} and \ref{lemmahandy} involving a bilinear map $m:T_{-1}\wedge T_{\leq-1}\to T_{\leq-1}$ that is a convenient rewriting of the bilinear map $ \rho(\Theta(-);-)$ appearing in the right-hand side of Equation~\eqref{commutator}.
We define it component-wise, so that 
for every $i\geq2$,  let  %$m_{-i}:\Lambda^2(T_{\leq-1})\to T_{\leq-1}$ be the unique degree $+1$ map such that its restriction 
$m_{-i}:T_{-1}\wedge T_{-i}\to T_{-i}$ be the linear map defined by composing $\rho_{-i}$ with $\Theta$:
\begin{center}
\begin{tikzpicture}
\matrix(a)[matrix of math nodes, 
row sep=5em, column sep=5em, 
text height=1.5ex, text depth=0.25ex] 
{T_{-1}\wedge T_{-i}&T_{-i}\\ 
\mathfrak{g}\otimes T_{-i}&\\}; 
\path[->](a-1-1) edge node[above]{$m_{-i}$}  (a-1-2); 
\path[->](a-1-1) edge node[left]{$\Theta\otimes\mathrm{id}$} (a-2-1);
\path[->](a-2-1) edge node[below right]{${\rho_{-i}}$} (a-1-2);
\end{tikzpicture}
\end{center}
As for $i=1$,  we take $m_{-1}=2\,\{\,.\,,.\,\}:\Lambda^2(T_{-1})\to T_{-1}$, as we have the following identity: 
\begin{equation}
m_{-1}(x,y)=\rho_{-1}(\Theta(x);y)+\rho_{-1}(\Theta(y);x)=2\{x,y\} %=m_{-1}(x,y) 
\end{equation}
where $\rho_{-1}$ denotes the action of $\mathfrak{g}$ on $T_{-1}=V[1]$.
Then, the map $m_{-i}$ lifts the action of $\mathfrak{g}$ on $T_{-i}$ to an action of $V$ on $T_{-i}$ and the quadratic constraint \eqref{eq:equiv2} implies that each $m_{-i}$ is $\mathfrak{h}$-equivariant.

Let  $m:\Lambda^2(T_{\leq-1})\to T_{\leq-1}$ be the unique morphism whose restriction to $\Lambda^2(T_{\leq-1})|_{-i-1}$ is $m_{-i}$ (for every $i\geq1$). Therefore, on any subspace of $\Lambda^2(T_{\leq-1})|_{-i-1}$ other than $T_{-1}\wedge T_{-i}$ (for $i\geq2$), the morphism $m$ is the zero morphism.
Both the bilinear map $q$ defined at the end of the proof of Proposition \ref{prophierarchy} and the bilinear map $m$ can be extended to multilinear maps on $\Lambda^p(T_{\leq-1})$ (for every $p\geq3$) as follows:
\begin{align}
{q}(x_1\wedge\ldots\wedge x_p)&=\sum_{\sigma\in \mathrm{Un}(2,p-2)}\,\epsilon^\sigma\,q(x_{\sigma(1)}\wedge x_{\sigma(2)})\wedge x_{\sigma(3)}\wedge \ldots\wedge x_{\sigma(p)}\label{extensionQ}\\
m(x_1\wedge\ldots\wedge x_p)&=\sum_{\sigma\in \mathrm{Un}(2,p-2)}\,\epsilon^\sigma\,m(x_{\sigma(1)}\wedge x_{\sigma(2)})\wedge x_{\sigma(3)}\wedge \ldots\wedge x_{\sigma(p)}\label{extensionM}
\end{align}
where $\mathrm{Un}(2,p-2)$ is the set of  $(2,p-2)$-unshuffles, and $\epsilon^\sigma$ has been defined in Equation \eqref{eqkoszul}. Then, the first Lemma is the following:
%This has the following consequence:

\begin{lemma}\label{lemmacommut}
The following identity holds:
\begin{equation}\label{equationfondamentale8}
m\circ q+q\circ m=0
\end{equation} That is: the following diagram is commutative, for every $i\geq1$:
\begin{center}
\begin{tikzpicture}
\matrix(a)[matrix of math nodes, 
row sep=4em, column sep=3em, 
text height=1.5ex, text depth=0.25ex] 
{&\Lambda^2 (T_{\leq-1})_{-i-1}&T_{-i-1}\\ 
\Lambda^3( T_{\leq-1})_{-i-2}&\wedge^2( T_{\leq-1})_{-i-2}&% T_{-i})_{-i}
\\}; 
\path[->](a-1-2) edge node[above]{${q}$}  (a-1-3); 
\path[->](a-2-1) edge node[below]{${q}$}  (a-2-2); 
\path[->](a-2-1) edge node[above left]{$m$} (a-1-2);
\path[->](a-2-2) edge node[below right]{$-{m}$} (a-1-3);
\end{tikzpicture}
\end{center}
\end{lemma}

\begin{proof}
Let $x\wedge y\wedge z\in \Lambda^3( T_{\leq-1})_{-i-2}$. If $x$, $y$, and $z$ have degree strictly lower than $-1$, then 
\eqref{equationfondamentale8} is trivially satisfied (because $m$ and $-m$ both vanish). 
%\trh{so we can say that now at least one element is in $T_{-1}=V$, and we can say that it is $a$. In particular $|b|+|c|=-i$. %If $|b|=|c|=0$ then $a,b,c\in X_0$, but this does not change the computation below because $\Theta\circ q_0=0$
%\noindent $\diamond$ \textbf{Case 1: $a\in T\_{-1}T_T_{-1}$ and $|b|,|c|\leq-1$ (in particular $|b|+|c|=-i$)}
%\trd{Shall we establish $\Theta=0$ on any element whose degree is strictly lower than 0 when $\Theta$ first defined?} 
Then, one can assume that one element is in $T_{-1}=V[1]$, say $x$, which is then of degree $-1$. In particular $|y|+|z|=-i-1$.
The action of ${q}$ on $x\wedge y\wedge z$ is defined through formula \eqref{extensionQ}:
\begin{equation}\label{applicat}
{q}(x\wedge y\wedge z)=q(x\wedge y)\wedge z-(-1)^{|y||z|}q(x\wedge z)\wedge y+(-1)^{-|y|-|z|}q(y\wedge z)\wedge x
\end{equation}
In the following we will define $\Theta$ to be $0$ on any element whose degree is  lower than $-1$. This will allow us to keep track of those terms even if they formally vanish. Keeping in mind that $|x|=-1$, then applying $m$ on both sides of Equation \eqref{applicat} implies that: %first term on the left hand side of Equation \eqref{equationfondamentale8} reads :
 \begin{align}
 &m({q}(x\wedge y\wedge z))\\
 &=m\Big(-(-1)^{|z|(|y|-1)} z\wedge q(x\wedge y) +(-1)^{2|y||z|-|y|}y\wedge q(x\wedge z) -(-1)^{-|y|-|z|+i+1}x\wedge q(y\wedge z)\Big)\\
&=-(-1)^{|z|(|y|-1)}\Theta(z)\cdot q(x\wedge y) +(-1)^{|y|}\Theta(y)\cdot q(x\wedge z)-\Theta(x)\cdot q(y\wedge z)\\
 &=-(-1)^{|z|(|y|-1)}q\big((\Theta(z)\cdot x)\wedge y+x\wedge (\Theta(z)\cdot y)\big)\\
 &\hspace{1.5cm}+(-1)^{|y|}q\big((\Theta(y)\cdot x)\wedge z+x\wedge (\Theta(y)\cdot z)\big)\\
 &\hspace{2cm}-q\big((\Theta(x)\cdot y)\wedge z+y\wedge (\Theta(x)\cdot z)\big)\\
% &=q_{-i}\big(m_{|b|}(a\wedge b)\wedge c+b\wedge m_{|c|}(a\wedge c)\big)+\cdots\\
 &=-q\Big(m(x\wedge y)\wedge z-(-1)^{|y||z|}m(x\wedge z)\wedge y + (-1)^{|y|+|z|}m(y\wedge z)\wedge x\Big)\\
 &=-q( {m}(x\wedge y\wedge z))
 \end{align}
\noindent as desired (where we used the fact that $q$ is $\mathfrak{g}$-equivariant, thus in particular $\mathfrak{h}$-equivariant). To pass from the antepenultimate line to the penultimate one, we used the fact that $m(y,z)\neq0$ only if neither $y$ nor $z$ have simultaneously their degrees strictly lower than $-1$. Obviously, Equation \eqref{equationfondamentale8} extends to the whole of~$\Lambda(F_\bullet)$.\end{proof} %;  
 %\trd{which implies $\rho_{-i-1}\big(\Theta(\ - \ ), \ \big)$ acts as a derivation }
% in fact $\mathfrak{h}$-equivariance is sufficient).

%First, notice that Equation \eqref{commutator} can be written component wise as:
%\begin{equation}\label{commutator2}
%    m_{-i}=\partial\circ q+q_{-i}\circ \partial
%\end{equation}
%for every $i\geq1$.

Going back to the proof of Proposition \ref{propdiff}, the construction of the degree $+1$ linear endomorphism $\partial=\big(\partial_{-i}:T_{-i-1}\to T_{-i}\big)_{i\geq1}$ is made by induction. 
 First, for $i=1$, the definition of $\partial_{-1}$ is a by-product of the observation that $m_{-1}$ factors through $T_{-2}$: %(which is by construction a quotient of $\Lambda^2(T_{-1})$ by a $\mathfrak{g}$-submodule of $\mathrm{Ker}\big(\{\,.\,,.\,\}\big)$):
 %:= \bigslant{\Lambda^2(T_{-1})}{K}$.
\begin{center}
\begin{tikzpicture}
\matrix(a)[matrix of math nodes, 
row sep=5em, column sep=6em, 
text height=1.5ex, text depth=0.25ex] 
{&T_{-1}\\
\Lambda^2(T_{-1})&T_{-2}\\} ; 
\path[->>](a-2-1) edge node[above]{$q_{-2}$} (a-2-2); %node[belowright]{$\varphi$}  (a-1-2); 
\path[->, dashed](a-2-2) edge node[right]{$\partial_{-1}$} (a-1-2);
\path[->](a-2-1) edge node[above left]{$m_{-1}=2\,\{.\,,.\}$} (a-1-2);
\end{tikzpicture}
\end{center}
The map $\partial_{-1}$ is defined as the unique linear map making  the above triangle commutative, i.e. such that 
\begin{equation}\label{commutator0bis}
\partial_{-1}\circ q(u,v)=m(u,v)
\end{equation}
for every $u,v\in T_{-1}$.
 In particular, it is surjective on the ideal of squares $\mathcal{I}$ and it is $\mathfrak{h}$-equivariant because $\{\,.\,,.\,\}$ has both of these properties. A priori $\partial_{-1}$ is not $\mathfrak{g}$-equivariant because the symmetric bracket $\{\,.\,,.\,\}$ -- equivalently $m_{-1}$ --  is not. By construction, Equation \eqref{commutator0bis} is equivalent to Equation \eqref{commutator0}.

%However it is not $\mathfrak{g}$-equivariant, but it belongs to a $\mathfrak{g}$-submodule of $ R_\Theta$. \textbf{give more details of why}

Now assume that we have constructed a family of $\mathfrak{h}$-equivariant linear maps $\partial_{(-n)}=(\partial_{-i}:T_{-i-1}\to T_{-i})_{1\leq i\leq n}$ up to order $n\geq1$ satisfying the following conditions: $\partial_{(-n)}$ is  the unique linear degree $+1$ endomorphism of $T_{\leq-1}$ restricting to $\partial_{-i}$ on $T_{-i-1}$ for every $1\leq i\leq n$ (in particular it vanishes on $T_{-1}$), extending as a graded derivation on $\Lambda^\bullet(\bigoplus_{i=1}^{n+1} T_{-i})$, and  satisfying the following equation on $\Lambda^\bullet(\bigoplus_{i=1}^{n+1} T_{-i})$:
\begin{equation}\label{commutator2}
    \partial_{(-n)}\circ q=m+q\circ \partial_{(-n)}
\end{equation}
To define the map $\partial_{-n-1}:T_{-n-2}\to T_{-n-1}$ such that it extends Equation \eqref{commutator2} one step further, we first define a degree $+1$ linear map $j_{-n-1}: \Lambda^2(T_{\leq-1})|_{-n-2} \to T_{-n-1}$ as the sum:
\begin{equation}\label{chapeau}
j_{-n-1} := m_{-n-1}  + q_{-n-1}\circ\partial_{(-n)}
\end{equation}
% and extend it $\Lambda^p (T_{\leq-1})$ (for every $p\geq3$) using Equations \eqref{extensionQ} and \eqref{extensionM}:
%as a function ${j}_{-i-1}$:
%\begin{equation}
%{j}_{-n-1}(x_1\wedge\ldots\wedge x_p)=\sum_{\sigma\in \mathrm{Un}(2,p-2)}\,\epsilon^\sigma\,j_{-n-1}(x_{\sigma(1)}\wedge x_{\sigma(2)})\wedge x_{\sigma(3)}\wedge \ldots\wedge x_{\sigma(p)}\end{equation}
%Notice that this map is $\mathfrak{h}$-equivariant since $m_{-n-1}$, $q_{-n-1}$ and $\partial_{(-n)}$ are. As an element of $\mathrm{Hom}(\Lambda^2)$
%\trd{indices to be fixed}
By construction the map $j_{-n-1}$ is $\mathfrak{h}$-equivariant, and allows us to prove the existence and uniqueness of the linear map $\partial_{-n-1}$:
\begin{lemma}\label{lemmadiff}
The map $j_{-n-1}$ factors through $T_{-n-2}$; in fact, there exists a unique $\mathfrak{h}$-equivariant linear map $$\partial_{-n-1}:T_{-n-2}\to T_{-n-1}$$ such that the following triangle is commutative:
\begin{center}
\begin{tikzpicture}
\matrix(a)[matrix of math nodes, 
row sep=5em, column sep=6em, 
text height=1.5ex, text depth=0.25ex] 
{&T_{-n-1}\\ 
\Lambda^2 (T_{\leq-1})|_{-n-2}&T _{-n-2}\\}; 
\path[->>](a-2-1) edge node[above]{$q_{-n-2}$}  (a-2-2); 
\path[->, dashed](a-2-2) edge node[right]{$\partial_{-n-1}$} (a-1-2);
\path[->](a-2-1) edge node[above left]{$j_{-n-1}$} (a-1-2);
\end{tikzpicture}
\end{center}
\end{lemma}

\begin{proof}
%\trd{It would b good to have a morestriking notation instead of br or overline - perhaps mathfrak or just say: We us the same notation for $ d, m etc$;
%the content should make it clear which is meant}

Let $x\in T_{-n-2}$. By surjectivity of $q_{-n-2}$, there exists $y\in \Lambda^2 (T_{\leq-1})|_{-n-2}$ such that $x=q_{-n-2}(y)$. Then one would define:
\begin{equation}\label{eqporto}
\partial_{-n-1}(x)=j_{-n-1}(y)
\end{equation} if  guaranteed that any other
choice of pre-image of $x$ does not change the result, i.e. that  $j_{-n-1}\big(\mathrm{Ker}(q_{-n-2})\big)=0$. 
 Thus, let $w\in \mathrm{Ker}(q_{-n-2})$.  Exactness of the sequence \eqref{equationcentrale}:
 \begin{equation}
\begin{tikzcd}[column sep=1cm,row sep=1cm]
%\Lambda^3\big(\bigoplus_{i=1}^{n} T_{-i}\big)\big|_{-n-2}\ar[r,"q"]&\Lambda^2\big(\bigoplus_{i=1}^{n+1} T_{-i}\big)\big|_{-n-2}\ar[r,"q_{-n-2}"]& T_{-n-2}\ar[r]&0
\Lambda^3(T_{\leq-1})|_{-n-2}\ar[r,"q"]&\Lambda^2(T_{\leq-1})|_{-n-2}\ar[r,"q_{-n-2}"]& T_{-n-2}\ar[r]&0
\end{tikzcd}
\end{equation} implies the equality $\mathrm{Ker}(q_{-n-2})={q}\big(\Lambda^3(T_{\leq-1})|_{-n-2}\big)$. Then there exists $\alpha\in \Lambda^3(T_{\leq-1})|_{-n-2}$ such that ${q}(\alpha)=w$, and by Lemma \ref{lemmacommut}, we have:
\begin{equation}m_{-n-1}(w)=-q\circ {m}(\alpha)\end{equation}
But by the induction hypothesis, $m$ satisfies Equation \eqref{commutator2} on $\Lambda^3(T_{\leq-1})|_{-n-2}$,  %${m}_{(-i)}={\partial}_{(-i)}\circ {q}_{(-i)}-{q}_{(-i+1)}\circ{\partial}_{(-i+1)}$ 
 so we obtain that:
\begin{equation}
m_{-n-1}(w)=-q_{-n-1}\circ{\partial}_{(-n)}\circ{q}(\alpha)+q_{-n-1}\circ{q}\circ\partial_{(-n)}(\alpha)=-q_{-n-1}\circ{\partial}_{(-n)}(w)
\end{equation}
where we used Equation \eqref{chevalleyeilen} to get rid of the second term on the right hand side.
Hence the result: $j_{-n-1}(w)=0$.
$\mathfrak{h}$-equivariance of $\partial_{-n-1}$ is straightforward because $q_{-n-2}$ is $\mathfrak{g}$-equivariant and $j_{-n-1}$ is $\mathfrak{h}$-equivariant, so for any $a\in\mathfrak{h}$ one has:
\begin{equation}\label{eqporto2}
\partial_{-n-1}(a\cdot x)=\partial_{-n-1}(a\cdot q_{-n-2}(y))=\partial_{-n-1}( q_{-n-2}(a\cdot y))=j_{-n-1}(a\cdot y)=a\cdot (j_{-n-1}(y))=a\cdot \partial_{-n-1}(x)
\end{equation}
Notice that a priori $\partial_{-n-1}$ is not $\mathfrak{g}$-equivariant.
%Let $T\in T\_{-1}T_{-i-2}$, then for any two pre-images $y,y'\in S^2(T\_{-1}T_T_{-1}\oplus \ldots\oplus T\_{-1}T_{-i})_{-i-1}$ of $\mathbb{T}$, the element $y-y'$ lives in $\mathrm{Ker}(q_{-i-1})$. By construction of $q_{-i-1}$, we have the equality $\mathrm{Ker}(q_{-i-1})=\mathrm{Im}({q}_{-i}|_{S^3(X_0\oplus \ldots\oplus X_{-i})_{-i}})$.
\end{proof}

The linear map $\partial_{-n-1}$ defined in Lemma \ref{lemmadiff} is unique and has the advantage of extending Equation \eqref{commutator2} to the level $-(n+1)$, by construction. That is to say, setting $\partial_{(-(n+1))}$ to be  the unique linear degree $+1$ endomorphism of $T_{\leq-1}$ that restricts to $\partial_{-i}$ on $T_{-i-1}$ for every $1\leq i\leq n+1$, it extends as a graded derivation on $\Lambda^\bullet(\bigoplus_{i=1}^{n+2} T_{-i})$, and  satisfies the following equation on $\Lambda^\bullet(\bigoplus_{i=1}^{n+2} T_{-i})$:
\begin{equation}\label{commutator3}
    \partial_{(-(n+1))}\circ q=m+q\circ \partial_{(-(n+1))}
\end{equation}
%Moreover, since $j_{-n-1}$ and $q_{-i-1}$ are $\mathfrak{h}$-equivariant, and since $q_{-i-1}$ is moreover surjective, it is a straightforward computation to show that $\partial_{-n-1}$ is also $\mathfrak{h}$-equivariant. 
By induction, this demonstrates the existence and the uniqueness of a degree $+1$ map $\partial$ satisfying the following identity on $\Lambda(T_{\leq-1})$:
\begin{equation}\label{commutator4}
\partial\circ q=m+q\circ \partial
\end{equation}
This equation, when restricted to $\Lambda^2(T_{\leq-1})$,  is equivalent to the three Equations \eqref{commutator0}, \eqref{commutator} and \eqref{commutatorbis}. Indeed, if two elements of $T_{-1}$ are involved in Equation \eqref{commutator4}, then the last term disappears and $m=m_{-1}$ so that we are left with Equation \eqref{commutator0bis}, which is equivalent to Equation~\eqref{commutator0}; if only one element of $T_{-1}$ is involved in Equation~\eqref{commutator4}, then we obtain \eqref{commutator}, while if no element of $T_{-1}$ is involved, $m=0$ and we obtain Equation~\eqref{commutatorbis}. This argument proves the first part of Proposition \ref{propdiff}.

%completes the proof of Proposition~\ref{propdiff}.
%This result shows in particular that the map $m_{-i-1}$ is null-homotopic and satisfies the desired equation, that is: we have proved that Item 5 of the induction hypothesis is satisfied at level $-i-1$. 
%This, together with the $\mathfrak{g}$-invariance of $q_{-i-1}$, implies that Item 1 of the induction hypothesis is satisfied at level $-i-1$. We now have to prove that Item 3 is also satisfied at level $-i-1$ to finish the proof of the inductive step:
Now we will show that this degree $+1$ linear map $\partial$  is a differential. The proof is made by induction and relies on the following lemma:
\begin{lemma}\label{lemmahandy}
The following identity holds:
\begin{equation}
    m\circ \partial+\partial\circ m=0
\end{equation}
That is: the following diagram is commutative for every $i\geq1$:
\begin{center}
\begin{tikzpicture}
\matrix(a)[matrix of math nodes, 
row sep=4em, column sep=3em, 
text height=1.5ex, text depth=0.25ex] 
{&T_{-i}\\ 
\Lambda^2 (T_{\leq-1})|_{-i-1}&T_{-i-1}\\
\Lambda^2 (T_{\leq-1})|_{-i-2}&\\
}; 
\path[->](a-2-2) edge node[right]{$\partial_{-i}$} (a-1-2); 
\path[->](a-3-1) edge node[left]{$\partial_{(-i)}$}  (a-2-1); 
\path[->](a-2-1) edge node[above left]{$-m_{-i}$} (a-1-2);
\path[->](a-3-1) edge node[below right]{$m_{-i-1}$} (a-2-2);
\end{tikzpicture}
\end{center}
\end{lemma}
%or, in other words:
%\begin{equation}\label{equationfondamentale2}
%m_{-i}{\partial}_{(-i)}=-\partial_{-i}m_{-i-1}
%\end{equation}
\begin{proof}
 Let $i\geq1$ and let $x\wedge y\in \Lambda^2 (T_{\leq-1})_{-i-2}$. If neither $x$ nor $y$ belong to $T_{-1}$, 
 then $m(x\wedge y)=0$ and $m\big({\partial}(x)\wedge y+(-1)^{|x|}x\wedge \partial(y)\big)=0$ (because either $\Theta$ vanishes on $x,y$ or $\Theta\circ\partial=0$), so that commutativity of the diagram is trivially satisfied. Thus, we can assume that at least one element, say $x$, belongs to $T_{-1}$. In that case $|y|=-i-1<-1$. As usual we assume that $\Theta$ vanishes on elements of degree lower than $-1$, and that $\partial(T_{-1})=0$. Then, one has: %Now, since $-i-1\leq -1$, the only possible other case is when $a\in T_{-1}$ and $|b|\leq-1$. For $|b|=-i-1$,
 %Then:
 \begin{align}
 -m\circ {\partial}(x\wedge y)&=m(x\wedge \partial(y))\\
 &=\Theta(x)\cdot \partial(y)-(-1)^{-i}\Theta(\partial(y))\cdot x\\
 &=\partial(\Theta (x)\cdot y)\\
 &=\partial\circ m(x\wedge y)
 \end{align}
 as desired. Notice that we used the $\mathfrak{h}$-equivariance of $\partial$ to pass from the antepenultimate line to the penultimate one. \end{proof}
%\end{proof}

%We can now turn to the following proposition:
%This Lemma is actually useful to prove that the degree $+1$ linear morphism $\partial$ is actually a differential on $T_{\leq-1}$:

%\begin{prop}\label{propdiff2}
%The degree $+1$ linear map $\partial=\big(\partial_{-i}:T_{-i-1}\to T_{-i}\big)_{i\geq1}$ defined in Proposition \ref{propdiff} is a differential, i.e. for every $i\geq1$,
%\begin{equation}\partial_{-i}\circ \partial_{-i-1}=0.\label{differentialzero}
%\end{equation}
%\end{prop}

%\begin{proof} %First, let $j:\Lambda^2(T_{\leq-1})\to$ be the unique degree $+1$ bilinear morphism whose restriction to $\Lambda^2(T_{\leq-1})_{-i-2}$ is $j_{-i-1}$.
%One can extend it $\Lambda^p (T_{\leq-1})$ (for every $p\geq3$) using Equations \eqref{extensionQ} and \eqref{extensionM}:
%as a function ${j}_{-i-1}$:
%\begin{equation}
%{j}(x_1\wedge\ldots\wedge x_p)=\sum_{\sigma\in \mathrm{Un}(2,p-2)}\,\epsilon^\sigma\,j(x_{\sigma(1)}\wedge x_{\sigma(2)})\wedge x_{\sigma(3)}\wedge \ldots\wedge x_{\sigma(p)}\end{equation}
%This notation will become handy during the proof.
The proof of the second part of Proposition \ref{propdiff} is then straightforward.
First, let us show that $\partial_{-1}\circ \partial_{-2}=0$. Let $x\in T_{-3}$, then by construction there exists $y\in T_{-1}\wedge T_{-2}$ such that $\partial_{-2}(x)=j_{-2}(y)$. Then, we get:
\begin{align}
    \partial_{-1}\circ j_{-2}(y)&=\partial_{-1}\circ m_{-2}(y)+\partial_{-1}\circ q_{-2}\circ \partial_{(-1)}(y)\\
    &=\partial_{-1}\circ m_{-2}(y)+m_{-1}\circ \partial_{(-1)}(y)\\
    &=0
\end{align}
where we used Lemma \ref{lemmahandy} to make the penultimate line vanish.

Secondly, let $n\geq1$ and assume that Equation \eqref{differentialzero} has been shown for every $1\leq i\leq n$. In particular the following equation is valid:
\begin{equation}\partial_{(-i)}\circ \partial_{(-i-1)}=0\label{differ0}
\end{equation} for every $1\leq i\leq n$, where $\partial_{(-i)}$ is the unique linear degree $+1$ endomorphism of $T_{\leq-1}$ that restricts to $\partial_{-j}$ on $T_{-j-1}$ for every $1\leq j\leq i$.
 We want to show that:
\begin{equation}\partial_{-n-1}\circ \partial_{-n-2}=0
\end{equation}

Let $x\in T_{-n-3}$.  By construction, we have  $\partial_{-n-2}(x)=j_{-n-2}(y)$ for some $y\in \Lambda^2 (T_{\leq-1})_{-n-3}$.
%By surjectivity of $q_{-i-1}$, there exists $y\in S^2(T_{-1}\oplus \ldots\oplus T_{-i-1})_{-i-1}$ such that $x=q_{-i-1}(y)$.
%Then we would like to show that 
%\begin{equation}
 %   \partial_{-n-1}\circ j_{-n-2}(y)=0.\end{equation} %This result relies on the following inclusion:
%\begin{equation}\label{inclusion2}\mathrm{Im}(j_{-i-1})\subset \mathrm{Ker}(\partial_{-i}),\end{equation}
%Hence Lemma 2 is proved.
Then, using Lemma \ref{lemmahandy}, we have:
\begin{align}
   \partial_{-n-1} \circ\partial_{-n-2}(x)
   &= \partial_{-n-1}\circ j_{-n-2}(y)\\
   &=\partial_{-n-1}\circ \big(m_{-n-2}+q_{-n-2}\circ \partial_{(-n-1)}\big)(y)\\
   &=\big(-m_{-n-1}\circ \partial_{(-n-1)}+\partial_{-n-1}\circ q_{-n-1}\circ\partial_{(-n-1)}\big)(y)\\
   &=\big(-m_{-n-1}+\partial_{-n-1}\circ q_{-n-2}\big)\circ \partial_{(-n-1)}(y)\\
   &= q_{-n-1} \circ \partial_{(-n)}\circ\partial_{(-n-1)}(y)\\
   &= 0
\end{align}
where we used Equation \eqref{commutator2} between the fourth and fifth line, and the induction hypothesis Equation \eqref{differ0} for $i=n$ at the penultimate line. This concludes the proof of Proposition \ref{propdiff}.
%where we used Lemma \ref{lemma4} and Item 5 of the induction hypothesis at level $-i$.
%\begin{equation}\partial_{-i}(w)=\partial_{-i}m_{-i-1}(z)-\partial_{-i}q_{-i}{d}_{-i}(z)\\=(m_{-i}-\partial_{-i}q_{-i}){d}_{-i}(z)=q_{-i+1}\underbrace{{d}_{-i+1}{d}_{-i}(z)}_{=\ 0}=0,\end{equation} %\ref{inclusion2}. This 
\end{proof}

%Moreover we assume that for every $1\leq i\leq n$ there exists a surjective homomorphism of $\mathfrak{g}$-modules $\mu_{-i}:R_{\Theta}\to R_{-i}$, such that:
%\begin{equation}
%    \partial_{-i}=\mu_{-i}(\Theta)
%\end{equation}
Before going further, we need to investigate the relationship between the linear maps $(\partial_{-i})_{i\geq1}$, and the embedding tensor $\Theta$. As $\Theta$, each of the linear map $\partial_{-i}$ is $\mathfrak{h}$-equivariant, but certainly not $\mathfrak{g}$-equivariant. The cyclic module generated by each linear map $\partial_{-i}$ %is subject to the action of $\mathfrak{g}$ and thus would 
is a $\mathfrak{g}$-submodule of $\mathrm{Hom}(T_{-i-1},T_{-i})$, that we denote $R_{-i}$:
\begin{equation}\label{defrep2}
R_{-i}:=\mathrm{Span}\big(\partial_{-i},a_1\cdot(a_2\cdot(\ldots(a_m\cdot \partial_{-i})\ldots))\,\big|\, a_1,a_2,\ldots,a_m\in\mathfrak{g}\big)%\subset \mathrm{Hom}(T_{-i-1}, T_{-i})
\end{equation}
 Recall that $R_\Theta$ is itself a cyclic $\mathfrak{g}$-submodule of $\mathrm{Hom}(V,\mathfrak{g})$ (see Equation \eqref{defrep}). Then the $R_{-i}$ are characterized by the following observation:

\begin{prop}\label{propref}
 For every $i\geq1$, there exists a surjective morphism of $\mathfrak{g}$-modules $\mu_{-i}: R_{\Theta}\to R_{-i}$, %making $R_{-i}$ isomorphic to a quotient of $R_\Theta$
such that:
 \begin{equation}\label{eqqqqq}
 \partial_{-i}=\mu_{-i}(\Theta)
 \end{equation}
 \end{prop} 
 
 \begin{proof}
 The proof is made by induction. We will show that the $R_{-i}$ are quotients of the cyclic module $R_\Theta$, so that the $\mu_{-i}$ are the quotient maps. First let us observe that,  for every $i\geq2$,  as an element of $\mathrm{Hom}(T_{-1}\wedge T_{-i}, T_{-i})$, the map $m_{-i}$ %is subject to the action of $\mathfrak{g}$ induced by the structure of $\mathfrak{g}$-module of $T_{\leq-1}$. For every $i\geq1$, let us call 
 generates a $\mathfrak{g}$-submodule  $M_{-i}$ of $\mathrm{Hom}(T_{-1}\wedge T_{-i}, T_{-i})$ through the successive actions of $\mathfrak{g}$ on $m_{-i}$:
 \begin{equation}
 M_{-i}:=\mathrm{Span}\big(m_{-i},a_1\cdot(a_2\cdot(\ldots(a_m\cdot m_{-i})\ldots))\,\big|\, a_1,a_2,\ldots,a_m\in\mathfrak{g}\big)%\subset \mathrm{Hom}(T_{-1}\wedge T_{-i}, T_{-i})
 \end{equation}
%On the one hand, by construction of $m_{-n-1}$ -- the composition of $\Theta$ and $\rho_{-n-1}$ (a $\mathfrak{g}$-invariant object), the representation $\rho_{-n-1}$ induces a surjective homomorphism of $\mathfrak{g}$-modules from $R_\Theta$ onto the  $\mathfrak{g}$-submodule generated by $m_{-n-1}$ in  $\mathrm{Hom}(\Lambda^2(T_{\leq-1})|_{-n-2}, T_{-n-1})$, which is then isomorphic to a quotient of $R_\Theta$.
 %For $i=1$, let us call $M_{-1}$ the $\mathfrak{g}$-submodule of $\mathrm{Hom}(\Lambda^2(T_{-1}),T_{-1})$ generated by the symmetric bracket. 

The representation $\rho_{-i}:\mathfrak{g}\otimes T_{-i}\to T_{-i}$ induces a morphism of $\mathfrak{g}$-modules $\overline{\rho_{-i}}$ from $\mathrm{Hom}(T_{-1}\wedge T_{-i},\mathfrak{g}\otimes T_{-i})$ to $\mathrm{Hom}(T_{-1}\wedge T_{-i}, T_{-i})$.
By construction, this homomorphism satisfies:
 \begin{equation}\label{dernier}
m_{-i}=\overline{\rho_{-i}}(\Theta\wedge\mathrm{id}_{T_{-i}})
\end{equation}
Here, the right hand side is the composition of $\overline{\rho_{-i}}$ with the following injective linear map:
\begin{align*}
\mathrm{Hom}(T_{-1},\mathfrak{g})&\xrightarrow{\hspace*{1.2cm}} \mathrm{Hom}(T_{-1}\wedge T_{-i},\mathfrak{g}\otimes T_{-i})\\
	\alpha\hspace{0.2cm}&\xmapsto{\hspace*{1.2cm}}\hspace{0.4cm}\alpha\wedge\mathrm{id}_{T_{-i}}
\end{align*}
This composition of maps yields a morphism of $\mathfrak{g}$-modules $\nu_{-i}=\overline{\rho_{-i}}\circ\alpha:\mathrm{Hom}(T_{-1},\mathfrak{g})\to\mathrm{Hom}(T_{-1}\wedge T_{-i}, T_{-i})$ sending $\Theta$ onto $m_{-i}$:
 \begin{equation}\label{maison}
m_{-i}=\nu_{-i}(\Theta)
\end{equation}
In particular, since $\nu_{-i}$ is $\mathfrak{g}$-equivariant, the action of $\mathfrak{g}$ on $m_{-i}$  depends only on the action of $\mathfrak{g}$ on $\Theta$, through~$\nu_{-i}$. 
In particular, it sends $R_\Theta$ onto $M_{-i}$, which is then isomorphic to a quotient of $R_\Theta$. This line of argument is still valid for $i=1$, where $m_{-1}=2\,\{\,.\,,.\,\}$.
%For $i=1$, recall that $m_{-1}$, which can be identified with twice the symmetric bracket $\{\,.\,,.\,\}$, is induced by the embedding tensor: 
%\begin{equation}
%m_{-1}(x,y)=2\{x,y\}=\rho_{-1}(\Theta(x);y)+\rho_{-1}(\Theta(y);x) %=m_{-1}(x,y) 
%\end{equation}
%where $\rho_{-1}$ denotes the action of $\mathfrak{g}$ on $T_{-1}=V[1]$.
%Following the same line of arguments as above, the $\mathfrak{g}$-equivariant map $\rho_{-1}$ induces a surjective homomorphism of $\mathfrak{g}$-modules $\nu_{-1}:R_\Theta\to M_{-1}$, satisfying the condition:
 %\begin{equation}
%m_{-1}=\nu_{-1}(\Theta)
%\end{equation}

%These data will be used to extend the graded Lie algebra structure on $T_{\leq0}$ to a dgLa structure on $T_{\leq0}\oplus R_\Theta[-1]$.
%\begin{equation}
%a\cdot(m_{-i})(u)=a\cdot(m_{-i}(u))-m_
%\end{equation}

\bigskip
\noindent \textbf{Initialization:}
Let us now turn to the maps $\partial_{-i}$: assume that we have built the differential $\partial=(\partial_{-i}:T_{-i-1}\to T_{-i})_{i\geq1}$ as in Proposition \ref{propdiff}.
%The linear map $\partial_{-1}:T_{-2}\to T_{-1}$ generates a $\mathfrak{g}$-submodule of $\mathrm{Hom}(T_{-2},T_{-1})$, that we call $R_{-1}$. 
%Recall that the symmetric bracket $\{\,.\,,.\,\}$ can be seen as an element of the $\mathfrak{g}$-module $\mathrm{Hom}(\Lambda^2(T_{-1}),T_{-1})$ and as such, generates a $\mathfrak{g}$-submodule $M_{-1}$.  %It implies that $\mathrm{Ker}(\mu_s)$ is a $\mathfrak{g}$-submodule of $R_\Theta$ and that $R_{sym}$ is isomorphic, as a $\mathfrak{g}$-module, to the quotient $\bigslant{R_\Theta}{\mathrm{Ker}(\mu_s)}$. 
Given that the map $q_{-2}:T_{-1}\wedge T_{-1}\to T_{-2}$ is surjective and $\mathfrak{g}$-equivariant, and that $m_{-1}=\partial_{-1}\circ q_{-2}$, it follows that $q_{-2}$ induces a morphism of $\mathfrak{g}$-modules $\overline{q_{-2}}:M_{-1}\to R_{-1}$, sending $m_{-1}$ onto $\partial_{-1}$.
%the $\mathfrak{g}$-equivariant quotient map $q_{-2}$ induces a homomorphism of $\mathfrak{g}$-modules $\overline{q_{-2}}:M_{-1}\to R_{-1}$. 
Pre-composing it with $\nu_{-1}$ induces a homomorphism of $\mathfrak{g}$-modules $\mu_{-1}:R_\Theta\to R_{-1}$:
\begin{center}
\begin{tikzpicture}
\matrix(a)[matrix of math nodes, 
row sep=5em, column sep=5em, 
text height=1.5ex, text depth=0.25ex] 
{R_{\Theta}&M_{-1}&R_{-1}\\}; 
\path[->](a-1-1) edge node[above]{$\nu_{-1}$} (a-1-2); 
\path[->](a-1-2) edge node[above]{$\overline{q_{-2}}$} (a-1-3);
\path[->](a-1-1) edge [bend right] node[below]{$\mu_{-1}$} (a-1-3);
\end{tikzpicture}
\end{center}
In particular the composition of maps implies that:
\begin{equation}
    \partial_{-1}=\mu_{-1}(\Theta)
\end{equation}
which, by definition of $R_{-1}$, implies that $\mu_{-1}$ is onto.
Moreover, since $\mathrm{Ker}(\mu_{-1})$ is a $\mathfrak{g}$-submodule of $R_\Theta$ then $R_{-1}$ is isomorphic, as a $\mathfrak{g}$-module, to the quotient $\bigslant{R_\Theta}{\mathrm{Ker}(\mu_{-1})}$.

\bigskip
\noindent \textbf{Inductive step:} Let $n\geq1$ and assume %that the family of linear maps $(\partial_{-i}:T_{-i-1}\to T_{-i})_{1\leq i\leq n}$ up to order $n\geq1$ 
that for every $1\leq i\leq n$ there exists a surjective homomorphism of $\mathfrak{g}$-modules $\mu_{-i}:R_{\Theta}\to R_{-i}$, such that:
\begin{equation}
    \partial_{-i}=\mu_{-i}(\Theta)
\end{equation}
Let us prove that this is still the case one level higher.

The bilinear map $j_{-n-1}$, as an element of $\mathrm{Hom}(\Lambda^2(T_{\leq-1})|_{-n-2}, T_{-n-1})$, generates a $\mathfrak{g}$-submodule of $\mathrm{Hom}(\Lambda^2(T_{\leq-1})|_{-n-2}, T_{-n-1})$ that we call  $S_{-n-1}$:
\begin{equation}
S_{-n-1}:=\mathrm{Span}\big(j_{-n-1}, a_1\cdot(a_2\cdot(\ldots(a_m\cdot j_{-n-1})\ldots))\,\big|\, a_1,a_2,\ldots,a_m\in\mathfrak{g}\big)%\subset\mathrm{Hom}(\Lambda^2(T_{\leq-1})|_{-n-2}, T_{-n-1})
\end{equation}
By the induction hypothesis and Equation \eqref{maison} at level $-n-1$, the right hand side of Equation \eqref{chapeau} consists of a set of  terms, each of which is the image of $\Theta$ through a $\mathfrak{g}$-equivariant map. Then it means that there exists a morphism of $\mathfrak{g}$-modules $s_{-n-1}:R_\Theta\to S_{-n-1}$ such that:
\begin{equation}
j_{-n-1}=s_{-n-1}(\Theta)
\end{equation}
By definition of $S_{-n-1}$, this morphism is surjective.
%On the one hand, the map $m_{-n-1}$ generates a $\mathfrak{g}$-module $M_{-n-1}$ of $\mathrm{Hom}(\Lambda^2(T_{\leq-1})|_{-i-1}, T_{-i})$ that is isomorphic to the quotient $\bigslant{R_\Theta}{\mathrm{Ker}(\overline{\rho_{-n-1}})}$ induced by the homomorphism of $\mathfrak{g}$-modules $\overline{\rho_{-n-1}}:R_\Theta\to R_{-n-1}$. On the other hand, the induction hypothesis assumes that each $R_{-i}$ is isomorphic to the quotient $\bigslant{R_\Theta}{\mathrm{Ker}(\mu_{-i})}$ induced by the homomorphism $\mu_{-i}:R_\Theta\to R_{-i}$. 
%The map $q$ being $\mathfrak{g}$-invariant, each term $q\circ \partial_{-i}$ on the right-hand side of Equation \eqref{chapeau}, as an element of  $\mathrm{Hom}(\Lambda^2(T_{\leq-1})|_{-i-1}, T_{-i})$, generates a $\mathfrak{g}$-submodule that is a quotient of $R_\Theta$. %Composing those quotient maps by the homomorphism $q_{-n-1}$ 

%Hence, the bilinear map $j_{-n-1}$, as an element of $\mathrm{Hom}(\Lambda^2(T_{\leq-1})|_{-n-2}, T_{-n-1})$, belongs to a $\mathfrak{g}$-submodule that is isomorphic, as a $\mathfrak{g}$-module, to a quotient of $R_\Theta$. Let us call $S_{-n-1}$ this $\mathfrak{g}$-submodule of $\mathrm{Hom}(\Lambda^2(T_{\leq-1})|_{-n-2}, T_{-n-1})$.
Now given that the map $q_{-n-2}:\Lambda^2(T_{\leq-1})|_{-n-2}\to T_{-n-2}$ is surjective and $\mathfrak{g}$-equivariant, and that $\partial_{-n-1}\circ q_{-n-2}=j_{-n-1}$ by Lemma \ref{lemmadiff}, it induces a morphism of $\mathfrak{g}$-modules $\overline{q_{-n-2}}:S_{-n-1}\to R_{-n-1}$, such that $\overline{q_{-n-2}}(j_{-n-1})=\partial_{-n-1}$. Pre-composing it with $s_{-n-1}$ induces in turn a homomorphism of $\mathfrak{g}$-modules $\mu_{-n-1}:R_\Theta\to R_{-n-1}$:
\begin{center}
\begin{tikzpicture}
\matrix(a)[matrix of math nodes, 
row sep=5em, column sep=5em, 
text height=1.5ex, text depth=0.25ex] 
{R_{\Theta}&S_{-n-1}&R_{-n-1}\\}; 
\path[->](a-1-1) edge node[above]{$s_{-n-1}$} (a-1-2); 
\path[->](a-1-2) edge node[above]{$\overline{q_{-n-2}}$} (a-1-3);
\path[->](a-1-1) edge [bend right] node[below]{$\mu_{-n-1}$} (a-1-3);
\end{tikzpicture}
\end{center}
The composition of maps implies that:
\begin{equation}
    \partial_{-n-1}=\mu_{-n-1}(\Theta)
\end{equation}
which, by definition of $R_{-n-1}$, implies that $\mu_{-n-1}$ is onto.
Moreover, since $\mathrm{Ker}(\mu_{-n-1})$ is a $\mathfrak{g}$-submodule of $R_\Theta$ then $R_{-n-1}$ is isomorphic, as a $\mathfrak{g}$-module, to the quotient $\bigslant{R_\Theta}{\mathrm{Ker}(\mu_{-n-1})}$.
%The fact that $q_{-n-2}\Lambda^2(T_{\leq-1})|_{-n-2}\to T_{-n-2}$ is onto, implies that the homomorphism $\mu_{-n-1}$ is surjective. 
 %In other words, the $\mathfrak{g}$-module $R_{-n-1}$ is isomorphic to the quotient $\bigslant{R_\Theta}{\mathrm{Ker}(\mu_{-n-1})}$, and then one can set:
%\begin{equation}
 %   \partial_{-n-1}=\mu_{-n-1}(\Theta)
%\end{equation}
This concludes the proof of Proposition \ref{propref}. \end{proof}
 
% \begin{remark}
%Proposition \ref{propref} states in a mathematical language what physicists say when they assert that $\partial$ in the same representation as $\Theta$.
%\end{remark}

Recall what is it we know so far. 
Let $(\mathfrak{g}, V,\Theta)$ be a Lie-Leibniz triple. Proposition \ref{prophierarchy} %canonically 
uniquely
associates to this Lie-Leibniz triple a negatively graded vector space $T_{\leq-1}$ equipped with a graded Lie algebra structure. %Then, Lemma \ref{propextension} has shown that this graded Lie algebra structure can be extended to $T_{\leq0}=T_{\leq-1}\oplus \mathfrak{g}$. %Let us denote by $\llbracket\,.\,,.\,\rrbracket$ the graded Lie bracket on $T_{\leq0}$. 
Then, Proposition \ref{propdiff} has shown that the complex $T_{\leq-1}$ can be equipped with a differential $\partial=(\partial_{-i}:T_{-i-1}\to T_{-i})_{i\geq1}$ which is almost compatible with the graded Lie bracket, see Equations \eqref{commutator0}- \eqref{commutatorbis}. We will now extend this differential and the Lie bracket to the graded vector space $\mathbb{T}=T_{\leq-1}\oplus\mathfrak{g}\oplus R_{\Theta}[-1]$, making it a differential graded Lie algebra. The construction is precisely made so that when the Lie-Leibniz triple is a Lie algebra crossed module or more generally when $V$ is a Lie algebra, one finds the differential graded Lie algebra associated to it as in Proposition \ref{experimental}. %Extending the algebraic structure is the topic of the next subsection, while the proof of Theorem \ref{centraltheorem} can be found in subsection \ref{restriction}.

\subsection{The differential graded Lie algebra structure on \texorpdfstring{$\mathbb{T}$}{T}}\label{patching}

The present subsection is dedicated to showing that the graded Lie bracket on $T_{\leq -1}$ defined in Proposition~\ref{prophierarchy} and the differential defined in Proposition \ref{propdiff} can be extended to the whole of $\mathbb{T}=T_{\leq-1}\oplus\mathfrak{g}\oplus R_{\Theta}[-1]$ in such a way that it forms a differential graded Lie algebra.

Let $(\mathfrak{g},V,\Theta)$ be a Lie-Leibniz triple, and let $T_{\leq-1}$ be the negatively graded Lie algebra associated to $(\mathfrak{g},V,\Theta)$ by Proposition \ref{prophierarchy}. Recall that $T_0=\mathfrak{g}$ and $T_{+1}=R_\Theta[-1]$.
 The differential defined on $T_{\leq-1}$ by Proposition \ref{propdiff} straightforwardly extends to $T_{-1}$ and $T_{0}$ by setting:
\begin{equation}\label{definetheta0}\partial_0=\Theta:T_{-1}\to T_0\quad \text{and}\quad \partial_{+1}=-\eta(-,\Theta):T_{0}\to T_{+1}\end{equation}
 We indeed have $\partial_0\circ\partial_{-1}=0$ because the subspace $\mathrm{Im}(\partial_{-1})=\mathcal{I}$ is by construction
  included in $\mathrm{Ker}(\Theta)$, while $\partial_{+1}\circ \partial_0=0$ because %whatever element $x\in T_{-1}=V[1]$ is chosen, 
%\begin{equation}
%\partial_{+1}\circ \partial_0(x)=-\eta(\Theta(x),\Theta)
%\end{equation} which identically vanishes since 
$\Theta$ is $\mathrm{Im}(\Theta)$-equivariant, by the quadratic constraint \eqref{eq:equiv2}. For convenience, we denote the extended differential $\partial$ as well.  %The cochain complex has thus indeed be extended on the right one step further.
%In other words, it can by formally understood as a $\mathfrak{g}$-equivariant linear map \begin{equation}
%\mu:R_\Theta[-1]\xrightarrow{\hspace{1cm}}\mathrm{Hom}(\mathbb{T},\mathbb{T})\big|_{+1}=\bigoplus_{i\geq-1}\mathrm{Hom}(T_{-i-1},T_{-i})\label{wowww}
%\end{equation}
% from $R_\Theta$ to $\bigoplus_{i\leq-1}\mathrm{Hom}(T_{-i-1},T_{-i})$
% whose co-restriction to $\mathrm{Hom}(T_{-i-1},T_{-i})$ coincides with $\mu_{-i}$. 
Now, let us extend the
graded Lie algebra structure defined in %by the family of maps $q=\big(q_i:\Lambda^2\big(\bigoplus_{j=1}^{i-1} T_{-j}\big)|_{-i}\to T_{-i}\big)_{i\leq-2}$ in 
Proposition~\ref{prophierarchy}  %even 
to the higher levels $T_0=\mathfrak{g}$ and $T_{+1}=R_{\Theta}[-1]$. For every $a,b\in\mathfrak{g}$ and $x\in T_{\leq-1}$, let us set:
\begin{align}
    \llbracket a,x\rrbracket&=a\cdot x\label{crochet2} \quad \text{and}\quad   \llbracket x,a\rrbracket=-a\cdot x,\\
    \llbracket a,b\rrbracket&=[a,b],\label{crochet3}
\end{align}
where $[\,.\,,.\,]$ is the Lie bracket on $\mathfrak{g}$. The first result is the following:

\begin{lemma}\label{propextension}
$(T_{\leq0},\llbracket\,.\,,.\,\rrbracket)$ is a graded Lie algebra.
\end{lemma}

%\begin{lemma}\label{propextension}
%Let $(\mathfrak{g},V,\Theta)$ be a Lie-Leibniz triple, and let $T_{\leq-1}$ be the negatively graded Lie algebra associated to $(\mathfrak{g},V,\Theta)$ by Proposition \ref{prophierarchy}. Then, the graded Lie algebra structure extends to $T_{\leq0}=T_{\leq-1}\oplus \mathfrak{g}$, by setting:
%\begin{align}
 %   \llbracket a,x\rrbracket&=a\cdot x\label{crochet2}\\
 %   \llbracket a,b\rrbracket&=[a,b]\label{crochet3}
%\end{align}
%for every $a,b\in\mathfrak{g}$ and $x\in T_{\leq-1}$, where $[\,.\,,.\,]$ is the Lie bracket on $\mathfrak{g}$. % and where $\rho:\mathfrak{g}\to\mathrm{End}(T_{\leq-1})$ is the representation of $\mathfrak{g}$ on $T_{\leq-1}$. %, inherited in the construction of $T_\bullet$.
%such that the graded Lie bracket on $T_0=\mathfrak{g}$ coincides with the Lie bracket defining the Lie algebra structure on $\mathfrak{g}$.
%\end{lemma}

\begin{proof}
%Let  $T_0=\mathfrak{g}$, and let us define the bracket of one element $a\in\mathfrak{g}$ and one homogeneous element $x\in T_{|x|}$ (where $|x|\leq-1$) as:
%\begin{equation}\label{crochet2}
 %   \llbracket a,x\rrbracket=\rho_{|x|}(a;x)
%\end{equation}
%where $\rho_{|x|}:\mathfrak{g}\to\mathrm{End}(T_{|x|})$ is the representation of $\mathfrak{g}$ on $T_{|x|}$, inherited in the construction of $T_\bullet$. The bracket between two elements of $\mathfrak{g}$ would be the pre-existing Lie bracket of $\mathfrak{g}$.
The Jacobi identity for the Lie bracket on $T_0$ is automatically satisfied because $\mathfrak{g}$ is a Lie algebra.
Then, one need only check that the bracket defined in Equation \eqref{crochet2} and the graded Lie bracket defined in Equation \eqref{crochet1} are compatible, in the sense that they satisfy the two Jacobi identities:
\begin{align}
    \llbracket a,\llbracket b,x\rrbracket\rrbracket&=\llbracket\llbracket a,b\rrbracket, x\rrbracket+\llbracket b,\llbracket a,x\rrbracket\rrbracket\label{diddy1}\\
    \llbracket a,\llbracket x,y\rrbracket\rrbracket&=\llbracket\llbracket a,x\rrbracket, y\rrbracket+\llbracket x,\llbracket a,y\rrbracket\rrbracket\label{diddy2}
\end{align}
for any $a,b\in\mathfrak{g}$ and $x,y\in T_{\leq-1}$. 
The first equation encodes the fact that $T_{|x|}$ is a $\mathfrak{g}$-module:
\begin{equation}
    a\cdot(b\cdot x)=[a,b]\cdot x+b\cdot(a\cdot x)
\end{equation}
so it is automatically satisfied. Recalling that on $T_{\leq-1}$ the bracket corresponds to the family of maps $q=(q_{-i}:\Lambda^2(T_{\leq-1})|_{-i}\to T_{-i})_{i\leq-2}$, Equation \eqref{diddy2} reads:
\begin{equation}
    a\cdot q(x,y)=q(a\cdot x,y)+q(x,a\cdot y)
\end{equation}
which corresponds to the $\mathfrak{g}$-equivariance of the map $q$, which is satisfied by construction.\end{proof}

%\begin{example}\label{examplediffercro}
%If $(\mathfrak{g},\mathfrak{c},\Theta)$ is a Lie algebra crossed module, then $T_{\leq-1}= T_{-1}= \mathfrak{c}[1]$ (by item 3. of Proposition \ref{prophierarchy}) and the graded Lie algebra structure defined on $T_{\leq0}=\mathfrak{g}\oplus \mathfrak{c}[1]$ by Lemma \ref{propextension} coincides with the action of $\mathfrak{g}$ on $\mathfrak{c}$ together with the Lie bracket on $\mathfrak{g}$. The Lie algebra structure of $\mathfrak{c}$ can then be reconstructed from Equation \eqref{equatdiff1}. %This example shows that adding differentials to the picture considerably enhance the possibilities of 
%\end{example}

\begin{remark}\label{remarkhystrion}
The graded Lie algebra $(T_{\leq0},\llbracket\,.\,,.\,\rrbracket)$, when equipped with the differential $\partial=\big(\partial_{-i}:T_{-i-1}\to T_{-i}\big)_{i\geq0}$ is \emph{not} a differential graded Lie algebra. Indeed, for any $a\in\mathfrak{g}$ and $x\in T_{\leq-1}$, the Leibniz rule reads:
\begin{equation}
    \partial\big(\llbracket a,x\rrbracket \big)=\llbracket \partial(a),x\rrbracket +\llbracket a,\partial(x)\rrbracket = \llbracket a,\partial(x)\rrbracket
\end{equation}
as the restriction of the differential on $\mathfrak{g}$ is (in the present case) the zero map. The difference between the right-hand side and the left-hand side is thus zero, meaning that we have the following identity:
\begin{equation}\label{eqidrasil}
 (a\cdot\partial)(x)=\llbracket a,\partial(x)\rrbracket-\partial\big(\llbracket a,x\rrbracket \big)=0
\end{equation}
for every $a\in\mathfrak{g}$. But this is in contradiction with the fact that the maps $\partial_{-i}:T_{-i-1}\to T_{-i}$ are certainly \emph{not} $\mathfrak{g}$-equivariant. Equation \eqref{eqidrasil} would be only valid on $\mathfrak{h}$, implying in turn that $T_{\leq-1}\oplus \mathfrak{h}$ is a differential graded Lie algebra.  However, this latter partial result  \cite{lavau:TH-Leibniz}, is quite unsatisfying since one loses most data from $\mathfrak{g}$. %, as was shown in $\partial$\cite{lavau:TH-Leibniz}. However, as said 
\end{remark}

We will now extend the graded Lie algebra structure to $T_{+1}=R_\Theta[-1]$ and check at the same time that it forms a graded Lie algebra structure on $\mathbb{T}$ which is compatible with the differential $\partial$.
 % One can extend the differential $\partial$ to $\mathbb{T}$ by adding a new map from $T_0$ to $T_{+1}$:%$\partial_{+1}=-\eta(-;\Theta)$:
%\begin{center}
%\begin{tikzcd}[column sep=1cm,row sep=0.4cm]
%\ldots\ar[r,"\partial_{-3}"]&T_{-3}\ar[r,"\partial_{-2}"]&T_{-2}\ar[r,"\partial_{-1}"]&V[1]\ar[r,"\Theta"]&\mathfrak{g}\ar[r,"-\eta(-;\Theta)"]& R_\Theta[-1]
%\end{tikzcd}
%\end{center}
%In other words, we have set:
%\begin{equation}\label{definetheta1}
%\partial_{+1}=-\eta(-;\Theta)
%\end{equation}
%One may also extend the graded Lie bracket $\llbracket\,.\,,.\,\rrbracket$ to $\mathbb{T}$, %This extension relies on the following fact:
%\begin{lemma}
%Let $\mathfrak{g}$ be a (semi-simple?) Lie algebra and let $U$ be a $\mathfrak{g}$-module. Then $U$ is irreducible if and only if every nonzero vector of $U$ is \emph{cyclic}, i.e. it generates $U$ through the action of $\mathfrak{g}$: $\mathfrak{g}\cdot u=U$ for every $u\neq0$.
%\end{lemma}
%The extension uses  the following supplementary definitions:
Recalling that the generators of $R_\Theta$ are of the form $\Theta_{a_1a_2\ldots a_m}=a_1\cdot(a_2\cdot (\ldots(a_m\cdot\Theta)\ldots))$ for some $a_1,a_2, \ldots, a_m\in\mathfrak{g}$,
we set the following definitions for the bracket involving $T_{+1}=R_{\Theta}[-1]$:
\begin{align}
%    \llbracket \Theta,a\rrbracket&=-\eta(a,\Theta)\\
    \llbracket \Theta, x\rrbracket&=\partial(x) \qquad \text{and}\qquad \llbracket x,\Theta\rrbracket=-(-1)^{|x|}\partial(x)\label{thetadiff}\\
     \llbracket\Theta_{a_1a_2\ldots a_m},x\rrbracket&=a_1\cdot \llbracket \Theta_{a_2\ldots a_m},x\rrbracket-\llbracket \Theta_{a_2\ldots a_m},a_1\cdot x\rrbracket\label{enforcedjacobi}\\
     \llbracket x,\Theta_{a_1a_2\ldots a_m}\rrbracket&=-(-1)^{|x|} \llbracket\Theta_{a_1a_2\ldots a_m},x\rrbracket\label{enforcedjacobi42}\\
    \llbracket T_{+1}, T_{+1}\rrbracket&=0\label{thetasquared}
\end{align}
for every homogeneous element $x\in T_{\leq0}$ and every $a_1,\ldots,a_m\in\mathfrak{g}$. As a side remark, be aware that physicists do not necessarily impose that $\llbracket T_{+1},T_{+1}\rrbracket=0$, see Section 3.4. of \cite{Bonezzi:2019bek}. %Equations \eqref{thetadiff}, \eqref{enforcedjacobi}, \eqref{enforcedjacobi42} and \eqref{thetasquared} define how the bracket $\llbracket\,.\,,.\,\rrbracket$ is extended to the whole of $T_{+1}$. %\textbf{We will study below that it automatically satisfies a whole set of Jacobi identities, but we need to enforce the following one:}
%\begin{equation}\label{enforcedjacobi2}
%\llbracket u, \llbracket v,x\rrbracket\rrbracket=-\llbracket v, \llbracket u,x\rrbracket\rrbracket
%\end{equation}
%for every $u,v\in T_{+1}$ and $x\in \mathbb{T}$.

Notice that Equations \eqref{thetadiff} and \eqref{enforcedjacobi} need a bit of explanation. Indeed, $\Theta$ is a linear map of vector spaces whereas $\partial$ is a family of linear maps, which do not necessarily even belong to the same $\mathfrak{g}$-modules. However, we have seen in the proof of Proposition \ref{propref} that each cyclic $\mathfrak{g}$-module $R_{-i}$ is a quotient of $R_\Theta$, through a quotient map $\mu_{-i}$ sending $\Theta$ on $\partial_{-i}$, see Equation \eqref{eqqqqq}.
We can straightforwardly extend Proposition \ref{propref} to $\partial_0$ and $\partial_{+1}$. Since $\partial_0=\Theta$, we know that $R_{0}=R_\Theta$ so that we set $\mu_0=\mathrm{id}_{R_\Theta}$. Next, we define $R_{+1}$ to be the cyclic $\mathfrak{g}$-submodule of  $\mathrm{Hom}(\mathfrak{g},R_\Theta)$ generated by $\partial_{+1}$, %Additionally, one can check that the linear map $\partial_{+1}=-\eta(-,\Theta)$ generates a $\mathfrak{g}$-submodule $R_{+1}$ of $\mathrm{Hom}(\mathfrak{g},R_\Theta)$. 
so that the $\mathfrak{g}$-equivariant linear map $-\eta$ sending $\Theta$ to $\partial_{+1}$ induces a morphism of $\mathfrak{g}$-modules $\mu_{+1}:R_\Theta\to R_{+1}$ such that  $\mu_{+1}(\Theta)=\partial_{+1}$. %Then, one can set $\mu=(\mu_{-i}:R_\Theta\to R_{-i})_{i\geq-1}$ to be the family of degree 0 $\mathfrak{g}$-equivariant linear maps sending $\Theta$ to $\partial=(\partial_{-i}:T_{-i-1}\to T_{-i})_{i\geq-1}$.
%\begin{equation}
 %   \partial_{-i}=\mu_{-i}(\Theta)\hspace{1cm}\text{for every $i\geq-1$.}
%\end{equation}
 %Altogether, the latter form a  family of linear maps $\mu=(\mu_{-i}:R_\Theta\to R_{-i})_{i\geq-1}$  so that
 Then one can make sense of Equation  \eqref{thetadiff} by understanding that $\llbracket\Theta,-\rrbracket$, evaluated on a homogeneous element $x$, actually corresponds to $\mu_{|x|+1}(\Theta)(x)=\partial_{|x|+1}(x)$. Moreover, $\mathfrak{g}$-equivariance of $\mu_{|x|+1}$ %, this allows us to symbolically write:
%\begin{equation}
%a_1\cdot(a_2\cdot (\ldots(a_m\cdot\partial)\ldots))=\llbracket \Theta_{a_1a_2\ldots a_m},-\rrbracket
%\end{equation}
%for any $a_1,\ldots,a_m\in\mathfrak{g}$. This
 allows to make sense of Equation \eqref{enforcedjacobi} as:
\begin{equation}\label{marcl}
\llbracket \Theta_{a_1a_2\ldots a_m},x\rrbracket=a_1\cdot(a_2\cdot (\ldots(a_m\cdot\partial)\ldots))(x)
\end{equation}
The right-hand side of Equation \eqref{marcl} could have been alternatively found by iterating Equation \eqref{enforcedjacobi} $m$~times on the left-hand side of Equation \eqref{marcl}.  %It is then transparent that on the left-hand side we are thinking about $a_1\cdot(a_2\cdot (\ldots(a_m\cdot\partial_{-i})\ldots))$ when the right-hand side is applied to an element of $T_{-i-1}$.

From Equations \eqref{thetadiff}-\eqref{thetasquared}, we can deduce the following identities:
 \begin{equation}\label{eq:imporg}
 \partial\big(\llbracket \Theta_a,x\rrbracket\big)=-\llbracket \Theta_a,\partial(x)\rrbracket\quad \text{and}\quad a\cdot(\partial)\big(\llbracket \Theta,x\rrbracket\big)=-\llbracket \Theta,a\cdot(\partial)(x)\rrbracket
 \end{equation} Then, Equations \eqref{eq:imporg} imply that the following identities are equivalent:
 \begin{align}
 \partial\big(\llbracket \Theta_{ab},x\rrbracket\big)=-\llbracket \Theta_{ab},\partial(x)\rrbracket&&\Longleftrightarrow &&
 a\cdot(\partial)\big(\llbracket \Theta_b,x\rrbracket\big)&=-\llbracket \Theta_b,a\cdot(\partial)(x)\rrbracket\label{diablo}\\
 &&\Longleftrightarrow &&  b\cdot(a\cdot(\partial))\big(\llbracket \Theta,x\rrbracket\big)&=-\llbracket \Theta,b\cdot(a\cdot(\partial))(x)\rrbracket
 \end{align} However, none of them is deducible from existing equations. Then, one is free to chose these identities to hold, implying that at the next level the following are equivalent:
\begin{align}
\partial\big(\llbracket \Theta_{abc},x\rrbracket\big)=-\llbracket \Theta_{abc},\partial(x)\rrbracket&&\Longleftrightarrow && a\cdot(\partial)\big(\llbracket \Theta_{bc},x\rrbracket\big)&=-\llbracket \Theta_{bc},a\cdot(\partial)(x)\rrbracket \label{diablo2}\\ 
&&\Longleftrightarrow && b\cdot(a\cdot(\partial))\big(\llbracket \Theta_{c},x\rrbracket\big)&=-\llbracket \Theta_{c},b\cdot(a\cdot(\partial))(x)\rrbracket\\
&&\Longleftrightarrow && c\cdot(b\cdot(a\cdot(\partial)))\big(\llbracket \Theta,x\rrbracket\big)&=-\llbracket \Theta,c\cdot(b\cdot(a\cdot(\partial)))(x)\rrbracket\label{diablo3}
\end{align}
being understood that none of them  is deducible from existing equations.  Assuming that any -- equivalently, all -- of them hold, we then find another set of equivalent identities at the next level, etc.  By induction, we are thus free (and led) to assume the following  compatibility condition between the bracket and the differential, generalizing the left-hand sides of Lines \eqref{eq:imporg}, \eqref{diablo} and \eqref{diablo2} to any level:
 \begin{equation}\label{enforcedjacobi2}
 \forall\ u\in T_{+1}, v\in T_{\leq-1}\qquad\partial(\llbracket u,v\rrbracket)=-\llbracket u,\partial(v)\rrbracket
 \end{equation}

 % the family of maps $\mu=(\mu_{-i}:R_\Theta\to R_{-i})_{i\geq-1}$ satisfying Equation \eqref{eqqqqq}. %Then, one can see  $\mu(\Theta)$ as a definition for $\llbracket\Theta,-\rrbracket$. %Moreover, this makes sense of  the action of $\mathfrak{g}$ on both sides of Equation \eqref{thetadiff}: by the homomorphism property, %acting on $\Theta$ and then taking $\mu$ is equivalent to taking $\mu$ and then acting on $\partial$:
%\begin{equation}\label{joker}
%\big\llbracket \llbracket a,\Theta\rrbracket, x\big\rrbracket= (a\cdot\partial_{-i}) (x)=\rho(a;\partial_{-i}(x))-\partial_{-i}(\rho(a;x))
%\end{equation}
%for every $x\in T_{-i-1}$ and $a\in\mathfrak{g}$. Notice that the term on the left hand side can be rewritten $\llbracket \Theta_a,x\rrbracket$.

%This relation is important to prove the Jacobi identities in the following Proposition:

It is now time to check the compatibility of the extended bracket and the extended differential on $\mathbb{T}$. %Recall what we have so far: given a Lie-Leibniz triple $(\mathfrak{g},V,\Theta)$, let $T_{\leq-1}$ be the graded vector space associated to $(\mathfrak{g},V,\Theta)$ by Proposition \ref{prophierarchy} and extended to degree 0 and $+1$ by appending to $T_{-1}$ the two vector spaces $T_0=\mathfrak{g}$ and $T_{+1}=R_{\Theta}[-1]$. We call $\mathbb{T}=\bigoplus_{i\geq-1}T_{-i}$ the graded vector space thus obtained, and 
Recall that the graded bracket $\llbracket \,.\,,.\,\rrbracket$ is defined by Proposition \ref{prophierarchy} and Equations \eqref{crochet2}, \eqref{crochet3}, \eqref{thetadiff},  \eqref{enforcedjacobi}, \eqref{enforcedjacobi42} and \eqref{thetasquared}, while the differential  is defined in Proposition \ref{propdiff} and extended to $T_0$ and $T_{+1}$ by Equations \eqref{definetheta0}. Additionally, we have required that these two objects satisfy Equation \eqref{enforcedjacobi2}.
Then, under these assumptions, we have the following important result, which generalizes Proposition \ref{experimental} to any Lie-Leibniz triple:

\begin{prop}\label{prop2}
 $(\mathbb{T},\llbracket\,.\,,.\,\rrbracket,\partial)$ is a differential graded Lie algebra.
\end{prop}

\begin{proof} The proof is not much different than that of Proposition \ref{experimental} because most problems only arise in the interaction between elements at level $+1$ and $0$.
We already know that $\partial^2=0$ by construction, and that the Jacobi identity is satisfied on $T_{\leq0}$ by Lemma \ref{propextension}.
First let us prove that the bracket and the differential are compatible, i.e. that they satisfy the Leibniz rule:
\begin{equation}\label{Leibnizz}
    \partial\big(\llbracket u,v\rrbracket \big)=\llbracket \partial(u),v\rrbracket +(-1)^{|u|}\llbracket u,\partial(v)\rrbracket 
\end{equation}
for every $u,v\in \mathbb{T}$. We have several cases to handle:

\textbf{1.} If $u,v\in T_{+1}$, or if $u\in T_{+1}$ and $v\in T_{0}$ (or conversely) the identity is trivial as a consequence of Equation \eqref{thetasquared} and because $\partial|_{T_{+1}}$ is  supposed to be the zero map.

\textbf{2.} Let $u\in T_{+1}$ and $v\in T_{\leq-1}$. Then, the Leibniz rule for $u$ and $v$ is nothing but Equation \eqref{enforcedjacobi2}. %, hence justifying that we assumed is was satisfied (nothing prevents us to do so). %The proof is made by induction. If $u=\Theta$ then Equation \eqref{Leibnizz} is automatically satisfied. \textbf{si on impose enforcedjacobi2 alors on a pas besoin de ce passage car tout est automatique !!} Assume now that $u$ is of the form $\Theta_{a}=\eta(a;\Theta)=-\partial(a)$ for some $a\in\mathfrak{g}$, then the left hand side of Equation \eqref{Leibnizz} reads:

\textbf{3.} If $u,v\in T_0=\mathfrak{g}$ then, recalling that $\partial|_{T_0}=-\eta(-;\Theta)$, Equation \eqref{Leibnizz} is equivalent to the following:
\begin{equation}
    -[u,v]\cdot\Theta=v\cdot(u\cdot\Theta)-u\cdot(v\cdot\Theta)
\end{equation}
which corresponds to the fact that $R_\Theta$ is a $\mathfrak{g}$-module. %$\eta:\mathfrak{g}\to\mathrm{End}(R_\Theta)$ is a Lie algebra morphism.

\textbf{4.} If $u\in T_0$ and $v\in T_{-1}$, then Equation \eqref{Leibnizz} becomes:
\begin{equation}
    \partial(u\cdot v)=\llbracket -u\cdot \Theta,v\rrbracket+u\cdot\partial(v)
\end{equation}
This equation is equivalent to:
\begin{equation}
   \llbracket u\cdot\Theta,v\rrbracket= u\cdot(\partial(v))-\partial(u\cdot v)
\end{equation}
which is nothing but Equation \eqref{enforcedjacobi} for $\Theta_u=u\cdot \Theta$.

\textbf{5.} If $u,v\in T_{\leq-1}$, then  we have several subcases (we will implicitly use Equation \eqref{crochet1} in this part):

\textbf{5.a.} If $u,v\in T_{-1}$, recalling that $\partial_0=\Theta$, Equation \eqref{Leibnizz} reads:
\begin{equation}
    \partial_{-1}\big(q_{-2}(u,v)\big)=\Theta(u)\cdot v+\Theta(v)\cdot u
\end{equation}
This equation is nothing but Equation \eqref{commutator0}.

\textbf{5.b.} If $u\in T_{-1}$ and $v\in T_{\leq-2}$, then Equation \eqref{Leibnizz} becomes:
\begin{equation}
    \partial\big(q(u,v)\big)=\Theta(u)\cdot v-q\big(u,\partial(v)\big)
\end{equation}
This equation can be rewritten as:
\begin{equation}
   \Theta(u)\cdot v=\partial\big(q(u,v)\big)+q\big(u,\partial(v)\big)
\end{equation}
which, given the definition of $m_{|v|}$, is nothing but Equation \eqref{commutator}.

\textbf{5.c.} If $u,v\in T_{\leq-2}$, then Equation \eqref{Leibnizz} can be rewritten as: %setting $-i=|u|+|v|+1$, then we have:
%\begin{equation}
 %   \partial\big(q(u,v)\big)=q\big(\partial(u),v\big)-(-1)^{|u|}q\big(u,\partial(v)\big)
%\end{equation}
%which could be rewritten as:
\begin{equation}
    \partial\big(q(u,v)\big)=q\circ \partial(u\wedge v)
\end{equation}
which is, once the term on the right hand side is transported to  the left hand side, nothing but Equation \eqref{commutatorbis}. %, since on this space $m(u,v)=0$.
Hence the three subcases of item 5., gathered together, are equivalent to the identities \eqref{commutator0}, \eqref{commutator}, and \eqref{commutatorbis}. This result was expected since by construction the bilinear map $m:\Lambda^2(T_{\leq-1})\to T_{\leq-1}$ is the composition of $\partial_0$ with the graded Lie bracket. To conclude, the Leibniz rule is thus satisfied on the whole of $\mathbb{T}$.

\bigskip
%Let us check the Jacobi identity and then we will check the Leibniz rule. 
Now let us check the Jacobi identity for the graded Lie bracket. For notational purposes, recall that for $u,v,w\in \mathbb{T}$, the Jacobi identity reads:
\begin{equation}\label{gradedjac}
    \llbracket u,\llbracket v,w\rrbracket \rrbracket =\llbracket \llbracket u,v\rrbracket ,w\rrbracket +(-1)^{|u||v|}\llbracket v,\llbracket u,w\rrbracket \rrbracket 
\end{equation}
We already know by Lemma \ref{propextension} that the graded Jacobi identity is satisfied on $T_{\leq0}$. Moreover it is trivially satisfied on $T_{+1}$ because the bracket is zero on this space. Then we need only check the cases for which one or two terms belong to $T_{+1}$, that is to say: \textbf{I.}~$u\in T_{+1}$ and $v,w\in T_{\leq0}$, and \textbf{II.}~$u,v\in T_{+1}$ and $w\in T_{\leq0}$. 
 It is sufficient to check both of these Jacobi identities on the generators of $T_{+1}$.

\textbf{I.} First notice that if $v$ or $w$ belongs to $T_0=\mathfrak{g}$ then Equation \eqref{gradedjac} is Equation \eqref{enforcedjacobi}. Thus one can restrict oneself to the situation where $v,w\in T_{\leq-1}$. The proof is made by induction. First assume that $u=\Theta$.
 Then, since $\llbracket \Theta,-\rrbracket=\partial$ (see Equation \eqref{thetadiff}), Equation \eqref{gradedjac} becomes:
\begin{equation}\label{flemme}
    \partial\big(\llbracket v,w\rrbracket\big)=\llbracket\partial(v),w\rrbracket+(-1)^{|v|}\llbracket v,\partial(w)\rrbracket
\end{equation}
which is Equation \eqref{Leibnizz} for $v,w\in T_{\leq-1}$, and has been proved to be satisfied in item 5. above.
Therefore, assume that $u=\Theta_{a}$ for some $a\in\mathfrak{g}$. Then, using Equation \eqref{enforcedjacobi}, one obtains:
\begin{align}
\llbracket\Theta_a,\llbracket v,w\rrbracket\rrbracket&=a\cdot\llbracket\Theta,\llbracket v,w\rrbracket\rrbracket-\llbracket\Theta,a\cdot\llbracket v,w\rrbracket\rrbracket\label{samarcande1}\\
&=a\cdot\big(\llbracket\llbracket\Theta,v\rrbracket, w\rrbracket+(-1)^{|v|}\llbracket v,\llbracket\Theta,w\rrbracket\rrbracket\big)-\llbracket\Theta,\llbracket a\cdot v,w\rrbracket-\llbracket v,a\cdot w\rrbracket\rrbracket\\
&=\llbracket a\cdot\llbracket\Theta,v\rrbracket,w\rrbracket+\llbracket\llbracket\Theta,v\rrbracket,a\cdot w\rrbracket+(-1)^{|v|}\big(\llbracket a\cdot v,\llbracket\Theta,w\rrbracket\rrbracket+\llbracket v,a\cdot \llbracket\Theta,w\rrbracket\rrbracket\big)\label{samarcande1bis}\\
&\hspace{1cm}-\llbracket\llbracket\Theta,a\cdot v\rrbracket,w\rrbracket-(-1)^{|v|}\llbracket a\cdot v,\llbracket\Theta,w\rrbracket\rrbracket-\llbracket\llbracket\Theta,v\rrbracket,a\cdot w\rrbracket-(-1)^{|v|}\llbracket v,\llbracket\Theta,a\cdot w\rrbracket\rrbracket\nonumber\\
&=\llbracket\llbracket\Theta_a,v\rrbracket w\rrbracket+(-1)^{|v|}\llbracket v,\llbracket\Theta_a,w\rrbracket\label{samarcande2}
\end{align}
where, from the first line to the second line we used Equation \eqref{flemme} on the term on the left, and Lemma~\ref{propextension} on the term on the right; from the second to the third line, we used Lemma \ref{propextension} on the first term, and Equation~\eqref{flemme} on the second one; the transition from the third line to the fourth and last one results from cancellation of four terms, and the use of Equation \eqref{enforcedjacobi}.

Now let $n\geq1$ and assume that Equation \eqref{gradedjac} has been shown for every term $u\in T_{+1}$ of the form $\Theta_{a_1\ldots a_n}$, and every $v,w\in T_{\leq-1}$. Chose $n+1$ elements $a_1,\ldots, a_{n+1}\in\mathfrak{g}$, then the proof of the following identity:
\begin{equation}
\llbracket\Theta_{a_1\ldots a_{n+1}},\llbracket v,w\rrbracket\rrbracket=\llbracket\llbracket\Theta_{a_1\ldots a_{n+1}},v\rrbracket w\rrbracket+(-1)^{|v|}\llbracket v,\llbracket\Theta_{a_1\ldots a_{n+1}},w\rrbracket
\end{equation}
relies on the same step as for Equations \eqref{samarcande1}-\eqref{samarcande2}, except that one uses the induction hypothesis in place of Equation \eqref{flemme}. This proves the first case. 
% In that case one can assume that $u=\Theta$, then since $\llbracket \Theta,-\rrbracket=\partial$ Equation \eqref{gradedjac} becomes:
%\begin{equation}
 %   \partial\big(\llbracket v,w\rrbracket\big)=\llbracket\partial(v),w\rrbracket+(-1)^{|v|}\llbracket v,\partial(w)\rrbracket
%\end{equation}
%which is Equation \eqref{Leibnizz} for $v,w\in T_{\leq0}$. Hence it is automatically satisfied. This concludes the proof.

\textbf{II.} Now assume that $u,v\in T_{+1}$. Then, by Equation \eqref{thetasquared}, the graded Jacobi identity reduces to:
\begin{equation}\label{gradedjac2}
  \llbracket u,\llbracket v,w\rrbracket \rrbracket = -\llbracket v,\llbracket u,w\rrbracket \rrbracket 
\end{equation}
Notice that if $w\in T_{0}=\mathfrak{g}$, then it is automatically satisfied because both $\llbracket v,w\rrbracket$ and $\llbracket u,w\rrbracket$ belong to $T_{+1}$. Hence one can assume that $w\in T_{\leq-1}$.
%The proof relies on the following observation, that for every $u=\Theta_{a_1\ldots a_n}$ and $v=\Theta_{b_1\ldots b_m}$, we have the following equalities:
%\begin{equation}
%\llbracket\Theta_{a_1\ldots a_n},\llbracket \Theta_{b_1\ldots b_m},w\rrbracket\rrbracket+\llbracket \Theta_{b_1\ldots b_m},\llbracket \Theta_{a_1\ldots a_n},w\rrbracket\rrbracket=
%-\llbracket\Theta_{b_1a_1\ldots a_n},\llbracket \Theta_{b_1\ldots b_m},w\rrbracket\rrbracket-\llbracket \Theta_{b_2\ldots b_m},\llbracket \Theta_{b_1a_1\ldots a_n},w\rrbracket\rrbracket%=
%-\llbracket\Theta_{a_2\ldots a_n},\llbracket \Theta_{a_1b_1\ldots b_m},w\rrbracket\rrbracket+\llbracket \Theta_{a_1b_1\ldots b_m},\llbracket \Theta_{a_2\ldots a_n},w\rrbracket\rrbracket
%\end{equation}
The proof is made by induction as well. It is sufficient to check Equation \eqref{gradedjac2} on the generators of $T_{+1}$. %Let $m\geq1$ and pick up $m$ elements $b_1,\ldots, b_m$ of $\mathfrak{g}$, and set $v=\Theta_{b_1\ldots b_m}$. This element will stay unchanged during the whole induction. First notice that if $u=\Theta$, then Equation \eqref{gradedjac2} is Equation \eqref{enforcedjacobi2}. Now assume that $u=\Theta_a$ for some $a\in\mathfrak{g}$,
Let $v$ be a fixed element of $T_{+1}$. If $u=\Theta$ then Equation \eqref{gradedjac2} corresponds to Equation \eqref{enforcedjacobi2}. Next, let $u=\Theta_{a}$ for some $a\in\mathfrak{g}$.
 Then, following the same lines of argument as for Equations \eqref{samarcande1}-\eqref{samarcande1bis} and recalling that $\llbracket T_{+1},T_{+1}\rrbracket=0$, one obtains the following identity:
\begin{equation}\label{marre}
\llbracket \Theta_a,\llbracket v,w\rrbracket \rrbracket=-\big(\llbracket a\cdot v,\llbracket\Theta,w\rrbracket\rrbracket+\llbracket v,a\cdot\llbracket\Theta,w\rrbracket\rrbracket\big)+\llbracket a\cdot v,\llbracket\Theta,w\rrbracket\rrbracket+\llbracket v,\llbracket\Theta,a\cdot w\rrbracket\rrbracket
\end{equation}
The right hand side can be recast as $\llbracket v, \llbracket \Theta_a,w\rrbracket\rrbracket$, which is the desired result. Notice that in the present context, where both $u$ and $v$ belong to $T_{+1}$, the use of Equation \eqref{enforcedjacobi2} was crucial to pass from the left-hand side of Equation \eqref{marre} to the right-hand side.

Now let $n\geq1$ and assume that Equation \eqref{gradedjac2} has been proved for every $u$ of the form $\Theta_{a_1\ldots a_n}$. Then pick some $a_1,\ldots, a_{n+1}\in\mathfrak{g}$ so, %and set $u=\Theta_{a_1\ldots a_{n+1}}$. Then, 
using Equation \eqref{enforcedjacobi}, one has:
\begin{align}
\llbracket \Theta_{a_1\ldots a_{n+1}},\llbracket v,w\rrbracket\rrbracket&=a_1\cdot\llbracket \Theta_{a_2\ldots a_{n+1}},\llbracket v,w\rrbracket\rrbracket-\llbracket\Theta_{a_2\ldots a_{n+1}},a_1\cdot\llbracket v,w\rrbracket\rrbracket\\
&=a_1\cdot\big(-\llbracket v,\llbracket\Theta_{a_2\ldots a_{n+1}},w\rrbracket\rrbracket\big)-\llbracket\Theta_{a_2\ldots a_{n+1}},\llbracket a_1\cdot v,w\rrbracket-\llbracket v,a_1\cdot w\rrbracket\rrbracket\\
&=-\llbracket a_1\cdot v,\llbracket\Theta_{a_2\ldots a_{n+1}},w\rrbracket\rrbracket-\llbracket v,a_1\cdot \llbracket\Theta_{a_2\ldots a_{n+1}},w\rrbracket\rrbracket\\
&\hspace{2cm}+\llbracket a_1\cdot v,\llbracket\Theta_{a_2\ldots a_{n+1}},w\rrbracket\rrbracket+\llbracket v,\llbracket\Theta_{a_2\ldots a_{n+1}},a_1\cdot w\rrbracket\rrbracket\nonumber\\
&=-\llbracket v,\llbracket\Theta_{a_1\ldots a_{n+1}},w\rrbracket\rrbracket
\end{align}
where we used three times the induction hypothesis, applied to $\Theta_{a_2\ldots a_{n+1}}$. This proves that Equation \eqref{gradedjac2} is satisfied at level $n+1$. Notice that this result could have been straightforwardly obtained from the generalization of equivalences \eqref{diablo}-\eqref{diablo3} at every level, justified by the fact that we assumed Equation \eqref{enforcedjacobi2} to hold. By induction this proves that Equation \eqref{gradedjac2} is satisfied over the generators of  $T_{+1}$, and thus for every $u,v\in T_{+1}$. This proves that $\llbracket\,.\,,.\,\rrbracket$ is a graded Lie bracket on $\mathbb{T}$.\end{proof}

%and let us use Equation \eqref{enforcedjacobi} on the term $\llbracket u,\llbracket v,w\rrbracket \rrbracket$ to obtain:
%\eqref{gradedjac2}:
%\begin{align}
 %\llbracket \Theta_a,\llbracket \Theta_{b_1\ldots b_m},w\rrbracket \rrbracket&=\rho\big(a;\llbracket \Theta,\llbracket \Theta_{b_1\ldots b_m},w\rrbracket \rrbracket\big)
%\end{align}

%\begin{remark}
%One notices that since the differential $\partial$ coincideswith the adjoint action of $\Theta$, the differential graded Lie algebra structure on $\mathbb{T}$ is equivalent to the graded Lie algebra structure on $\mathbb{T}$ where we forget about the differential. The differential can be recovered through the symbolic formula $\partial=\llbracket\Theta,-\rrbracket$.
% \end{remark}
 
 %\begin{remark}
 %In light of this result, one can notice that the bilinear map $m:\Lambda^2(F_{\leq-1})\to F_{\leq-2}$ is the composition of $\partial_0$ with the Lie bracket. In other words:
 %\end{remark}
 
 The fact that the maps $(\partial_{-i})_{i\geq-1}$ generate $\mathfrak{g}$-modules $(R_{-i})_{i\geq-1}$  satisfying Proposition \ref{propref} is of crucial importance in the proof of Proposition \eqref{prop2}. Indeed this fact allows us to identify the action of the differential $\partial$ and the adjoint action of the embedding tensor $\Theta$, thus proving the Jacobi identities that the graded bracket has to satisfy. Without this identification, one cannot simply extend the graded Lie algebra structure defined on $T_{\leq-1}$ in Proposition \ref{prophierarchy} to the whole complex $\mathbb{T}=T_{\leq-1}\oplus \mathfrak{g}\oplus R_\Theta[-1]$, and should rather limit oneself to extend it to $T_{\leq-1}\oplus \mathfrak{h}$ only, as was emphasized in Remark \ref{remarkhystrion}.
 By overcoming the problem established in this remark, Proposition \eqref{prop2} is thus a great improvement to the situation, since we do not lose any data in the process. This justifies the following characterization:

\begin{definition}
%Let $(\mathfrak{g},V, \Theta)$ be a Lie-Leibniz triple. The (differential) graded Lie algebra $(\mathbb{T}, \llbracket\,.\,,.\,\rrbracket,\partial)$ obtained through Proposition \ref{prophierarchy}, \ref{propextension},\ref{propdiff} and \ref{prop2}, is called the \emph{tensor hierarchy algebra} associated to the Lie-Leibniz triple $(\mathfrak{g},V, \Theta)$.\footnote{Maybe footnote on the appellation THA.}
 Given a Lie-Leibniz triple $(\mathfrak{g},V, \Theta)$, the (differential) graded Lie algebra $(\mathbb{T}, \llbracket\,.\,,.\,\rrbracket,\partial)$ defined in Proposition \ref{prop2} is called the \emph{tensor hierarchy associated to $(\mathfrak{g},V, \Theta)$}. %we have constructed %obtained  through Propositions \ref{prophierarchy}, \ref{propdiff}, and \ref{prop2}.
\end{definition}

\begin{remark}
The tensor hierarchy here defined should not be confused with the `tensor hierarchy algebra' defined in \cite{Palmkvist-THA}, which a priori is \emph{not} a differential graded Lie algebra, but a $\mathbb{Z}$-graded Lie superalgebra with a subspace at degree $+1$ (in our convention) accommodating all possible embedding tensors satisfying the representation constraint.  Moreover, in Section 3.4 of \cite{Bonezzi:2019bek}, the authors explain that although the condition $\llbracket R_{\Theta},R_{\Theta} \rrbracket=0$ defines a completely consistent mathematical model, it is however not general enough for gauged supergravity. 
\end{remark}

When applying this construction to specific cases of Lie-Leibniz triples we obtain the following result, whose proof is immediate:

\begin{prop}
The tensor hierarchy associated to a Lie-Leibniz triple $(\mathfrak{g},V,\Theta)$ for which $V$ is a Lie algebra is precisely the 3-term differential graded Lie algebra obtained in Proposition \ref{experimental}. In particular the tensor hierarchy associated to a Lie algebra crossed module $\mathfrak{c}\xrightarrow{\Theta}\mathfrak{g}$  is the 3-term differential graded Lie algebra:
\begin{center}
\begin{tikzcd}[column sep=1cm,row sep=0.4cm]
\mathfrak{c}\ar[r,"\Theta"]&\mathfrak{g}\ar[r,"0"]&\mathbb{R}[-1]
\end{tikzcd}
\end{center}
\end{prop}

\begin{example}\label{lodayexample2}
We take a particular case of the general Example \ref{lodayexample}. Let $A$ be the associative algebra of 2-by-2 real matrices, i.e. $A=\mathcal{M}_{2\times2}(\mathbb{R})$ and the associative product $\cdot$ is the matrix product $\times$. The associative algebra structure on $A$ induces a Lie algebra structure on $A$ where the Lie bracket $[\,.\,,.\,]_A$ is the commutator of matrices; this is $\mathfrak{gl}_2(\mathbb{R})$.
 Let $D:A\to A$ be the linear projection on the diagonal matrices:
\begin{equation}
D:	\begin{pmatrix}
a_{11}& a_{12}\\
a_{21}&a_{22}\\
\end{pmatrix}\xmapsto{\hspace*{1.2cm}}\begin{pmatrix}
a_{11}& 0\\
0&a_{22}\\
\end{pmatrix}
\end{equation}
This linear map satisfies Equation \eqref{propD}. Then, we can define a Leibniz product on $A$ via  Equation \eqref{majora2}:
\begin{equation}\label{majora3}
\begin{pmatrix}
a_{11}& a_{12}\\
a_{21}&a_{22}\\
\end{pmatrix}\circ \begin{pmatrix}
b_{11}& b_{12}\\
b_{21}&b_{22}\\
\end{pmatrix}=\begin{pmatrix}
0& (a_{11}-a_{22})b_{12}\\
-(a_{11}-a_{22})b_{21}&0\\
\end{pmatrix}
\end{equation}

From now on, when we want to emphasize the algebraic structures equipping $A$, we denote by $\mathfrak{g}$ the Lie algebra structure $(A,[\,.\,,.\,]_A)$ and by $V$ the Leibniz algebra structure $(A,\circ)$. The action of $\mathfrak{g}$ on $V$ (equivalently, on $A$) is materialized via the adjoint action of the Lie bracket. As the image of $D$ consists of the Lie subalgebra $\mathfrak{h}$ of diagonal 2-by-2 matrices, Equation \eqref{majora3} shows that Equation \eqref{eqhalo} is automatically satisfied. Together with Equation \eqref{majora2}, it implies that the triple $(\mathfrak{g},V,D)$ is a Lie-Leibniz triple. Since $\big[a\circ b- D([a,b]_A),c\big]_A\neq0$, by the discussion at the end of Example \ref{lodayexample}, this Lie-Leibniz triple is not semi-strict. Then, let us compute how $\mathfrak{g}$ acts on the embedding tensor $D$. Let $a, b\in A$, if we denote by $\eta:\mathfrak{g}\to\mathrm{End}\big(\mathrm{End}(A)\big)$ the representation of $\mathfrak{g}$ on $\mathrm{End}(A)$ induced by the adjoint action of $\mathfrak{g}$ on $A$, then we have:
\begin{equation}\label{mixmiya}
\eta(a;D)(b)=[a,Db]_A-D([a,b]_A)=-\big(b\circ a-D([b,a]_A)\big)=\begin{pmatrix}
b_{12}a_{21}-a_{12}b_{21}& -(b_{11}-b_{22})a_{12}\\
(b_{11}-b_{22})a_{21}&b_{21}a_{12}-a_{21}b_{12}\\
\end{pmatrix}
\end{equation}
In particular it confirms that the embedding tensor is not $\mathfrak{g}$-equivariant, but that it is $\mathfrak{h}$-equivariant, i.e. $\mathrm{Im}(D)$-equivariant (the right-hand side of Equation \eqref{mixmiya} vanishes when $a$ is a diagonal matrix).

As a consequence, the symmetric bracket $\{\,.\,,.\,\}:S^2A\to A$ is not $\mathfrak{g}$-equivariant and its kernel is not a $\mathfrak{g}$-module. Let us decompose $S^2A$ into $\mathfrak{g}$-modules to find the biggest $\mathfrak{g}$-submodule $K$ of $\mathrm{Ker}\big(\{\,.\,,.\,\}\big)$. A basis of $A$ is given by the four matrices:
\begin{equation}
E_{11}=\begin{pmatrix}
1& 0\\
0&0\\
\end{pmatrix},\quad E_{12}=\begin{pmatrix}
0& 1\\
0&0\\
\end{pmatrix},\quad E_{21}=\begin{pmatrix}
0& 0\\
1&0\\
\end{pmatrix}\quad \text{and}\quad E_{22}=\begin{pmatrix}
0& 0\\
0&1\\
\end{pmatrix}
\end{equation}
The $\mathfrak{g}$-action of a matrix $a=\begin{pmatrix}
a_{11}& a_{12}\\
a_{21}&a_{22}\\
\end{pmatrix}$ on each of them is given by the adjoint action $[a,.\,]_A$:
\begin{align}
[a,E_{11}]_A&=-a_{12}E_{12}+a_{21}E_{21}, &[a,E_{12}]_A&=-a_{21}E_{11}+(a_{11}-a_{22})E_{12}+a_{21}E_{22}\label{joel1}\\
[a,E_{22}]_A&=a_{12}E_{12}-a_{21}E_{21}, &[a,E_{21}]_A&=a_{12}E_{11}-(a_{11}-a_{22})E_{21}-a_{12}E_{22}\label{joel2}
\end{align}

The matrices $E_{ij}$ induce a basis of the 10-dimensional space $S^2A$ via the symmetric product $E_{ij}\odot E_{kl}$. We can alternatively describe this space with degree 2 polynomials of four variables by assigning the variable $X$ to $E_{11}$, $Y$ to $E_{12}$, $Z$ to $E_{21}$ and $T$ to $E_{22}$. Then the action \eqref{joel1}, \eqref{joel2} of $\mathfrak{g}$ on the matrices $E_{ij}$ induces an action on the monomials $X$, $Y$, $Z$, $T$:
\begin{align}
a\cdot X&=-a_{12}Y+a_{21}Z, &a\cdot Y&=-a_{21}X+(a_{11}-a_{22})Y+a_{21}T\label{joel3}\\
a\cdot T&=a_{12}Y-a_{21}Z, &a\cdot Z&=a_{12}X-(a_{11}-a_{22})Z-a_{12}T\label{joel4}
\end{align}
The action of $\mathfrak{g}=\mathfrak{gl}_2(\mathbb{R})$ on $S^2A$ is induced from Equations \eqref{joel3} and \eqref{joel4} by derivation.

Since $\mathfrak{gl}_2(\mathbb{R})$ is reductive, %the decomposition of $S^2(A)=S^2(\mathfrak{gl}_2(\mathbb{R}))$ then the four modules $U_i$ correspond to the decomposition of $S^2(\mathfrak{gl}_2(\mathbb{R}))$ into irreducible $\mathfrak{gl}_2(\mathbb{R})$-modules (equivalently: $\mathfrak{sl}_2(\mathbb{R})$-modules because $[\mathfrak{gl}_2(\mathbb{R}),\mathfrak{gl}_2(\mathbb{R})]=\mathfrak{sl}_2(\mathbb{R})$). 
the symmetric space $S^2A=S^2(\mathfrak{gl}_2(\mathbb{R}))$ is completely reducible into four irreducible $\mathfrak{gl}_2(\mathbb{R})$-submodules (equivalently: $\mathfrak{sl}_2(\mathbb{R})$-modules):
 \begin{align*}
 U_1&=\langle XT-YZ\rangle\\
 U_2&=\langle X^2+T^2+XT+YZ\rangle\\
 U_3&=\big\langle X^2-T^2, XY+YT, XZ+ZT\big\rangle\\
 U_4&=\big\langle Y^2, Z^2, X^2+T^2-2(XT+YZ), XY-YT, XZ-ZT\big\rangle
 \end{align*}
  The generators of the modules form an alternative basis for $S^2A$
because  one can pass from the standard basis of $S^2A$ to the generators of $U_1,U_2,U_3$ and $U_4$ by the following formula:
 \begin{align}
 &\begin{tabular}[c]{@{}c@{}}$a_{X^2}X^2+a_{XY}XY+ a_{XZ}XZ+a_{XT} XT+a_{Y^2}Y^2$\\$+a_{YZ}YZ+a_{YT}YT+a_{Z^2}Z^2+a_{ZT}ZT+a_{T^2}T^2$\end{tabular}=\label{changelev}\\
 &\hspace{6cm}  \begin{tabular}[c]{@{}c@{}}$\frac{a_{XT}-a_{YZ}}{2}(XT-YZ)+\frac{2(a_{X^2}+a_{T^2})+a_{XT}+a_{YZ}}{6}(X^2+T^2+XT+YZ)$\\ $+\frac{a_{X^2}-a_{T^2}}{2}(X^2-T^2)+\frac{a_{XY}+a_{YT}}{2}(XY+YT)+\frac{a_{XZ}+a_{ZT}}{2}(XZ+ZT)$\\$+ a_{Y^2}Y^2+a_{Z^2}Z^2+\frac{a_{X^2}+a_{T^2}-a_{XT}-a_{YZ}}{6}(X^2+T^2-2(XT+YZ))$\\ $+\frac{a_{XY}-a_{YT}}{2}(XY-YT)+\frac{a_{XZ}-a_{ZT}}{2}(XZ-ZT)$\end{tabular}\nonumber
 \end{align}
We then have two equivalent ways of parametrizing $S^2A$, whose respective usefulness is relative to the given computation.

The symmetric product $a\odot b\in S^2A$   of two matrices can be decomposed on the polynomials in $X$, $Y$, $Z$, $T$ as the following:
\begin{align}
a\odot b&=a_{11}b_{11}X^2+\frac{1}{2}(a_{11}b_{12}+b_{11}a_{12})XY+\frac{1}{2}(a_{11}b_{21}+b_{11}a_{21})XZ+\frac{1}{2}(a_{11}b_{22}+b_{11}a_{21})XT+a_{12}b_{12}Y^2\label{kappa}\\
&\hspace{1cm}+\frac{1}{2}(a_{12}b_{21}+b_{12}a_{21})YZ+\frac{1}{2}(a_{12}b_{22}+b_{12}a_{22})YT+a_{21}b_{21}Z^2+\frac{1}{2}(a_{21}b_{22}+b_{21}a_{22})ZT+a_{22}b_{22}T^2\nonumber
\end{align}
By comparing the anti-diagonal matrix elements of the matrix on the right-hand side of the symmetrized version of Equation \eqref{majora3} to the decomposition \eqref{kappa}, we deduce that the symmetric bracket $\{\,.\,,.\,\}:S^2A\to A$ is defined on an element of $S^2A$ as: %sends an element $ a_{X^2}X^2+a_{XY}XY+ a_{XZ}XZ+a_{XT} XT+a_{Y^2}Y^2+a_{YZ}YZ+a_{YT}YT+a_{Z^2}Z^2+a_{ZT}ZT+a_{T^2}T^2$ to the 2-by-2 matrix $\begin{pmatrix}
%0& a_{XY}-a_{YT}\\
%-a_{XZ}+a_{ZT}&0\\
%\end{pmatrix}$.
\begin{equation}\label{symker}
 \begin{tabular}[c]{@{}c@{}}$a_{X^2}X^2+a_{XY}XY+ a_{XZ}XZ+a_{XT} XT+a_{Y^2}Y^2$\\$+a_{YZ}YZ+a_{YT}YT+a_{Z^2}Z^2+a_{ZT}ZT+a_{T^2}T^2$\end{tabular}\xmapsto{\hspace*{0.6cm}\{\,.\,,\,.\,\}\hspace*{0.6cm}} \begin{pmatrix}
0& a_{XY}-a_{YT}\\
-a_{XZ}+a_{ZT}&0\\
\end{pmatrix}
 \end{equation}
 The elements of the left-hand side of Equation \eqref{symker} that do not appear on the right-hand side live in the kernel of the symmetric bracket, which is then the 8-dimensional subspace of $S^2A$ generated by the following vectors:
 \begin{equation*}
 \mathrm{Ker}\big(\{\,.\,,.\,\}\big)=\Big\langle X^2, Y^2, Z^2, T^2, XT, YZ, XY+YT, XZ+ZT\Big\rangle
 \end{equation*}

 The three first $\mathfrak{g}$-submodules $U_1, U_2, U_3$ are contained in $ \mathrm{Ker}\big(\{\,.\,,.\,\}\big)$ while for the last module, only the first three generators belong to the kernel:
 \begin{equation*}
 \mathrm{Ker}\big(\{\,.\,,.\,\}\big)=U_1\oplus U_2\oplus U_3\oplus \big\langle Y^2, Z^2, X^2+T^2-2(XT+YZ)\big\rangle
 \end{equation*}
 As  the action of $\mathfrak{g}$ on the remaining generators $Y^2, Z^2, X^2+T^2-2(XT+YZ)$ leaves the kernel, the biggest $\mathfrak{g}$-submodule of the kernel is the 5-dimensional subspace $K=U_1\oplus U_2\oplus U_3$. We can then identify $T_{-2}=\left(\bigslant{S^2(A)}{K}\right)[2]$, the degree $-2$ space of the tensor hierarchy, with $U_4[2]$ ($U_4$ shifted by $-2$).
 %\begin{equation*}
 %T_{-2}\simeq\big\langle Y^2, Z^2, X^2+T^2-2(XT+YZ), XY-YT, XZ-ZT\big\rangle
 %\end{equation*}

The projection $p:S^2A\to U_4$ then induces the graded Lie bracket on $T_{-1}=A[1]$ taking values in $T_{-2}$ by Equation \eqref{crochet1}, that is to say $\llbracket a,b\rrbracket=q_{-2}(a\odot b)=2\,p(a\odot b)$. We can give an explicit formula of this bracket, thanks to the decomposition \eqref{kappa} and the correspondence~\eqref{changelev}:
\begin{align}
\left\llbracket\begin{pmatrix}
a_{11}& a_{12}\\
a_{21}&a_{22}\\
\end{pmatrix}, \begin{pmatrix}
b_{11}& b_{12}\\
b_{21}&b_{22}\\
\end{pmatrix}\right\rrbracket&= \frac{1}{6}\left(\begin{tabular}[c]{@{}c@{}}$2a_{11}b_{11}-(a_{11}b_{22}+b_{11}a_{21})$\\$+2a_{22}b_{22}-(a_{12}b_{21}+b_{12}a_{21})$\end{tabular}\right)(X^2+T^2-2(XT+YZ))\\
&\hspace{1cm}+\frac{1}{2}\big( (a_{11}b_{12}+b_{11}a_{12})-(a_{12}b_{22}+b_{12}a_{22})\big)(XY-YT)\nonumber\\
&\hspace{1.5cm}+\frac{1}{2}\big( (a_{11}b_{21}+b_{11}a_{21})-(a_{21}b_{22}+b_{21}a_{22})\big)(XZ-ZT)\nonumber\\
&\hspace{2.5cm}+2\,a_{12}b_{12} Y^2+2\,a_{21}b_{21}Z^2\nonumber
\end{align}
The linear map $\partial_{-1}:T_{-2}\to T_{-1}$ is then defined as:
\begin{equation}
\begin{tabular}[c]{@{}c@{}}$pY^2+qZ^2$\\$+r(X^2+T^2-2(XT+YZ))$\\$+s(XY-YT)+t(XZ-ZT)$ \end{tabular}\xmapsto{\hspace*{0.6cm}\partial_{-1}\hspace*{0.6cm}}\begin{pmatrix}
0& 2s\\
-2t&0\\
\end{pmatrix}
\end{equation}
so that we indeed have $\partial_{-1}(\llbracket a,b\rrbracket)=2\{a,b\}$. The map $\partial_{-1}$ has a 3-dimensional kernel.
The tensor hierarchy associated to the Lie-Leibniz triple $(\mathfrak{g},V,D)$ then begins with the following sequence:
\begin{center}
\begin{tikzcd}[column sep=1cm,row sep=0.4cm]
%\ldots\ar[r,"\partial_{-4}"]&T_{-4}\ar[r,"\partial_{-3}"]&T_{-3}\ar[r,"\partial_{-2}"]&T_{-2}\ar[r,"\partial_{-1}"]&V[1]\ar[r,"\partial_0=\Theta"]&\mathfrak{g}
\ldots\ar[r]&U_4[2]\ar[r,"\partial_{-1}"]&A[1]\ar[r,"D"]&A\ar[r,"-\eta(-;D)"]& R_D[-1]
\end{tikzcd}
\end{center}
where $R_D\subset\mathrm{End}(A)$ is the $\mathfrak{g}$-submodule generated by $D$.

In order to define $T_{-3}$, one need to decompose the following diagram in terms of irreducible $\mathfrak{sl}_2(\mathbb{R})$-modules (because they are the irreducible $\mathfrak{gl}_2(\mathbb{R})$-modules):
\begin{equation}\label{digrammatique}
\begin{tikzcd}[column sep=1cm,row sep=1cm]
0\ar[r]&\ar[d,"u_{-3}"] A\otimes (U_1\oplus U_2\oplus U_3)\ar[r,"\mathrm{id}"]&\ar[d, "k_{-3}"] A\otimes (U_1\oplus U_2\oplus U_3)\ar[r]&0\\
S^3(A)\ar[r,"-d_3"]&\ar[d]A\otimes S^2(A)\ar[r,"-d_2"]&\ar[d]F_{-3}\ar[r]&0\\
S^3(A)\ar[r]&A\otimes U_4\ar[r,"q_{-3}"]&T_{-3}\ar[r]&0
\end{tikzcd}
\end{equation}
Here, recall that lines and columns are exact sequences. For any $k\geq0$, let us denote by $V_k$ the irreducible $\mathfrak{sl}_2(\mathbb{R})$-module of highest weight $k$, so in particular it is a $k+1$-dimensional vector space. With this convention, $V_2$ denotes the 3-dimensional adjoint representation of $\mathfrak{sl}_2(\mathbb{R})$, and $U_4=V_4$. Given the discussion above and since $\mathfrak{gl}_2(\mathbb{R})=\mathfrak{sl}_2(\mathbb{R})\oplus \langle I_2\rangle= V_2\oplus V_0$, the decomposition of the second symmetric power of $A=\mathfrak{gl}_2(\mathbb{R})$ in terms of the $\mathfrak{sl}_2(\mathbb{R})$-modules $U_1,U_2,U_3,U_4$ can be rewritten as:
\begin{equation}
S^2(\mathfrak{gl}_2(\mathbb{R}))=S^2(V_2\oplus V_0)=V_0\oplus V_0\oplus V_2 \oplus V_4
\end{equation}
Then, using the classical formulas for reducibility of (symmetric) powers of $\mathfrak{sl}_2(\mathbb{R})$-modules, we have:
\begin{align}
S^3(\mathfrak{gl}_2(\mathbb{R}))&=2\,V_0\oplus 2\,V_2\oplus V_4\oplus V_6\label{souls1}\\
\mathfrak{gl}_2(\mathbb{R})\otimes S^2(\mathfrak{gl}_2(\mathbb{R}))&= 3\, V_0\oplus 5\, V_2\oplus 3\, V_4\oplus V_6\label{souls2}
\end{align}
where the integers in front of each module denote their multiplicities. By exactness of the lines of Diagram~\eqref{digrammatique}, the difference between Lines \eqref{souls2} and \eqref{souls1} gives:
\begin{equation}
F_{-3}=V_0\oplus 3\,V_2\oplus 2\, V_4
\end{equation}
% \begin{equation}(V_2\oplus V_0)\otimes \underbrace{S^2(V_2)}_{=\,V_4\oplus V_0}=(V_2\otimes V_4)\oplus V_2\oplus V_4\oplus V_0=\underbrace{V_2\oplus V_6}_{=\, S^3(V_2)}\oplus V_4\oplus V_2\oplus V_4\oplus V_0\end{equation}
 %\begin{equation}
 %\mathfrak{sl}_2\otimes S^2(\mathfrak{sl}_2)=V_2\otimes (V_4\oplus V_0)=2\,V_2\oplus V_4\oplus V_6
 %\end{equation}

As every space in the present example is completely reducible, the module $T_{-3}$ can be understood as a sub-module of both $F_{-3}$ and $A\otimes U_4=\mathfrak{gl}_2(\mathbb{R})\otimes V_4=V_2\oplus 2\, V_4\oplus V_6$. From this observation we deduce that:
\begin{equation}
T_{-3}\ \subset\  F_{-3}\cap (A\otimes U_4)= V_2\oplus 2\, V_4
\end{equation}
Since $T_{-3}\simeq \bigslant{A\otimes U_4}{S^3(A)}$, by carefully analyzing the intersection of $A\otimes U_4$ with $S^3(\mathfrak{gl}_2(\mathbb{R}))$, we can deduce the explicit decomposition of $T_{-3}$. For example,  %As $\mathfrak{gl}_2(\mathbb{R})=\mathfrak{sl}_2(\mathbb{R})\oplus \langle I_2\rangle$,
as we have the inclusion $A\otimes U_4\subset \mathfrak{gl}_2(\mathbb{R})\otimes S^2(\mathfrak{sl}_2)$ and since the latter tensor product contains $S^3(\mathfrak{sl}_2)=V_2\oplus V_6$, we deduce that the term $V_2\oplus V_6$ in $A\otimes U_4$ corresponds to $S^3(\mathfrak{sl}_2)$ so, by exactness of the bottom line in Diagram \eqref{digrammatique}, $V_2$ cannot appear in $T_{-3}$. Now, there is only one module $V_4$ in $S^3(\mathfrak{gl}_2(\mathbb{R}))$, while there are two in $A\otimes U_4$ so $T_{-3}$ contains at least one module $V_4$, if not two.  The question is then if either one of the two modules $V_4$ in  the decomposition of $A\otimes U_4$ lies in the image of~$S^3(\mathfrak{gl}_2(\mathbb{R}))$.  Since  the two modules $V_4$ in $A\otimes U_4$ are precisely those of $\mathfrak{gl}_2(\mathbb{R})\otimes S^2(\mathfrak{sl}_2(\mathbb{R}))$, it is equivalent to asking if either one of the two submodules $V_4$ of $\mathfrak{gl}_2(\mathbb{R})\otimes S^2(\mathfrak{sl}_2(\mathbb{R}))$ lies in the image of~$S^3(\mathfrak{gl}_2(\mathbb{R}))$.

The following, speculative argument suggests that neither of them lie in this image. First, notice that:
\begin{align}
S^3(\mathfrak{gl}_2(\mathbb{R}))&=S^3(\mathfrak{sl}_2(\mathbb{R}))\oplus S^2(\mathfrak{gl}_2(\mathbb{R}))\label{meluche1}\\
\mathfrak{gl}_2(\mathbb{R})\otimes S^2(\mathfrak{gl}_2(\mathbb{R}))&=\big(\mathfrak{gl}_2(\mathbb{R})\otimes S^2(\mathfrak{sl}_2(\mathbb{R}))\big)\oplus \big(\mathfrak{gl}_2(\mathbb{R})\otimes \mathfrak{gl}_2(\mathbb{R})\big) \label{meluche2}\end{align}
Then, comparing the decompositions \eqref{meluche1} and \eqref{meluche2}, we deduce that the submodule $S^2(\mathfrak{gl}_2(\mathbb{R}))$ of $S^3(\mathfrak{gl}_2(\mathbb{R}))$ sits in the submodule $\mathfrak{gl}_2(\mathbb{R})\otimes \mathfrak{gl}_2(\mathbb{R})$ of $\mathfrak{gl}_2(\mathbb{R})\otimes S^2(\mathfrak{gl}_2(\mathbb{R}))$. Then, since $S^3(\mathfrak{sl}_2)=V_2\oplus V_6$, the submodule $V_4$ of $S^3(\mathfrak{gl}_2(\mathbb{R}))$ is necessarily in $S^2(\mathfrak{gl}_2(\mathbb{R}))$, so it seems to be sent to the submodule $V_4$ of $\mathfrak{gl}_2(\mathbb{R})\otimes \mathfrak{gl}_2(\mathbb{R})$ in Line \eqref{meluche2}. %, and not to any of those of $\mathfrak{gl}_2(\mathbb{R})\otimes \mathfrak{sl}_2(\mathbb{R})$. 
From this, we conjecture that none of the  two submodules $V_4$ of $\mathfrak{gl}_2(\mathbb{R})\otimes S^2(\mathfrak{sl}_2(\mathbb{R}))$ -- and thus, of $A\otimes U_4$ -- lies in the image of~$S^3(\mathfrak{gl}_2(\mathbb{R}))$. Then, Diagram \eqref{digrammatique} can be recast in the following form, where the colors mark the correspondence between the domain and the range of each map (and where the purple color means that it is both red and blue): %We argue that none of them lies in this image for the following reason: the two modules $V_4$ in $A\otimes U_4$ are those of $\mathfrak{gl}_2(\mathbb{R})\otimes S^2(V_2)$, and one comes from $V_2\otimes S^2(V_2)=S^3(V_2)\oplus V_4$ and the other comes from the  
\begin{center}
\begin{tikzcd}[column sep=1cm,row sep=1cm]
0\ar[r]&\ar[d] \textcolor{red}{3\, V_0}\oplus \textcolor{red}{4\, V_2} \oplus \textcolor{red}{V_4}  \ar[r,"\mathrm{id}"]&\ar[d] \textcolor{red}{3\, V_0}\oplus \textcolor{red}{4\, V_2} \oplus \textcolor{red}{V_4}  \ar[r]&0\\
\begin{tabular}[c]{@{}c@{}}$\textcolor{blue}{2\, V_0}\oplus \textcolor{blue}{2\, V_2} $\\$\oplus\, \textcolor{blue}{V_4}\oplus \textcolor{blue}{V_6}$\end{tabular} \ar[r]&\ar[d]\begin{tabular}[c]{@{}c@{}}$\textcolor{violet}{2\, V_0}\oplus\textcolor{red}{V_0}$\\$
\oplus\,\textcolor{blue}{V_2}\oplus\textcolor{violet}{V_2}\oplus \textcolor{red}{3\, V_2} $\\$\oplus\, \textcolor{violet}{V_4} \oplus2\, V_4$\\$\oplus\, \textcolor{blue}{V_6}$\end{tabular}\ar[r]&\ar[d]\textcolor{red}{V_0}\oplus \textcolor{red}{3\,V_2}\oplus2\, V_4\ar[r]&0\\
\begin{tabular}[c]{@{}c@{}}$\textcolor{blue}{2\, V_0}\oplus \textcolor{blue}{2\, V_2} $\\$\oplus\, \textcolor{blue}{V_4}\oplus \textcolor{blue}{V_6}$\end{tabular} \ar[r]&\textcolor{blue}{V_2}\oplus 2\,V_4\oplus \textcolor{blue}{V_6}\ar[r]& 2\,V_4\ar[r]&0
\end{tikzcd}
\end{center}

From this we deduce that $T_{-3}=(V_4\oplus V_4)[3]$. The Lie bracket between $T_{-2}=V_4[2]$ and $T_{-1}=A[1]$ is uniquely defined as the bottom horizontal line, while the linear map $\partial_{-2}:T_{-3}\to T_{-2}$ is also uniquely defined by Lemma \ref{lemmadiff}. This gives the next level of the tensor hierarchy associated to the Lie-Leibniz triple $(\mathfrak{g},V,D)$:
\begin{center}
\begin{tikzcd}[column sep=1cm,row sep=0.4cm]
%\ldots\ar[r,"\partial_{-4}"]&T_{-4}\ar[r,"\partial_{-3}"]&T_{-3}\ar[r,"\partial_{-2}"]&T_{-2}\ar[r,"\partial_{-1}"]&V[1]\ar[r,"\partial_0=\Theta"]&\mathfrak{g}
\ldots\ar[r]&(V_4\oplus V_4)[3]\ar[r,"\partial_{-2}"]&V_4[2]\ar[r,"\partial_{-1}"]&A[1]\ar[r,"D"]&A\ar[r,"-\eta(-;D)"]& R_D[-1]
\end{tikzcd}
\end{center}
We will not provide the next space of the hierarchy as the computations become more cumbersome, but in principle there is no mathematical obstruction to this task.
\end{example}
  %is completely reducible:
 %$U_4=V_4$ is a submodule of $S^2(\mathfrak{sl}_2(\mathbb{R}))$
%The question is then if, beyond the pair of terms $V_2\oplus V_6$, either one of the two modules $V_4$ in  the right-hand side lives in the image of $S^3(\mathfrak{gl}_2(\mathbb{R}))$.
%We will not go further in the computation given that it becomes more and more cumbersome, but will provide in the next example how it could be performed in principle.

%\begin{example}
%Let us apply this construction to a Lie algebra crossed module $\mathfrak{c}\xrightarrow{\Theta}\mathfrak{g}$. Example \ref{examplediffercro} tells us that $T_{\leq0}=\mathfrak{g}\oplus \mathfrak{c}[1]$. Equation \eqref{equatdiff2} implies that the embedding tensor $\Theta$ is $\mathfrak{g}$-equivariant, which implies in turn that it belongs to the $\mathfrak{g}$-module $\mathbb{R}$ on which $\mathfrak{g}$ acts trivially. Hence by Proposition \ref{prop2}, one deduces that the following cochain complex:
%\begin{center}
%\begin{tikzcd}[column sep=1cm,row sep=0.4cm]
%V\ar[r,"\Theta"]&\mathfrak{g}\ar[r,"0"]&\mathbb{R}
%\end{tikzcd}
%\end{center}
%is the tensor hierarchy associated to the Lie algebra crossed module $(\mathfrak{g},\mathfrak{c},\Theta)$. 
%They happen to be  equivalent.
%\end{example}

\subsection{Proof of Theorem 1 and open questions}\label{restriction}

Since the tensor hierarchy $\mathbb{T}$ associated to a given Lie-Leibniz triple $(\mathfrak{g},V,\Theta)$ includes $\mathfrak{g}$ and $V$ as $T_{0}$ and $T_{-1}$, we immediately have the following result, which is the first part of Theorem \ref{centraltheorem}:

\begin{prop}\label{prop26}
The function $G:\emph{\textbf{Lie-Leib}}\to\emph{\textbf{DGLie}}_{\leq1}$ associating a tensor hierarchy to a given Lie-Leibniz triple is injective-on-objects.
\end{prop}

Thus, tensor hierarchies can be seen as Lie-ifications of Lie-Leibniz triples, where no information has been lost, generalizing what occurs for Lie algebra crossed modules.
We now show that, when restricted to a particular subcategory of \textbf{Lie-Leib}, this function is a functor. %on the category \textbf{Lie-Leib}
%In the next section we will use our construction to build $L_\infty$-algebra extensions of Leibniz algebras.
%by showing
The following Proposition is the second part of Theorem \ref{centraltheorem}:
\begin{prop}\label{prop7}
The restriction of the  function $G:\emph{\textbf{Lie-Leib}}\to\emph{\textbf{DGLie}}_{\leq1}$ to \emph{\textbf{compLie-Leib}} is a faithful functor. Moreover, \textbf{\emph{compLie-Leib}} is the biggest wide subcategory of \textbf{\emph{Lie-Leib}} (with respect to inclusion) such that this functorial property holds.
\end{prop}

%\begin{remark}
%Actually, the function $G$ is a faithful functor on the more general sub-category of \emph{\textbf{Lie-Leib}} having the property that $\mathrm{Ker}\big(\{\,.\,,.\,\}\big)$ is a $\mathfrak{g}$-module. This includes semi-strict Lie-Leibniz triples, but also Lie-Leibniz triples $(\mathfrak{g},V,\Theta)$ for which $V$ is a Lie algebra.
%\end{remark}

%For now, we have proved that one can associate to any Lie-Leibniz triple a differential graded Lie algebra called its associated tensor hierarchy.
%Let us prove now that this assignement is functorial on $\textbf{semLie-Leib}$, i.e. that a morphism of semi-strict Lie-Leibniz triples induces a dgLa morphism between their tensor hierarchies. Thus,  in the following, we assume that all Lie-Leibniz triples are semi-strict, so that we work in the category \textbf{semLie-Leib}. 
%a morphism of Leibniz algebra is canonically transported to a morphism of dgLas between their associated tensor hierarchies. 

\begin{proof} 
Let us first prove that the assignment $G$ is functorial on $\textbf{compLie-Leib}$, i.e. that a compatible morphism of  Lie-Leibniz triples (see Definition \ref{defcomp}) induces a dgLa morphism between their tensor hierarchies.  Faithfulness will follow from the definition of the action of $G$ on morphisms.
 % We will show that any morphism of  Lie-Leibniz triples $(\varphi,\chi):(\mathfrak{g},V,\Theta)\to(\mathfrak{g}',V',\Theta')$  canonically induces a morphism of differential graded Lie algebras $\phi:(T,\llbracket\,.\,,.\,\rrbracket,\partial)\to(T',\llbracket\,.\,,.\,\rrbracket',\partial')$ whose restrictions to $T_0=\mathfrak{g}$ and $T_{-1}=V[1]$ satisfy:
%\begin{equation}
 %   \phi|_{T_0}=\varphi\hspace{1cm}\text{and}\hspace{1cm}
 %   \phi|_{T_{-1}}=\chi
%\end{equation}

% The proof is made by induction.
  Let $(\mathfrak{g},V,\Theta)$ and $(\mathfrak{g}',V',\Theta')$ be Lie-Leibniz triples, and let $(\mathbb{T},\llbracket\,.\,,.\,\rrbracket,\partial)$ and  $(\mathbb{T}',\llbracket\,.\,,.\,\rrbracket',\partial')$ be their  associated tensor hierarchies.
 Let us introduce some notation: 
  for every $i\geq1$, let $T_{(-i)}$ be the direct sum $\bigoplus_{0\leq j\leq i} T_{-j}$, the same convention applying to $\mathbb{T}'$. 
In the proof, $\phi_{-i}$ will always denote a degree 0 linear map from $T_{-i}$ to $T'_{-i}$, and $\phi_{(-i)}$ will denote the unique  degree 0 linear map from $T_{(-i)}$ to $T'_{(-i)}$ that restricts to $\phi_{-i}$ on $T_{-i}$. It can be straightforwardly extended to $\Lambda^\bullet (T_{(-i)})$ as a morphism of graded algebras. 
 
 \bigskip
 \textbf{Definition of $\phi_{+1}$:} Let $(\varphi, \chi):(\mathfrak{g},V,\Theta)\to(\mathfrak{g}',V',\Theta')$ be any (compatible) morphism of Lie-Leibniz triples, and let us denote $\phi_{0}=\varphi$ and $\phi_{-1}=\chi:V[1]\to V'[1]$. The compatibility of the respective graded Lie brackets on $\mathbb{T}$ and $\mathbb{T}'$ with $\phi_0$ and $\phi_{-1}$ is deduced from the fact that $\phi_0$ is a Lie algebra morphism and by Equation~\eqref{eq:jfk2}:
  \begin{align}
      \big\llbracket \phi_{0}(a),\phi_{0}(b)\big\rrbracket&=\phi_{0}\big(\llbracket a,b\rrbracket\big)\\
       \big\llbracket \phi_{0}(a),\phi_{-1}(x)\big\rrbracket&=\phi_{-1}\big(\llbracket a,x\rrbracket\big)
  \end{align}
  for every $a,b\in T_{0}$ and $x\in T_{-1}$. % such that $-1\leq|x|+|y|\leq0$. %Recall that we have %set the convention that 
  %$\phi_{(-1)}= \phi_{0}\oplus \phi_{-1}:T_{(-1)}\to T'_{(-1)}$. 
  %For a  semi-strict Lie-Leibniz triple $(\mathfrak{g}, V,\Theta)$, %the action of $\mathfrak{g}$ on $\Theta$ is zero, so that $R_\Theta=\mathbb{R}$ and $\partial_{+1}=0$.  
Now, let $R_\Theta\subset\mathrm{Hom}(V,\mathfrak{g})$ and $R_{\Theta'}\subset\mathrm{Hom}(V',\mathfrak{g}')$ the cyclic $\mathfrak{g}$-modules generated by $\Theta$ and $\Theta'$, respectively (see Equation \eqref{defrep}). %We denote by $\eta:\mathfrak{g}\to\mathrm{End}(R_\Theta)$ and $\eta':\mathfrak{g}'\to\mathrm{End}(R_{\Theta'})$ the respective action of $\mathfrak{g}$ and $\mathfrak{g}'$ on these modules. 
Then we define the linear map $\phi_{+1}:R_\Theta[-1]\to R_{\Theta'}[-1]$ to be the unique morphism of cyclic $\mathfrak{g}$-modules sending $\Theta$ to $\Theta'$ and compatible with $\phi_0$. That is to say, since $\Theta$ (resp. $\Theta'$) generates $R_\Theta$ (resp. $R_\Theta'$), we have:
\begin{equation}\label{defphi1}
    \phi_{+1}\big(\Theta_{a_1\ldots a_n}\big)=\Theta'_{\varphi(a_1)\ldots\varphi(a_n)}
\end{equation}
where $\Theta_{a_1\ldots a_n}=a_1\cdot (a_2\cdot(\ldots(a_n\cdot\Theta)\ldots))$, for every $a_1,\ldots, a_n\in \mathfrak{g}$.
Moreover, %the fact that $\phi_{+1}(\Theta)=\Theta'$ and S
since $R_{-i}$ (resp. $R'_{-i}$) is generated by $\partial_{-i}$ (resp. $\partial'_{-i}$) -- see Equation \eqref{defrep2}  -- the fact that $\phi_{+1}(\Theta)=\Theta'$ implies that the map $\phi_{+1}$ induces a morphism of cyclic $\mathfrak{g}$-modules from $R_{-i}$ to $R'_{-i}$ -- also denoted $\phi_{+1}$ for convenience -- such that:
\begin{equation}\label{bouteille}
\phi_{+1}(\partial_{-i})=\partial'_{-i}
\end{equation}
%so that Equation \eqref{eq:jfk} reads: 
%  \begin{equation}\label{compatdiff}
 %     \partial'_{0}\circ\phi_{-1}=\phi_{0}\circ\partial_{0}
 % \end{equation}
  %The similar equation involving $\partial_{+1}$ is trivially satisfied since the latter vanishes. 
 %we deduce that the map $\phi_{+1}|_{R_{-i}}:R_{-i}\to R'_{-i}$ is surjective.

So far, there are two compatibility equations involving $\phi_{+1}$ and the graded Lie bracket. The first one is of the form:
  \begin{equation}\label{eq:pasdid}
      \big\llbracket \phi_{+1}(\Theta),\phi_{-1}(x)\big\rrbracket=\phi_{0}\big(\llbracket\Theta,x\rrbracket\big)
  \end{equation}
  for any $x\in T_{-1}$. It is automatically satisfied because it is a rewriting of Equation \eqref{eq:jfk}, under the convention that $\partial=\llbracket\Theta,-\rrbracket$. The second equation is:
   \begin{equation}\label{eq:pasdid2}
      \big\llbracket \phi_{+1}(\Theta),\phi_{0}(a)\big\rrbracket=\phi_{+1}\big(\llbracket\Theta,a\rrbracket\big)
  \end{equation}
  for every $a\in T_0$. It is also satisfied because it can be equivalently rewritten as 
  \begin{equation}
  -\varphi(a)\cdot\Theta'=\phi_{+1}(-a\cdot\Theta)
  \end{equation}
   which is a particular case of Equation \eqref{defphi1}.
   Using successively Equations \eqref{defphi1}, \eqref{enforcedjacobi} and \eqref{eq:jfk2}, together with the fact that $\varphi$ is a morphism of Lie algebras, one can show that Equations \eqref{eq:pasdid} and \eqref{eq:pasdid2} imply that the following two equations hold for any generator $\Theta_{a_1\ldots a_m}$ of $T_{+1}=R_\Theta[-1]$:
   \begin{align}
       \big\llbracket \phi_{+1}(\Theta_{a_1\ldots a_m}),\phi_{-1}(x)\big\rrbracket&=\phi_{0}\big(\llbracket\Theta_{a_1\ldots a_m},x\rrbracket\big)\label{eqraja1}\\
     \big\llbracket \phi_{+1}(\Theta_{a_1\ldots a_m}),\phi_{0}(a)\big\rrbracket&=\phi_{+1}\big(\llbracket\Theta_{a_1\ldots a_m},a\rrbracket\big)\label{eqraja2}
   \end{align}
    Hence, the first three maps $\phi_{+1}$, $\phi_0$ and $\phi_{-1}$ are so far compatible with the graded Lie brackets on $\mathbb{T}$ and $\mathbb{T}'$, respectively. 
Now, because of Equations \eqref{thetadiff} and \eqref{bouteille}, one can check that the identity $\partial_{0}'\circ \phi_{-1}=\phi_{0}\circ \partial_{0}$ corresponds to Equation
\eqref{eq:pasdid}, and that $\partial_{+1}'\circ \phi_{0}=\phi_{+1}\circ \partial_{+1}$ corresponds to Equation \eqref{eq:pasdid2}.
  
  \bigskip
 \textbf{Definition of $\phi_{-i}$, $i\geq2$:} The proof is made by induction. Let $T_{-2}=\left(\bigslant{S^2(V)}{K}\right)[2]$ and  $T_{-2}'=\left(\bigslant{S^2(V')}{K'}\right)[2]$ . Since the morphism $(\varphi, \chi):(\mathfrak{g},V,\Theta)\to(\mathfrak{g}',V',\Theta')$ is a compatible morphism, Equation~\eqref{comstring} 
 implies that $\chi\odot\chi$ passes to the quotient, i.e. that it induces a well-defined map $\phi_{-2}:T_{-2}\to T'_{-2}$ that renders the following prism commutative:
\begin{center}
\begin{tikzpicture}[on top/.style={preaction={draw=white,-,line width=#1}},
on top/.default=5pt]
\matrix(a)[matrix of math nodes, 
row sep=2em, column sep=3em, 
text height=1.5ex, text depth=0.25ex] 
{ &T_{-2}&&&T'_{-2}\\
&&&&\\
\Lambda^2(T_{-1})&&&\Lambda^2(T'_{-1})&\\ 
&T_{-1}&&&T'_{-1}\\}; 
\path[->](a-3-1) edge node[above right]{$\phi_{-1}\wedge \phi_{-1}$} (a-3-4); %node[below right]{$\varphi$}  (a-1-2); 
\path[->>](a-3-1) edge node[above left]{$q_{-2}$} (a-1-2);
\path[->](a-3-1) edge node[below left]{$m_{-1}$} (a-4-2);
\path[->](a-3-4) edge node[above right]{$m'_{-1}$} (a-4-5);
\path[->>](a-3-4) edge node[above left]{$q'_{-2}$} (a-1-5);
\path[->](a-1-2) edge node[above]{$\phi_{-2}$} (a-1-5);
\path[->](a-4-2) edge node[above]{$\phi_{-1}$} (a-4-5);
\path[->](a-1-2) edge[on top=5pt] node[left]{$\partial_{-1}$} (a-4-2);
\path[->](a-1-5) edge node[right]{$\partial_{-1}'$} (a-4-5);
\end{tikzpicture}
\end{center}
where we have written  $\chi$ under its current denomination  $\phi_{-1}:T_{-1}\to T'_{-1}$. This is precisely here that we use the compatibility property of the morphism. We refer to the end of the proof for a discussion of what happens when the morphism is not compatible.
%The map $q_{-2}$ has been defined in subsection \ref{construction} and the map $m_{-1}$ has been defined in subsection \ref{structure}.  
%where $\partial_{-1}$ is the unique map factorizing the symmetric bracket through $\bigslant{S^2(V)}{\mathrm{Ker}\big(\{\,.\,,.\,\}\big)}$:  %(resp. $\bigslant{S^2(V')}{\mathrm{Ker}\big(\{\,.\,,.\,\}'\big)}$):
%\begin{equation}
 %   \{\,.\,,.\,\}=\delta\circ p
%\end{equation}
%and respectively, $\delta'$ is the unique map factorizing the symmetric bracket through $\bigslant{S^2(V')}{\mathrm{Ker}\big(\{\,.\,,.\,\}'\big)}$ 
%Let $(T,\llbracket\,.\,,.\,\rrbracket,\partial)$ and $(T',\llbracket\,.\,,.\,\rrbracket',\partial')$ be the two differential graded Lie algebras respectively associated to $(\mathfrak{g},V,\Theta)$ and $(\mathfrak{g}',V',\Theta')$ by Proposition 1.
 %We will now show that there exists a canonical dgLa morphism $\phi:T\to T'$ such that $\phi|_{T_0}=\varphi$, $\phi|_{T_{-1}}=\chi$, and $\phi|_{T_{-2}}=\tau$.  

The commutative diagram proves the following two equations:
 \begin{align}
\partial'_{-1}\circ\phi_{-2}&=\phi_{-1}\circ\partial_{-1}\label{eq:didpas}\\
\phi_{-2}\big(\llbracket x,y\rrbracket\big)&=\big\llbracket\phi_{-1}(x),\phi_{-1}(y)\big\rrbracket
\end{align} 
for every $x,y\in T_{-1}$. The fact that $\chi$ satisfies Equation \eqref{eq:jfk2} implies that $\phi_{-2}$ in turn satisfies:
\begin{equation}\label{qqqsd}
   \phi_{-2}\big(\llbracket a,x\rrbracket\big)=\big\llbracket\phi_{0}(a),\phi_{-2}(x)\big\rrbracket
\end{equation}
for every $a\in \mathfrak{g}$ and $x\in T_{-2}$.
The last compatibility equations to check are:
  \begin{align}
      \big\llbracket \phi_{+1}(\Theta),\phi_{-2}(x)\big\rrbracket&=\phi_{-1}\big(\llbracket\Theta,x\rrbracket\big)\label{eq:pasdid3}\\
       \big\llbracket \phi_{+1}(\Theta_{a_1\ldots a_m}),\phi_{-2}(x)\big\rrbracket&=\phi_{-1}\big(\llbracket\Theta_{a_1\ldots a_m},x\rrbracket\big)\label{eq:pasdid3bis}
  \end{align}
 for every $x\in T_{-2}$. By Equations \eqref{thetadiff} and \eqref{bouteille}, the first equation is an alternative form of Equation \eqref{eq:didpas}.  Equation \eqref{eq:pasdid3bis} can be proved by using Equations \eqref{eq:jfk2}, \eqref{enforcedjacobi}, \eqref{defphi1},  \eqref{qqqsd}, as well as Equation \eqref{eq:pasdid3}. For example, for $m=1$, one has:
 \begin{align}
\big\llbracket \phi_{+1}(\Theta_{a}),\phi_{-2}(x)\big\rrbracket&=\llbracket \Theta'_{\varphi(a)},\phi_{-2}(x)\big\rrbracket\label{bientot}\\
&=\varphi(a)\cdot \llbracket \Theta',\phi_{-2}(x)\rrbracket-\llbracket \Theta', \varphi(a)\cdot \phi_{-2}(x)\rrbracket\\
&=\varphi(a)\cdot \phi_{-1}\big(\llbracket\Theta,x\rrbracket\big)-\llbracket\Theta',\phi_{-2}(a\cdot x)\rrbracket\\
&=\phi_{-1}\big(a\cdot\llbracket\Theta,x\rrbracket\big)-\phi_{-1}\big(\llbracket\Theta,a\cdot x\rrbracket\big)\\
&=\phi_{-1}\big(\llbracket\Theta_{a},x\rrbracket\big)\label{bientot2}
 \end{align}

Lets now turn to the inductive step.
Assume that the linear map $\phi_{(-i)}:T_{(-i)}\to T'_{(-i)}$ is constructed up to order $i\geq2$, and that is satisfies:
 \begin{align}
\partial'\circ\phi_{(-i)}(x)&=\phi_{(-i)}\circ\partial(x)\label{eq:morphism3}\\
\phi_{(-i)}\big(\llbracket x,y\rrbracket\big)&=\big\llbracket\phi_{(-i)}(x),\phi_{(-i)}(y)\big\rrbracket\label{eq:morphism4}
\end{align} 
 for every $x,y\in T_{\leq0}$  such that $-i\leq|x|+|y|\leq0$.
 Moreover, notice that for $a\in T_0=\mathfrak{g}$, and $u\in T_{\leq0}$, Equation~\eqref{eq:morphism4} implies the following result:
 \begin{equation}\label{eq:morphism5}
     \phi_{(-i)}(a\cdot u)=\phi_0(a)\cdot \phi_{(-i)}(u)
 \end{equation}
 %where $\rho$ (resp. $\rho'$) represents the action of $\mathfrak{g}$ (resp. $\mathfrak{g}'$) on the tensor hierarchy $\mathbb{T}$ (resp. $\mathbb{T}'$). In particular, t
 This equation extends to every element $y\in \Lambda(T_{\leq-1})$. 
  The only remaining non-trivial equations are:
 \begin{align}\big\llbracket \phi_{+1}(\Theta),\phi_{-i}(x)\big\rrbracket&=\phi_{-i+1}\big(\llbracket\Theta,x\rrbracket\big)\label{nontriv}\\
\big\llbracket \phi_{+1}(\Theta_{a_1\ldots a_m}),\phi_{-i}(x)\big\rrbracket&=\phi_{-i+1}\big(\llbracket\Theta_{a_1\ldots a_m},x\rrbracket\big)\label{nontriv2} %\hspace{1cm}\text{for every $x\in T_{-i}$}
  \end{align} for every $x\in T_{-i}$. By Equations \eqref{thetadiff} and \eqref{bouteille}, Equation \eqref{nontriv} is automatically satisfied because it is  equivalent to Equation \eqref{eq:morphism3}, while Equation \eqref{nontriv2} is a consequence of the former as well as \eqref{eq:morphism5}, as can be shown by using the same kind of arguments as for Equations \eqref{bientot}-\eqref{bientot2}.

Let us construct $\phi_{-i-1}:T_{-i-1}\to T'_{-i-1}$ such that Equations \eqref{eq:morphism3} and \eqref{eq:morphism4} are valid one step higher, i.e. such that $i$ is replaced by $i+1$. Then, as explained above, Equations \eqref{nontriv} and \eqref{nontriv2} would be obtained as a mere consequence of the former.    %the two following equations are satisfied:
 %\begin{align}
%\partial'_{(-2)}\circ\phi_{(-3)}(a)&=\phi_{(-2)}\circ\partial_{(-2)}(a)\label{eq:morphism1}\\
%\phi_{(-3)}\big([a,b]\big)&=\big[\phi_{(-3)}(a),\phi_{(-3)}(b)\big]\label{eq:morphism2}
%\end{align} 
% for every $a,b\in T_{(-2)}$  such that $|a|+|b|\geq-3$. 
Since, by Equation \eqref{crochet1}, the graded Lie bracket $\llbracket \,.\,,.\,\rrbracket$ (resp. $\llbracket \,.\,,.\,\rrbracket'$)  coincides with $q$ (resp. $q'$), %proving Equation \eqref{eq:morphism4} for $i+1$ is equivalent to showing that the following diagram is commutative:
% and where it is understood that $\phi_{(-i)}$ acts a an algebra morphism.
%where it is understood that $\phi_{(-i)}$ acts a an algebra morphism. 
%When $i=2$, this is a consequence of the definition of $\phi_{-2}$:
%\begin{equation}
 %   \phi_{-2}\circ q_{-2}=q'_{-2}\circ \phi_{-1}\label{eq:tau1}
%\end{equation}
%which was proved in the first step.
 %Thus, 
 Equation \eqref{eq:morphism4} can be rewritten as:
\begin{equation}
    \phi_{(-i)}\circ q= q'\circ \phi_{(-i)}
\end{equation}
and straightforwardly extended to $\Lambda^3(T_{\leq-1})|_{-i-1}$ (in particular, see Equation \eqref{extensionQ}), making the following diagram commutative :
 \begin{center}
\begin{tikzpicture}
\matrix(a)[matrix of math nodes, 
row sep=5em, column sep=5em, 
text height=1.5ex, text depth=0.25ex] 
{
\Lambda^3(T_{\leq-1})|_{-i-1}&\Lambda^3(T'_{\leq-1})|_{-i-1}\\
\Lambda^2(T_{\leq-1})|_{-i-1}&\Lambda^2(T'_{\leq-1})|_{-i-1}\\
}; 
\path[->](a-1-1) edge node[above]{$\phi_{(-i)}$} (a-1-2); 
\path[->](a-1-1) edge node[left]{$q$}  (a-2-1); 
\path[->](a-2-1) edge node[above]{$\phi_{(-i)}$} (a-2-2); 
\path[->](a-1-2) edge node[right]{$q'$}  (a-2-2); 
\end{tikzpicture}
\end{center}
where it is understood that $\phi_{(-i)}$ acts as an algebra morphism.
Commutativity of the above diagram, together with the fact that $\mathrm{Im}\big(q|_{\Lambda^3(T_{\leq-1})}\big)=\mathrm{Ker}\big(q|_{\Lambda^2(T_{\leq-1})}\big)$ (and a similar equality for $q'$), implies in turn that: 
\begin{equation*}
    \phi_{(-i)}\big(\mathrm{Ker}\big(q|_{\Lambda^2(T_{\leq-1})|_{-i-1}}\big)\big)\subset \mathrm{Ker}\big(q'|_{\Lambda^2(T'_{\leq-1})|_{-i-1}}\big)
    \end{equation*}

Since $\bigslant{\Lambda^2(T_{\leq-1})|_{-i-1}}{\mathrm{Ker}\big(q|_{\Lambda^2(T_{\leq-1})|_{-i-1}}\big)}\simeq T_{-i-1}$ (and a similar identity for $T'_{-i-1}$), this implies that   $\phi_{(-i)}$ passes to the quotient, % $\bigslant{\Lambda^2(T_{\leq-1})|_{-i-1}}{\mathrm{Ker}\big(q|_{\Lambda^2(T_{\leq-1})|_{-i-1}}\big)}\simeq T_{-i-1}$, 
and induces a linear map $\phi_{-i-1}:T_{-i-1}\to T'_{-i-1}$ which makes the following diagram commutative: 
 %Hence, $\phi_{(-2)}\big(\mathrm{Im}(q_{(-1)})\big)\subset \mathrm{Im}(q'_{(-1)})$, so that the map $\phi_{(-2)}$ passes to the quotient and canonically defines a map $\phi_{-3}:T_{-3}\to T'_{-3}$, such that the following diagram is commutative:
 \begin{center}
\begin{tikzpicture}
\matrix(a)[matrix of math nodes, 
row sep=5em, column sep=5em, 
text height=1.5ex, text depth=0.25ex] 
{
\Lambda^2(T_{\leq-1})|_{-i-1}&\Lambda^2(T'_{\leq-1})|_{-i-1}\\
T_{-i-1}&T'_{-i-1}\\
%T_{-2}&T'_{-2}\\
}; 
\path[->](a-1-1) edge node[above]{$\phi_{(-i)}$} (a-1-2); 
\path[->](a-1-1) edge node[left]{$q$}  (a-2-1); 
\path[->](a-2-1) edge node[above]{$\phi_{-i-1}$} (a-2-2); 
\path[->](a-1-2) edge node[right]{$q'$}  (a-2-2); 
%\path[->](a-2-1) edge node[left]{$\partial_{-2}$}  (a-3-1); 
%\path[->](a-3-1) edge node[above]{$\phi_{-2}$} (a-3-2); 
%\path[->](a-2-2) edge node[right]{$\partial'_{-2}$}  (a-3-2); 
\end{tikzpicture}
\end{center}
Replacing back $q$ by $\llbracket\,.\,,.\,\rrbracket$, commutativity of the diagram is equivalent to the following equation:
\begin{equation}\label{direct}
    \phi_{-i-1}\big(\llbracket x,y\rrbracket\big)=\big\llbracket\phi_{(-i)}(x),\phi_{(-i)}(y)\big\rrbracket
\end{equation}
for every $x,y\in T_{\leq-1}$ such that $|x|+|y|=-i-1$.

The similar condition involving an element of $\mathfrak{g}$ and an element of $T_{-i-1}$ comes from the fact that $q:\Lambda^2(T_{\leq-1})\to T_{\leq-2}$ (resp. $q'$) is a surjective $\mathfrak{g}$-equivariant (resp. $\mathfrak{g}'$-equivariant) map. More precisely, let $a\in\mathfrak{g}$ and $x\in T_{-i-1}$, then there exists $u\in\Lambda^2(T_{\leq-1})|_{-i-1}$ such that $x=q(u)$. But then, %using  Equation \eqref{eq:morphism5}, 
one obtains:
\begin{align}
    \phi_{-i-1}\big(\llbracket a,x\rrbracket\big)&=\phi_{-i-1}(a\cdot q(u))\\
    &= \phi_{-i-1}(q(a\cdot u))=q'(\phi_{(-i)}(a\cdot u))\\
    &=q'\big(\phi_0(a)\cdot \phi_{(-i)}(u)\big)=\phi_0(a)\cdot q'(\phi_{(-i)}(u))\\
    &=\phi_0(a)\cdot\phi_{-i-1}(q(u))=\big\llbracket \phi_0(a),\phi_{-i-1}(x)\big\rrbracket
\end{align}
where we used successively Equation \eqref{crochet2}, the $\mathfrak{g}$-equivariance of $q$, Equation \eqref{direct}, Equation \eqref{eq:morphism5}, $\mathfrak{g}'$-equivariance of $q'$, Equation \eqref{direct} again, and finally Equation \eqref{crochet2}. The proof does not depend on the choice of preimage $u$ we make because $\mathrm{Ker}\big(q|_{\Lambda^2(T_{\leq-1})|_{-i-1}}\big)$ is stable under the action of $\mathfrak{g}$. Thus, we have extended Equation \eqref{eq:morphism4} one level higher:
\begin{equation}
\phi_{(-i-1)}\big(\llbracket x,y\rrbracket\big)=\big\llbracket\phi_{(-i-1)}(x),\phi_{(-i-1)}(y)\big\rrbracket\label{eq:morphismter}
\end{equation} 
 for every $x,y\in T_{\leq0}$  such that $-i-1\leq|x|+|y|\leq0$.

There are two more compatibility equations between $\phi_{-i-1}$ and the graded Lie bracket to check, that is:
\begin{align}
    \phi_{-i}\big(\llbracket\Theta,x\rrbracket\big)&=\big\llbracket\phi_{+1}(\Theta),\phi_{-i-1}(x)\big\rrbracket\label{finalequ}\\
     \phi_{-i}\big(\llbracket\Theta_{a_1\ldots a_m},x\rrbracket\big)&=\big\llbracket\phi_{+1}(\Theta_{a_1\ldots a_m}),\phi_{-i-1}(x)\big\rrbracket\label{finalequ2}
\end{align}
for every $x\in T_{-i-1}$ and every $a_1,\ldots, a_m\in\mathfrak{g}$. Equation \eqref{finalequ} can be rewritten in the form of Equation \eqref{eq:morphism3}, at level $-i-1$:
\begin{equation}
    \phi_{-i}\big(\partial(x)\big)=\partial'\big(\phi_{-i-1}(x)\big)
\end{equation}
%so it is necessary and sufficient to prove this equation to finish the proof of the inductive step.
This is indeed satisfied; by chosing a pre-image $u\in\Lambda^2(T_{\leq-1})_{-i-1}$ of $x$, one has the following result:
\begin{align}
    \partial'\big(\phi_{-i-1}(x)\big)&=\partial'\big(\phi_{-i-1}(q(u))\big)=\partial'\big(q'(\phi_{(-i)}(u))\big)=q'\big(\partial'(\phi_{(-i)}(u))\big)+m'(\phi_{(-i)}(u))\\
    &=q'\big(\phi_{(-i)}(\partial(u))\big)+\phi_{(-i)}(m(u))=\phi_{(-i)}\big(q(\partial(u))+m(u)\big)=\phi_{-i}\big(\partial(x)\big)
\end{align}
where we have successively used Equations \eqref{direct}, \eqref{commutator} (or \eqref{commutatorbis}) \eqref{eq:morphism3}, \eqref{eq:morphism4}, and Equation \eqref{commutator} (or \eqref{commutatorbis}) again. Notice that we passed from $m'(\phi_{(-i)}(u))$ to $\phi_{(-i)}(m(u))$ by using Equation \eqref{eq:morphism4} and the definition of $m$ (resp. $m'$). Proving Equation \eqref{finalequ2} then relies on using Equations \eqref{eq:morphismter} (on $T_{0}\wedge T_{-i-1}$) and \eqref{finalequ}, as well as the same kind of arguments as for Equations \eqref{bientot}-\eqref{bientot2}.  %Notice that in each case $m$ and $m'$ vanish because $u$ and $\phi_{(-i)}(u)$ do not possess any component of degree 0. 
This concludes the proof of the inductive step.
%\begin{align}
 %   \partial'\big(\phi_{-i-1}(x)\big)&=\partial'\big(\phi_{-i-1}(q(u))\big)=\partial'\big(q'(\phi_{(-i)}(u))\big)=q'\big(\partial'(\phi_{(-i)}(u))\big)+\underbrace{m'\big(\phi_{(-i)}(u)\big)}_{=\ 0}\\
  %  &=q'\big(\phi_{(-i)}(\partial(u))\big)=\phi_{(-i)}\big(q(\partial(u))\big)=\phi_{(-i)}\big(\partial(x)\big)-\underbrace{\phi_{(-i)}\big(m(u)\big)}_{=\ 0}
%\end{align}
%where we used Equations  \eqref{eq:morphism3}, \eqref{eq:morphism4} and \eqref{eq:morphism5}. Notice that in each case $m$ and $m'$ vanish because $u$ and $\phi_{(-i)}(u)$ do not possess any component of degree 0. This concludes the proof of the inductive step.
\bigskip

%\textbf{Inductive step:} Assume that the linear map $\phi_{(-i)}:T_{(-i)}\to T'_{(-i)}$ is constructed up to order $-i$, and that is satisfies:
% \begin{align}
%\partial'_{(-i+1)}\circ\phi_{(-i)}&=\phi_{(-i+1)}\circ\partial_{(-i+1)}\\
%\phi_{(-i)}\big(\llbracket x,y\rrbracket\big)&=\big\llbracket\phi_{(-i)}(x),\phi_{(-i)}(y)\big\rrbracket
%\end{align} 
 %for every $x,y\in T_{(-i-1)}$  such that $|x|+|y|\geq-i-1$. Then we construct $\phi_{-i-1}:T_{-i-1}\to T'_{-i-1}$ by following exactly the same lines of argument as in the second step, so that Equations \eqref{eq:morphism3} and \eqref{eq:morphism4} extend to level $-i-1$.
 
  The induction provides us with a linear application $\phi:\mathbb{T}\to \mathbb{T}'$ satisfying the following conditions:
 \begin{align}
% \phi|_{T_{-i}}&=\phi_{-i}\\
\partial'\circ\phi(a)&=\phi\circ\partial(a)\\
\phi\big(\llbracket a,b\rrbracket\big)&=\big\llbracket\phi(a),\phi(b)\big\rrbracket
\end{align} 
Moreover, this differential graded Lie algebra morphism between $\mathbb{T}$ and $\mathbb{T}'$ is 
 such that $\phi|_{T_{-1}}=\chi$,  $\phi|_{T_0}=\varphi$, and $\phi_{+1}(\Theta)=\Theta'$. The  morphism $\phi$ is the image through the function $G:\textbf{Lie-Leib}\to\textbf{DGLie}_{\leq1}$ of the compatible morphism of Lie-Leibniz triples $(\varphi,\chi)$:
 \begin{equation}
     G(\varphi,\chi)=\phi
 \end{equation}
 Then, when restricted to the wide subcategory \textbf{compLie-Leib}, the injective-on-objects function $G$ is a functor.
Moreover, since by construction $G(\varphi,\chi)=G(\varphi',\chi')$ if only if $\varphi=\varphi'$ and $\chi=\chi'$, $G$ is faithful.

 % and $\phi|_{T_{-2}}=\tau$.
%This morphism has been constructed from the data of the morphism of strict Lie-Leibniz triples $(\varphi,\chi)$.

Finally, let us explain why  \textbf{compLie-Leib} is the biggest wide subcategory of \textbf{Lie-Leib} (with respect to inclusion) such that $G$ is a (faithful) functor. The idea is that in the general case where the morphism  of Lie-Leibniz triples $(\varphi,\chi)$ is not compatible in the sense of Definition \ref{defcomp},
an obstruction arises. Indeed, let $(\varphi,\chi)$ be any morphism of Lie-Leibniz triples between $(\mathfrak{g},V,\Theta)$ and $(\mathfrak{g}',V',\Theta')$. Then, as seen in the first step of the induction,  the couple $(\varphi,\chi)$ would induce a morphism of differential graded Lie algebras between the associated tensor hierarchies $\mathbb{T}$ and $\mathbb{T}'$ if and only if there  exists a map $\tau:\bigslant{S^2(V)}{K}\to\bigslant{S^2(V')}{K'}$ making the following diagram commutative:
\begin{center}
\begin{tikzpicture}[on top/.style={preaction={draw=white,-,line width=#1}},
on top/.default=5pt]
\matrix(a)[matrix of math nodes, 
row sep=2.5em, column sep=3.5em, 
text height=1.5ex, text depth=0.25ex] 
{&S^2(V)&&S^2(V')\\
&&&\\
&\bigslant{S^2(V)}{K}&&\bigslant{S^2(V')}{K'}\\}; 
%\path[->](a-3-1) edge node[above right]{$\chi\odot\chi$} (a-3-4); %node[below right]{$\varphi$}  (a-1-2); 
%\path[->>](a-3-1) edge node[above left]{$p$} (a-1-2);
%\path[->](a-3-1) edge node[below left]{$\{\,.\,,.\,\}$} (a-4-2);
%\path[->](a-3-4) edge node[above right]{$\{\,.\,,.\,\}'$} (a-4-5);
%\path[->>](a-3-4) edge node[below right]{$p'$} (a-1-5);
\path[->](a-1-2) edge node[above]{$\chi\odot\chi$} (a-1-4);
\path[->](a-3-2) edge node[above]{$\tau$} (a-3-4);
\path[->](a-1-2) edge[on top=5pt] node[left]{$p$} (a-3-2);
\path[->](a-1-4) edge node[right]{$p'$} (a-3-4);
\end{tikzpicture}
\end{center}
However, the existence of such a map $\tau$ is dependent on the fact that $\chi\odot\chi(K)\subset K'$, namely that $(\varphi,\chi)$ is a compatible morphism.
\end{proof}

\begin{remark}
The equivalence of 2-categories between \emph{\textbf{Lie$\times$Mod}} and that of strict Lie 2-algebras \cite{baez:Lie2alg} does not seem to extend to \emph{\textbf{compLie-Leib}}, as it is not obvious how to define a homotopy between two morphisms of tensor hierarchies $(\phi_{-i})_i$ and $(\phi'_{-i})_i$ from the data of a homotopy between morphisms of Lie-Leibniz triples.\end{remark}

%In the general case where the Lie-Leibniz triples are not strict, i.e. when the embedding tensor needs not be $\mathfrak{g}$-invariant, one could not unequivocally define a linear map $\phi_{+1}:R_\Theta\to R_{\Theta'}$. This prevents to define a canonical assignment from morphisms of Lie-Leibniz triples to morphisms of differential graded Lie algebras.

We conclude this section and the paper by some discussion on the algebraic structures involved and open questions. First,  what is the
significance of only partial functoriality in Theorem \ref{centraltheorem}? While the construction of tensor hierarchies from Lie-Leibniz triples was inspired by theoretical physical considerations, the question of functoriality was motivated by the mathematical relevance of characterizing the injective-on-objects function~$G$. With respect to this latter motivation, the result established in Theorem \ref{centraltheorem} seems to be the best one can hope for. It hints at the possibility that the more relevant notion of morphism in the category of Lie-Leibniz triples is that of \emph{compatible morphisms} (see Definition \ref{defcomp}).

 Second, what is the meaning of the cohomology of the tensor hierarchy associated to a given Lie-Leibniz triple? In supergravity theories, the tensor hierarchy appears as a tool to cope with gauge theories in which the gauge algebra is not Lie \cite{deWit:2008ta, deWit-sam:endofhierarchy, Trigiante:2016mnt}. The cohomology of the tensor hierarchy then parametrizes the  $p$-form gauge fields which cannot be obtained from $p+1$-form gauge fields, thus measuring some sort of independence in the gauge parameters. In the algebraic construction described in the present paper, the meaning of the cohomology is not so clear as the relationship with the original physical inspiration has vanished.
 However, the following argument may shed some light on this question.
 For every Lie-Leibniz triple $(\mathfrak{g},V,\Theta)$ for which:
\begin{equation}\label{condstring}
K=\mathrm{Ker}\big(\{\,.\,,.\,\}\big)
\end{equation} the map $\partial_{-1}:T_{-2}\to T_{-1}$ is injective because $T_{-2}=\bigslant{S^2(V)}{\mathrm{Ker}\big(\{\,.\,,.\,\}\big)}[2]$ is isomorphic to the ideal of squares $\mathcal{I}$, the image of $\partial_{-1}$. %Since $\partial_{-1}\big(\partial_{-2}(T_{-3})\big)=0$, this implies in turn  that $\partial_{-2}=0$. 
Moreover, one can then show that Condition \eqref{condstring} implies that $\partial_{-2}=0$, as would be expected from $T_{\leq-1}$ being a cochain complex. %, in agreement with the fact that $\mathrm{Ker}(\partial_{-1})=\mathrm{Im}(\partial_{-2})$, so that is always zero when the Lie-Leibniz triple under consideration satisfies condition \eqref{condstring}. 
The question arises: what about the higher cohomology spaces? 
We have some reason to think that the following conjecture holds:
\begin{conjecture}\label{conjecturr}
For a Lie-Leibniz triple  $(\mathfrak{g},V,\Theta)$ satisfying condition \eqref{condstring} (so in particular for semi-strict Lie-Leibniz triples), the associated tensor hierarchy induces an exact sequence of the ideal of squares $\mathcal{I}$:
\begin{center}
\begin{tikzcd}[column sep=1cm,row sep=0.4cm]
\ldots\ar[r,"\partial_{-4}"]&T_{-4}\ar[r,"\partial_{-3}"]&T_{-3}\ar[r,"\partial_{-2}"]&T_{-2}\ar[r,"\partial_{-1}"]&\mathcal{I}[1]\ar[r]&0
\end{tikzcd}
\end{center}
\end{conjecture}
\noindent This would be plausible since in this case, the role of $\mathfrak{g}$ in the construction of the tensor hierarchy is not as salient as compared to the general case (e.g. in defining $T_{-2}$). %This raises the question of the significance of the the homology of the tensor hierarchy associated 

%{example5}
 %It would seem that the injective-on-objects function $G:\textbf{Lie-Leib}\to\textbf{DGLie}_{\leq1}$ being a functor has something to do with the existence of resolutions of the ideal of squares. %However, we know from examples in supergravity that this is not necessarily the case for general Lie-Leibniz triples \cite{???}. 
%This resolution of this conjecture is still under investigation.

Third, the motivation behind this axiomatic construction of tensor hierarchies goes well beyond supergravity theories (the following argument is taken from \cite{Bonezzi:2019ygf}). In classical gauge theories, gauge transformations are interpreted as smooth transformations of the fields $\phi^i$ leaving the action invariant:
\begin{equation}
0=\delta_\epsilon(S[\phi^i])=\int\frac{\delta S}{\delta \phi^i}\delta_\epsilon\phi^i
\end{equation}
where $\phi^i$ denote all the fields and $\epsilon$ denotes an arbitrary vector in the space of gauge parameters $V$. Invariance of the action under the commutator of two gauge transformations $\big[\delta_\epsilon, \delta_{\epsilon'}\big]$ implies that the space of gauge transformations forms a Lie algebra $\mathfrak{g}$.
%\begin{equation}
%0=[\delta_\epsilon,\delta_{\epsilon'}](S[\phi^i])=\int2\frac{\delta^2 S}{\delta\phi^i\delta\phi^j}\delta_{[\epsilon}\phi^i\delta_{\epsilon']}\phi^j+\frac{\delta S}{\delta \phi^k}[\delta_\epsilon,\delta_{\epsilon'}]\phi^k\quad\text{implies}\quad [\delta_\epsilon,\delta_{\epsilon'}]\phi^k=0
%\end{equation}
 From this, physicists deduce that the space of gauge parameters $V$ is itself a Lie algebra and $\delta:V\to\mathfrak{g}$ is a Lie algebra morphism. But as was noticed in \cite{Bonezzi:2019ygf, Hohm:2019wql}, it could well be that $V$ is a Leibniz algebra and $\delta$ is an embedding tensor. As is currently established in higher gauge theories, if the gauge algebra is not Lie, we expect some \emph{higher form of Lie algebras}  to appear, as these are the ones
integrating to (higher) symmetry groups. In the case at hand, the corresponding higher gauge algebra would be the tensor hierarchy, or an alternative, equivalent higher structure \cite{2021arXiv210600108B}. The present article thus provides a rigorous construction of the needed material.

Finally, the question arises as to how the statement of Theorem \ref{centraltheorem} generalizes to pairs of algebras over other operads.
 As explained earlier the motivation behind the construction of the tensor hierarchy associated to a Lie-Leibniz triple comes from theoretical physical considerations. We did not (yet) consider the possibility of extending this construction to other kinds of algebras, with suitable compatibility conditions. This question may be answered once we have a more categorical approach to the construction of tensor hierarchies, departing from the technical, computational proofs presented in Section \ref{section3}. The present paper only provides an explicit, constructive proof of Theorem \ref{centraltheorem}, and hints at the existence of a proof using more operadic and categorical tools.  Actually, this kind of  proof may be of great help in order to generalize Theorem \ref{centraltheorem} to the category of Leibniz algebroids.

 % affichŽ sur la page qui annonce les annexes

% dŽbut des annexes
\appendix
% ajout ˆ la ToC :
%\addappheadtotoc
% page qui annonce les annexes
%\appendixpage

\section{Corrigendum (in collaboration with Jakob Palmkvist)}\label{corrige}%\footnote{\emph{Orebro University, Department of Mathematics, SE-701 82, Orebro, Sweden}})}

In \cite{lavauLieAlgebraCrossed2023}, Theorem \ref{centraltheorem}, $\textbf{DGLie}_{\leq1}$ should be replaced by $\textbf{DGLie}$. The reason is that when the embedding tensor is not $\mathfrak{g}$-equivariant, consistency of the differential graded Lie algebra structure at level $+1$ may sometimes require the introduction of spaces of higher degrees (see Proposition \ref{experimental2} and Sequence \eqref{eq:sequence}). This observation had been missed in \cite{lavauLieAlgebraCrossed2023} and  is based on  an example inspired by supergravity theories contradicting Proposition~\ref{experimental} in \cite{lavauLieAlgebraCrossed2023} which has been provided by one of the authors of the present Corrigendum (see Example~\ref{example1}). Likewise, the part of the proof of Proposition 23 in \cite{lavauLieAlgebraCrossed2023} involving the space of degree~$+1$ is modified accordingly.  These modifications imply that $\textbf{DGLie}_{\leq1}$ should be replaced by $\textbf{DGLie}$ in Propositions~\ref{prop26} and~\ref{prop7}, and thus in Theorem \ref{centraltheorem} too. Table~\ref{fig1} summarizes how this updated result applies to the most distinguished cases.

\begin{comment}
\begin{theorem}\label{centraltheorem}
The functor
\begin{align*}
\overline{G}:\hspace{0.2cm}\emph{\textbf{Lie$\times$Mod}}&\xrightarrow{\hspace*{1.2cm}} \hspace{0.4cm}\emph{\textbf{DGLie}}\\
	\big(\mathfrak{c}\xrightarrow{\Theta}\mathfrak{g}\big)&\xmapsto{\hspace*{1.2cm}} \big(\mathfrak{c}\xrightarrow{\Theta}\mathfrak{g}\xrightarrow{0}\mathbb{R}[-1]\big)
\end{align*}
can be canonically extended to an injective-on-objects function $G:\emph{\textbf{Lie-Leib}}\to\emph{\textbf{DGLie}}$ such that the restriction of this function to the wide sub-category \textbf{\emph{compLie-Leib}} is a faithful functor. Moreover, \textbf{\emph{compLie-Leib}} is the biggest wide subcategory of \textbf{\emph{Lie-Leib}} (with respect to inclusion) such that this functorial property holds. % which coincides with $\overline{G}$ on \emph{\textbf{Lie$\times$Mod}}.
\end{theorem}
\end{comment}

\medskip
\noindent \textbf{I.}  First, Proposition \ref{experimental} in \cite{lavauLieAlgebraCrossed2023} should be rewritten:
\begin{customthm}{13'}\label{experimental2}
To every Lie-Leibniz triple $(\mathfrak{g},V,\Theta)$ in which $V$ is a Lie algebra, there corresponds a differential graded Lie algebra:
\begin{center}
\begin{tikzcd}[column sep=1cm,row sep=0.4cm]
V[1]\ar[r,"\Theta"]&\mathfrak{g}\ar[r,"-\eta(-;\Theta)"]& R_\Theta[-1]\ar[r,"\text{$\llbracket\Theta,.\,\rrbracket$}"]&T_2\ar[r,"\text{$\llbracket\Theta,.\,\rrbracket$}"]&\ldots
\end{tikzcd}
\end{center}
which canonically extends the graded Lie algebra structure on $V[1]\oplus \mathfrak{g}$ defined by Equations \eqref{bronze}-\eqref{silver}. %\eqref{bronze}-\eqref{silver}. %Here, $ R_\Theta\subset \mathrm{Hom} (V,\mathfrak{g})$ is the cyclic $\mathfrak{g}$-submodule generated by the embedding tensor $\Theta$ (see Equation \eqref{defrep}), and $\phi:\mathfrak{g}\to\mathrm{End}( R_\Theta)$ denotes the associated action of $\mathfrak{g}$.
\end{customthm}

\begin{remark}
The error in the proof of Proposition \ref{experimental} in \cite{lavauLieAlgebraCrossed2023} was that Equation \eqref{cannotbe} cannot be enforced, as both sides of the equation are determined independently from existing data, and so they may not coincide. Likewise, one cannot assume Equation \eqref{enforcedjacobi2}. Moreover, Equation \eqref{imporfg2} and Equation \eqref{thetasquared} in \cite{lavauLieAlgebraCrossed2023} should be replaced by $\llbracket\Theta,\Theta\rrbracket=0$ (which is equivalent to the cohomological equation $\partial\circ\partial=0$).
\end{remark}

\begin{proof}
The proof of Proposition \ref{experimental} in \cite{lavauLieAlgebraCrossed2023} is valid, until Equation \eqref{cannotbe}. The construction of the spaces of degrees higher than one is similar to that appearing in \cite{palmkvistTensorHierarchyAlgebra2014} (where the grading convention is reversed).
 We generate the spaces $T_i$, $i\geq2$ by induction, from the commutators of the adjoint action of the elements at level $1\leq k\leq i-1$. Given an element $\chi\in T_{1}=R_{\Theta}[-1]$, the adjoint action of $\chi$ is defined from  Equations~\eqref{imporfg} and \eqref{imporfgbis} in \cite{lavauLieAlgebraCrossed2023} by
$\mathrm{ad}_\chi(x)=\llbracket \chi,x\rrbracket$.

Then, one defines $T_2$ to be the vector space generated by elements of the form $\left(\mathrm{ad}_\chi\circ \mathrm{ad}_\psi+\mathrm{ad}_\psi\circ \mathrm{ad}_\chi\right)\big|_V$, which correspond to the graded commutator $\big[\mathrm{ad}_\chi,\mathrm{ad}_\psi\big]\big|_V$.
If $T_2$ is zero (see Example \ref{example1bis}), the process stops and the differential graded Lie algebra obtained through Proposition \ref{experimental2} is $V[1]\to \mathfrak{g}\to T_1$. Otherwise, we represent the element $\big[\mathrm{ad}_\chi,\mathrm{ad}_\psi\big]\big|_V$ of $T_2$ by the bracket $\llbracket \chi,\psi\rrbracket$ and we assign to it a degree $+2$. The evaluation of $\llbracket \chi,\psi\rrbracket$ on an element $x\in V$ is symbolized by the nested bracket $\llbracket\llbracket \chi,\psi\rrbracket,x\rrbracket$ or, alternatively, by the notation $\mathrm{ad}_{\llbracket \chi,\psi\rrbracket}(x)$. By construction, we then have
% by $\llbracket\llbracket \chi,\psi\rrbracket,.\,\rrbracket$ %When evaluating the latter element on $x\in V$, we write %$\mathrm{ad}_{\llbracket \chi,\psi\rrbracket}(x)=
%$\llbracket\llbracket \chi,\psi\rrbracket,x\rrbracket$, allowing to write 
the following Jacobi identity, for every $x\in V[1]$:
\begin{equation}\label{eq:action2}
\llbracket\llbracket \chi,\psi\rrbracket,x\rrbracket=\llbracket\chi,\llbracket \psi,x\rrbracket\rrbracket+\llbracket\psi,\llbracket \chi,x\rrbracket\rrbracket
\end{equation}

Moreover, the bracket $\llbracket \chi,\psi\rrbracket$ is graded skew-symmetric, i.e. $\llbracket \chi,\psi\rrbracket=\llbracket \psi,\chi\rrbracket$ and represents the action of $T_1$ on $T_1$, taking values in $T_2$. %In particular, if $\chi=\Theta$, then $$
%The action of $T_1$ on $T_1$ is precisely given by:
%\begin{equation}
%\llbracket \chi,\psi\rrbracket=
%\end{equation}
The action of $\mathfrak{g}$ on $T_2$ is then defined by:
\begin{equation}\label{eq:actiong}
a\cdot \llbracket \chi,\psi\rrbracket\overset{\mathrm{def}}{=}\llbracket a\cdot\chi,\psi\rrbracket+\llbracket \chi,a\cdot\psi\rrbracket
\end{equation}
Therefore, if $\chi=\psi=\Theta$, we have:
\begin{equation}\label{eq:zeroness}
0=a\cdot\llbracket\Theta,\Theta\rrbracket=2\llbracket \Theta, a\cdot \Theta\rrbracket
\end{equation}
which is consistent with Equation \eqref{imporfg} in \cite{lavauLieAlgebraCrossed2023}. As $a\cdot\Theta=\eta(a,\Theta)$, where $\eta:\mathfrak{g}\to \mathrm{Hom}(V,\mathfrak{g})$ is the representation of $\mathfrak{g}$ on $\mathrm{Hom}(V,\mathfrak{g})$, Equation \eqref{eq:zeroness} means that $\llbracket\Theta,\eta(a,\Theta)\rrbracket=0$, implying in turn that the linear operator $\partial_{2}=\llbracket\Theta,.\,\rrbracket:T_1\to T_2$ extends the chain complex $V[1]\longrightarrow \mathfrak{g}\longrightarrow  R_{\Theta}[-1]$ to the right.

One defines $T_3$ to be the vector space generated by elements of the form:
\begin{equation}\label{eq:rule}
\big[\mathrm{ad}_\phi,\mathrm{ad}_{\llbracket \chi,\psi\rrbracket}\big]\big|_V\overset{\mathrm{def}}{=}\left(\mathrm{ad}_\phi\circ \mathrm{ad}_{\llbracket \chi,\psi\rrbracket}-\mathrm{ad}_{\llbracket \chi,\psi\rrbracket}\circ \mathrm{ad}_\phi\right)\big|_V\end{equation} Such elements %correspond to the graded commutator $\big[\mathrm{ad}_\phi,\mathrm{ad}_{\llbracket \chi,\psi\rrbracket}\big]\big|_V$, 
are represented in $T_3$ by the nested bracket $\llbracket \phi,\llbracket \chi,\psi\rrbracket\rrbracket$, inheriting the graded skew-symmetry property of the commutator from Equation \eqref{eq:rule}:
\begin{equation}
 \llbracket \phi,\llbracket \chi,\psi\rrbracket\rrbracket=-\llbracket \llbracket \chi,\psi\rrbracket, \phi\rrbracket
\end{equation}
This property is consistent with the grading of $\phi$ and $\llbracket \chi,\psi\rrbracket$.
To elements of $T_3$ are attributed a degree $+3$.%Contrary to earlier, even if $T_3$ is zero, the induction process does not necessarily stop because one may generate elements at level 4 from two elements at level 2.  

The definition of the nested bracket $\llbracket \phi,\llbracket \chi,\psi\rrbracket\rrbracket$  corresponds to the following Jacobi identity:
\begin{equation}\label{eq:action3}
\llbracket\llbracket \phi,\llbracket \chi,\psi\rrbracket\rrbracket,x\rrbracket=\llbracket\phi,\llbracket\llbracket \chi,\psi\rrbracket,x\rrbracket\rrbracket-\llbracket \llbracket \chi,\psi\rrbracket, \llbracket \phi,x\rrbracket\rrbracket
\end{equation}
characterizing the action of $T_3$ on every element $x\in V[1]$. Using Equation \eqref{eq:action2}, the first term on the right-hand side reads:
\begin{equation}\label{eq:mana1}
\llbracket\phi,\llbracket\llbracket \chi,\psi\rrbracket,x\rrbracket\rrbracket=\llbracket\phi,\llbracket\chi,\llbracket\psi,x\rrbracket\rrbracket\rrbracket+\llbracket\phi,\llbracket\psi,\llbracket\chi,x\rrbracket\rrbracket\rrbracket
\end{equation}
while using Equation \eqref{eq:actiong} shows that the second term of the right-hand side of Equation \eqref{eq:action3} reads:
\begin{equation}\label{eq:mana2}
-\llbracket \llbracket \chi,\psi\rrbracket, \llbracket \phi,x\rrbracket\rrbracket=-\llbracket\chi,\llbracket\psi,\llbracket\phi,x\rrbracket\rrbracket\rrbracket-\llbracket\psi,\llbracket\chi,\llbracket\phi,x\rrbracket\rrbracket\rrbracket
\end{equation}
Summing both right-hand sides of Equations \eqref{eq:mana1} and \eqref{eq:mana2}, gathering terms by using Equation \eqref{eq:action2} and eventually using Equation \eqref{eq:action3}, one deduces that, as elements of $T_3$:
\begin{equation}\label{eq:jacobi}
\llbracket \phi,\llbracket \chi,\psi\rrbracket\rrbracket=\llbracket \llbracket \phi,\chi\rrbracket,\psi\rrbracket-\llbracket \chi,\llbracket \phi,\psi\rrbracket\rrbracket
\end{equation}

This is the Jacobi identity for three elements of degree $+1$, and can be understood as how $T_2$ acts on $T_1$ (or conversely). In particular, for $\phi=\chi=\Theta$, one has:
\begin{equation}
\llbracket \Theta,\llbracket \Theta,\psi\rrbracket\rrbracket=\llbracket \underbrace{\llbracket \Theta,\Theta\rrbracket}_{=\,0},\psi\rrbracket-\llbracket \Theta,\llbracket \Theta,\psi\rrbracket\rrbracket
\end{equation}
From this, one deduces that $\partial_3=\llbracket \Theta,.\,\rrbracket:T_2\to T_3$ extends the chain complex:
\begin{center}
\begin{tikzcd}[column sep=1cm,row sep=0.4cm]
V[1]\ar[r,"\Theta"]&\mathfrak{g}\ar[r,"-\eta(-;\Theta)"]& R_\Theta[-1]\ar[r,"\text{$\llbracket\Theta,.\,\rrbracket$}"]&T_2
\end{tikzcd}
\end{center}
to the right.
 The action of $\mathfrak{g}$ on $T_3$ is defined in a similar fashion as in Equation \eqref{eq:actiong} and meets no obstruction.
 
  If $T_3$ is zero-dimensional, then the process stops because $T_4$  could only be generated by elements of the form $\big[\mathrm{ad}_\xi,\mathrm{ad}_{\llbracket \zeta,\llbracket \chi,\psi\rrbracket\rrbracket}\big]\big|_V$ or $\big[\mathrm{ad}_{\llbracket \xi,\zeta\rrbracket},\mathrm{ad}_{\llbracket \chi,\psi\rrbracket}\big]\big|_V$ (where the commutator is graded), but the latter kind of elements can be rewritten using terms of the former kind. Indeed, by using Equations \eqref{eq:action3} and \eqref{eq:jacobi}, together with the knowledge of the action of $\mathfrak{g}$ on $T_3$, one finds:
  \begin{equation}
  \big[\mathrm{ad}_{\llbracket \xi,\zeta\rrbracket},\mathrm{ad}_{\llbracket \chi,\psi\rrbracket}\big](x)=\big[\mathrm{ad}_\xi,\mathrm{ad}_{\llbracket \zeta,\llbracket \chi,\psi\rrbracket\rrbracket}\big](x)+\big[\mathrm{ad}_\zeta,\mathrm{ad}_{\llbracket \xi,\llbracket \chi,\psi\rrbracket\rrbracket}\big](x)
  \end{equation} 
 for every $x\in V$. If $T_3$ is not zero-dimensional, the induction continues, with the space at level $k\geq4$ being defined as: %the vector space generated by elements of the form $\big[\mathrm{ad}_{}\big]$
\begin{equation*}
T_{k}= \mathrm{Span}\left(\big[\mathrm{ad}_{T_1},\mathrm{ad}_{T_{k-1}}\big]\big|_{V}\right) %\sum_{j=1}^{E(\frac{k+1}{2})}\mathrm{Span}\left(\big[\mathrm{ad}_{T_j},\mathrm{ad}_{T_{k-j}}\big]\big|_{V}\right)
\end{equation*}
where the commutator is again graded. If  $T_i=0$ for some $i$ then $T_{j}=0$ for every $j\geq i$. % and $E$ is the floor function. 
The graded Lie bracket between two elements $\Xi$ and $\Omega$ whose degrees sum up to $k$ comes from the very definition of $T_k$, while the graded Jacobi identity between elements of various degrees can be obtained from a similar technique performed through Equations \eqref{eq:actiong} and \eqref{eq:action3}--\eqref{eq:jacobi}. At each step, the chain complex
\begin{center}
\begin{tikzcd}[column sep=1cm,row sep=0.4cm]
V[1]\ar[r,"\Theta"]&\mathfrak{g}\ar[r,"-\eta(-;\Theta)"]& R_\Theta[-1]\ar[r,"\text{$\llbracket\Theta,.\,\rrbracket$}"]&T_2\ar[r,"\text{$\llbracket\Theta,.\,\rrbracket$}"]&T_3\ar[r,"\text{$\llbracket\Theta,.\,\rrbracket$}"]&\ldots\ar[r,"\text{$\llbracket\Theta,.\,\rrbracket$}"]&T_{k-1}
\end{tikzcd}
\end{center}
 can be extended to the right by the linear map $\partial_k=\llbracket\Theta,.\,\rrbracket: T_{k-1}\to T_{k}$. Eventually, the compatibility between the graded Lie bracket and the differential is guaranteed by construction.
\end{proof}

\medskip
\noindent \textbf{II.} We modify the proof of Proposition \ref{prop2} in \cite{lavauLieAlgebraCrossed2023} following the same lines as in heading \textbf{I.}. The only remaining data to be introduced for the differential graded Lie algebra structure to be consistent is the action of the positively graded spaces $T_k$, for $k\geq2$, on the negatively graded spaces $T_{-i}$, for $i\geq2$. This is done by noticing that 1. the action of $T_k$ on $V$ is properly defined by the proof of Proposition \ref{experimental2}, and 2. that the graded Lie bracket on $T_{\leq-1}$ is surjective on $T_{\leq-2}$ (this is item 2 of Proposition \ref{prophierarchy} in \cite{lavauLieAlgebraCrossed2023}). Then, we proceed by induction and we set, to begin with, for any $x,y\in V[1]$ and $\Xi\in T_k$:
\begin{equation}\label{eqdefomega0}
\llbracket \Xi,\llbracket x,y\rrbracket\rrbracket\overset{\mathrm{def}}{=}\llbracket \llbracket \Xi,x\rrbracket,y\rrbracket+(-1)^{k}\llbracket x,\llbracket \Xi,y\rrbracket\rrbracket
\end{equation}

Assume that the action of $T_k$ has been defined on $T_{-j}$ for $1\leq j\leq i-1$, then for any $u\in T_{|u|}$ and $v\in T_{|v|}$ -- where $|u|$ symbolizes the (necessarily negative) degree of $u$ -- such that  % $\llbracket u,v\rrbracket\in T_{-i}$ such that 
$-i-1\leq |u|,|v|\leq -1$ and $|u|+|v|=-i$, % -- where $|u|$ symbolizes the (necessarily negative) degree of $u$ -- 
we set:
\begin{equation}\label{eqdefomega}
\llbracket \Xi,\llbracket u,v\rrbracket\rrbracket\overset{\mathrm{def}}{=}\llbracket \llbracket \Xi,u\rrbracket,v\rrbracket+(-1)^{k|u|}\llbracket u,\llbracket \Xi,v\rrbracket\rrbracket
\end{equation}
In order for the definition of the left-hand side to be consistent -- e.g. it should not depend on the choice of representent of $\llbracket u,v\rrbracket$ in the free graded Lie algebra of $V[1]$ -- one should check by induction that Equation~\eqref{eqdefomega} respects the graded Jacobi identity, in the following sense:
\begin{equation}\label{eqdefomega2}
\llbracket \llbracket \Xi,u\rrbracket,\llbracket v,w\rrbracket\rrbracket+(-1)^{k|u|}\llbracket u,\llbracket \llbracket\Xi,v\rrbracket, w\rrbracket\rrbracket+(-1)^{k(|u|+|v|)}\llbracket u,\llbracket v, \llbracket\Xi,w\rrbracket\rrbracket\rrbracket+\circlearrowright\,=0
\end{equation}
where the arrow symbolizes graded circular permutation. This statement has been justified and proved (for the non-graded case) in  Section 2.2 in~\cite{palmkvistNonlinearRealizationsLie2022}.

Definitions \eqref{eqdefomega0} and \eqref{eqdefomega}, being consistent with the existing graded Lie bracket and differential, turn the following chain complex:
\begin{equation}\label{eq:sequence}
\begin{tikzcd}[column sep=1cm,row sep=0.4cm]
\ldots\ar[r,"\partial_{-3}"]&T_{-3}\ar[r,"\partial_{-2}"]&T_{-2}\ar[r,"\partial_{-1}"]&V[1]\ar[r,"\Theta"]&\mathfrak{g}\ar[r,"-\eta(-;\Theta)"]& R_\Theta[-1]\ar[r,"\text{$\llbracket\Theta,.\,\rrbracket$}"]&T_2\ar[r,"\text{$\llbracket\Theta,.\,\rrbracket$}"]&T_{3}\ar[r,"\text{$\llbracket\Theta,.\,\rrbracket$}"]&\ldots
\end{tikzcd}
\end{equation}
into a differential graded Lie algebra, called the \emph{tensor hierarchy} associated to the Lie-Leibniz triple $(\mathfrak{g},V,\Theta)$ in \cite{lavauLieAlgebraCrossed2023}.
We present  in Table \ref{fig1} the tensor hierarchies associated to several distinguished Lie-Leibniz triples.

\begin{table}[!h]
\centering
\begin{tabular}{l|c|}
\hline
  \multicolumn{1}{|c|}{\textbf{Lie-Leibniz triple}} & \textbf{corresponding dgLa}\\
  \hline
   \multicolumn{1}{|c|}{\begin{tabular}{@{}c@{}}\textbf{Lie alg. crossed module} \\ $V\overset{\Theta}{\longrightarrow} \mathfrak{g}$\\
  {\small $V$ is a Lie algebra}\\
  {\small  $\Theta$ is $\mathfrak{g}$-equivariant}
   \end{tabular}}& $V[1]\overset{\Theta}{\longrightarrow} \mathfrak{g}\overset{0}{\longrightarrow} \mathbb{R}[-1]\hspace{0.475cm} $\\
    \hline
    \multicolumn{1}{|c|}{\begin{tabular}{@{}c@{}}\textbf{augmented Leibniz alg. \cite{bordemannGlobalIntegrationLeibniz2017}}\\
    $V\overset{\Theta}{\longrightarrow} \mathfrak{g}$\\
  {\small  $\Theta$ is $\mathfrak{g}$-equivariant}\end{tabular}}& $\ldots\overset{\partial_{-2}}{\longrightarrow}T_{-2}\overset{\partial_{-1}}{\longrightarrow}V[1]\overset{\Theta}{\longrightarrow} \mathfrak{g}\overset{0}{\longrightarrow} \mathbb{R}[-1]\hspace{3.43cm}$ \\
   \hline
  \multicolumn{1}{|c|}{\begin{tabular}{@{}c@{}}\textbf{Proposition \ref{experimental2}}\\
  $V\overset{\Theta}{\longrightarrow} \mathfrak{g}$ \\ {\small  $V$ is a Lie algebra}\end{tabular}}&$\hspace{2.955cm} V[1]\overset{\Theta}{\longrightarrow} \mathfrak{g}\overset{-\eta(\,.\,,\Theta)}{\longrightarrow} R_\Theta[-1]\overset{\llbracket \Theta,.\,\rrbracket}{\longrightarrow}T_{2}\overset{\llbracket \Theta,.\,\rrbracket}{\longrightarrow}\ldots$ \\
  \hline
    \multicolumn{1}{|c|}{\begin{tabular}{@{}c@{}} \textbf{general Lie-Leibniz triple} \\ $V\overset{\Theta}{\longrightarrow} \mathfrak{g}$ %\\ {\small  $V$ is a Leibniz algebra}\\ {\small  $\Theta$ is not $\mathfrak{g}$-equivariant} 
  \end{tabular} }& $\ldots\overset{\partial_{-2}}{\longrightarrow}T_{-2}\overset{\partial_{-1}}{\longrightarrow}V[1]\overset{\Theta}{\longrightarrow} \mathfrak{g}\overset{-\eta(\,.\,,\Theta)}{\longrightarrow}R_\Theta[-1]\overset{\llbracket \Theta,.\,\rrbracket}{\longrightarrow}T_{2}\overset{\llbracket \Theta,.\,\rrbracket}{\longrightarrow}\ldots$ \\
    \hline
\end{tabular}
\caption{Description of the injective-on-object function \textbf{LieLeib}$\longrightarrow$\textbf{DGLie} in various cases.}
\label{fig1}
\end{table}

\begin{remark}
It is conjectured  that the negatively graded part of these tensor hierarchies are isomorphic to the positively graded part of %the algebra denoted $W(\mathfrak{g})$ in~\cite{cederwallTensorHierarchyAlgebras2020} or, equivalently, 
the Borcherds algebras considered in \cite{cederwallSuperalgebrasConstraintsPartition2015, palmkvistTensorHierarchyAlgebra2014}. This isomorphism will be explored elsewhere.
\end{remark}

\medskip
\noindent \textbf{III.} As a consequence of heading \textbf{II.}, $\textbf{DGLie}_{\leq1}$ should be replaced by $\textbf{DGLie}$ in Propositions \ref{prop26} and \ref{prop7} in  \cite{lavauLieAlgebraCrossed2023}. The proof of the former is not impacted by this change however, as it relies on the non-positively graded part of the algebra only. On the other hand, the functorial property described in Proposition \ref{prop7} is still valid under the modifications brought in by heading \textbf{II.}. This is because the linear map $\phi_{1}:T_{1}\to  T_{1}'$ straightforwardly extends by induction to a graded Lie algebra morphism from $T_{1\leq i}$ to $T_{1\leq i}'$, as we now show.

To start with, one sets:
\begin{equation}
\phi_{2}(\llbracket \chi,\psi \rrbracket)\overset{\mathrm{def}}{=}\llbracket\phi_{1}(\chi),\phi_{1}(\psi)\rrbracket'
\end{equation}
for every $\chi,\psi\in T_1=R_\Theta[-1]$, and more generally, for every $\chi_1,\ldots, \chi_i\in T_1$:
\begin{equation}\label{eqinducmor2}
\phi_{i}\big(\llbracket \chi_1,\llbracket \chi_2,\llbracket\ldots\llbracket \chi_{i-1},\chi_{i}\rrbracket\ldots\rrbracket\big)\overset{\mathrm{def}}{=}\llbracket \phi_1(\chi_1),\llbracket \phi_1(\chi_2),\llbracket\ldots\llbracket \phi_1(\chi_{i-1}),\phi_1(\chi_{i})\rrbracket'\ldots\rrbracket'
\end{equation}
%Let  $i\geq2$ and assume that the morphism $\phi_j:T_j\to T_j'$ has been defined for every $1\leq j\leq i$. 
Since the graded Lie bracket of two elements $\Xi\in T_{j_1}$, $\Omega\in T_{j_2}$, $j_1,j_2\geq1$, can be written as a sum of nested brackets of elements of $T_1$, %$T_{i+1}$ is generated by $\big[\mathrm{ad}_{T_{1}},\mathrm{ad}_{T_{i}}\big]\big|_{V}$, we
%brackets of the form $\big[\mathrm{ad}_{\chi},\mathrm{ad}_{\Xi}\big]\big|_{V}$, where $\chi\in T_{1}$ and $\Xi\in T_{i}$,  one then sets at level $i+1$:
one deduces from Equation \eqref{eqinducmor2} that:
\begin{equation}\label{eqinducmor}
\phi_{j_1+j_2}(\llbracket \Xi,\Omega \rrbracket)=\llbracket\phi_{j_1}(\Xi),\phi_{j_2}(\Omega)\rrbracket'
\end{equation}
%Using Equation \eqref{eqinducmor} inductively, one has:
%\begin{equation}\label{eqinducmor2}
%\phi_{i+1}\big(\llbracket \chi_1,\llbracket \chi_2,\llbracket\ldots\llbracket \chi_{i},\chi_{i+1}\rrbracket\ldots\rrbracket\big)=\llbracket \phi_1(\chi_1),\llbracket \phi_1(\chi_2),\llbracket\ldots\llbracket \phi_1(\chi_{i}),\phi_1(\chi_{i+1})\rrbracket\ldots\rrbracket
%\end{equation}
%for every $\chi_k\in T_1$. 
%This describes uniquely the linear map $\phi_{i+1}$. 
This implies that the collection of linear maps $(\phi_i)_{i\geq1}$ forms a graded Lie algebra morphism from $T_{\geq1}$ to $T_{\geq1}'$.

These linear maps  are moreover compatible with their non-positively graded counterpart, in the sense that Equations \eqref{eqraja1}, \eqref{eqraja2}, \eqref{eq:pasdid3}, \eqref{eq:pasdid3bis} and \eqref{finalequ}, \eqref{finalequ2} in  \cite{lavauLieAlgebraCrossed2023} are still valid for elements of degree higher than $1$. %$\phi_{i}$ ($i\geq1$). 
More precisely, these equations, together with Equations  \eqref{eqdefomega} and \eqref{eqinducmor}, imply by induction that for any $\Xi\in T_i$ ($i\geq2$) and $u\in T_{-j}$ ($j\geq0$) one has:
\begin{equation}
\phi_{i-j}(\llbracket \Xi,u\rrbracket)=\llbracket \phi_i(\Xi),\phi_{-j}(u)\rrbracket
\end{equation}
Let $\phi:\mathbb{T}\to \mathbb{T}'$ be the unique degree 0 linear map restricting to $\phi_i$ on $T_i$, $i\in\mathbb{Z}$. From the proof of Proposition~\ref{prop7} in \cite{lavauLieAlgebraCrossed2023} together with the above discussion, we obtain that $\phi$ is a differential graded Lie algebra morphism.
As a consequence, $\textbf{DGLie}_{\leq1}$ should be replaced by $\textbf{DGLie}$ in Proposition \ref{prop7}. As Propositions \ref{prop26} and \ref{prop7} are constitutive of Theorem \ref{centraltheorem}, we conclude that $\textbf{DGLie}_{\leq1}$ should be replaced by $\textbf{DGLie}$ in Theorem~1 in~\cite{lavauLieAlgebraCrossed2023}.

 %Since it is $\mathfrak{g}$-equivariant, we conclude that It is ob extends naturally to all positive degrees so that the functorial property of Theorem \ref{centraltheorem} is not impacted by the discovery.

%\noindent\textbf{IV.} Most examples of tensor hierarchies arise in gauged supergravity, we refer to \cite{cederwallTensorHierarchyAlgebras2020, palmkvistTensorHierarchyAlgebra2014}.

\medskip
\noindent \textbf{IV.}  Example \ref{example1} was proposed by Jakob Palmkvist as a counter-example to Proposition \ref{experimental} in~\cite{lavauLieAlgebraCrossed2023}. It has been inspired by the treatment of tensor hierarchies in gauged supergravity, from which the notion originates, and in extended geometry. We refer to \cite{cederwallTensorHierarchyAlgebras2020, palmkvistTensorHierarchyAlgebra2014} and references within for details and more examples on this perspective. The graded Lie algebras appearing in Examples \ref{example1} and \ref{example1bis} are the Lie superalgebras of Cartan type denoted by $W(3)$ and $W(2)$ in \cite{kacLieSuperalgebras1977}. The relation between these Lie superalgebras and the more general tensor hierarchy algebras (originally introduced in the gauged supergravity context) was explained in \cite{carboneGeneratorsRelationsLie2019}.

\begin{example}\label{example1}
The data for the Lie-Leibniz triple is $\mathfrak{g}=\mathfrak{gl}_3(\mathbb{R})$, $V=\mathbb{R}^3$ is the defining representation, and the embedding tensor is defined on the basis vectors $e_1, e_2, e_3$ as:
\begin{equation}
\Theta(e_1)=-E_{22}, \quad \Theta(e_2)=E_{21}\quad \text{and} \quad \Theta(e_3)=0
\end{equation}
where $E_{ij}$ is the $3\times 3$ matrix with 1 at the $i$-th line and the $j$-th column. The Leibniz product is null, except in the following two cases:
\begin{equation}
e_1\circ e_2=-e_2\quad \text{and}\quad e_2\circ e_1=e_2
\end{equation}
Thus, the Leibniz product is actually a Lie bracket. The quadratic constraint is verified.

%Let us act on $\Theta$ with $E_{13}$ on the one hand, and  with $E_{32}$ on the other hand 
Let $\chi=E_{13}\cdot \Theta$ and $\psi=E_{32}\cdot\Theta$. We can compute their action on the basis vectors of $V$, using the following generic formula, valid for any vector of $T_{+1}=R_{\Theta}[-1]$:
\begin{equation}\label{eq:actiong2}
(a\cdot \chi)(x)=[a,\chi(x)]-\chi(a\cdot x)
\end{equation}
 for any $x\in V$. It gives:
\begin{align}
\chi(e_i)&=[E_{13},-E_{22}\delta_{i}^1+E_{21}\delta_i^2]-\Theta(e_1\delta_{i}^3)=-E_{23}\delta_i^2+E_{22}\delta_i^3\\
\psi(e_i)&=[E_{32},-E_{22}\delta_{i}^1+E_{21}\delta_i^2]-\Theta(e_3\delta_{i}^2)=-E_{32}\delta_{i}^1+E_{31}\delta_i^2
\end{align}
where $\delta^j_i$ is the Kronecker delta.
Then, we compute the right-hand side of Equation \eqref{eq:action2} in order to prove that the left-hand side is not zero, hence proving that Proposition \ref{experimental} in \cite{lavauLieAlgebraCrossed2023} was erroneous. We have:
\begin{equation}
\llbracket\chi,\llbracket\psi,e_i\rrbracket\rrbracket+\llbracket\psi,\llbracket\chi,e_i\rrbracket\rrbracket=\delta_i^1E_{32}\cdot\chi-\delta_i^2E_{31}\cdot \chi+\delta_i^2E_{23}\cdot\psi-\delta_i^3E_{22}\cdot \psi
\end{equation}
The first term cannot be canceled by any other and is not identically null since, by Equation \eqref{eq:actiong2}:
\begin{align}
E_{32}\cdot \chi(e_i)=[E_{32},-E_{23}\delta_i^2+E_{22}\delta_i^3]-\chi(e_3\delta^2_i)=-E_{33}\delta_i^2+E_{22}\delta_i^2+E_{32}\delta_i^3
\end{align}
This means that $\llbracket\chi,\llbracket\psi,e_1\rrbracket\rrbracket+\llbracket\psi,\llbracket\chi,e_1\rrbracket\rrbracket\neq0$, proving that $\llbracket \chi,\psi\rrbracket\neq0$. Notice that in that particular example,  although $T_2$ is not zero, it turns out that $T_3$ is, so every space of degree higher than 3 is identically zero.
\end{example}

\begin{example}\label{example1bis} This example illustrates that it is possible that $\mathrm{dim}(R_\Theta)>1$ while at the same time $T_i=0$ for every $i\geq2$. Let $\mathfrak{g}=\mathfrak{gl}_2(\mathbb{R})$ and $V=\mathbb{R}^2$ be the defining representation of $\mathfrak{g}$. As embedding tensor, we set:
\begin{equation}
\Theta(e_1)=-E_{12}\quad \text{and} \quad \Theta(e_2)=E_{11}
\end{equation}
The kernel of the embedding tensor is zero, so the Leibniz product induced on $V$ is actually Lie: $
e_1\circ e_2=-e_1$ and $e_2\circ e_1=e_1$, and there are no spaces in degrees lower than $-1$. The image of the embedding tensor in $\mathfrak{gl}_2(\mathbb{R})$ -- that is to say the subspace spanned by $E_{12}$ and $E_{11}$ -- acts trivially on $\Theta$. From Equation \eqref{eq:actiong2}, one deduces the action of the remaining two matrices on the embedding tensor:
\begin{align}
E_{21}\cdot\Theta(e_i)&=-E_{22}\delta_i^1+E_{21}\delta_i^2\\
E_{22}\cdot\Theta(e_i)&=E_{12}\delta_i^1-E_{11}\delta_i^2
\end{align}
One notices that $E_{22}\cdot \Theta=-\Theta$, while $\psi=E_{21}\cdot\Theta:V\to \mathfrak{g}$ is another embedding tensor as it satisfies the quadratic constraint. The only non-trivial actions of matrices of $\mathfrak{g}$ on $\psi$ are the following: $E_{11}\cdot\psi=-\psi$ and $E_{12}\cdot\psi=\Theta$.
The space $R_\Theta$ is thus two-dimensional and we only have to compute the nested bracket involving $\Theta$ and $\psi$:
\begin{equation}\label{eq:example3}
\llbracket\Theta,\llbracket\psi,e_i\rrbracket\rrbracket+\llbracket\psi,\llbracket\Theta,e_i\rrbracket\rrbracket=-\delta_i^1\Theta-\delta_i^2\psi+ \delta_i^1\Theta+\delta_i^2\psi=0
\end{equation}
There is thus no need to introduce the degree $2$ bracket $\llbracket\Theta,\psi\rrbracket$ so $T_2$, and more generally every $T_i$ for $i>1$, is zero. Finally, notice that Equation \eqref{eq:example3} is nothing but Equation \eqref{eq:zeroness}, as seen from  the present example.
\end{example}

\begin{example}\label{example2}
This example provides a case where the Lie algebra is the same as that of \ref{example1bis}, but where the representation $V$ is different, thus generating another orbit at level +1 and at higher levels as well.
This example is a continuation of Example \ref{lodayexample2} of \cite{lavauLieAlgebraCrossed2023}, where the spaces of positive degree had not been computed.
% by computing the bracket of any two elements of the $\mathfrak{g}$-orbit $R_D$ of the embedding tensor $D$ and show that this bracket vanishes.  %(in $\mathrm{Hom}(V,\mathfrak{g})$). 
In the present case, the Lie algebra is $\mathfrak{g}=\mathfrak{gl}_2(\mathbb{R})$, while the Leibniz algebra is $V=\mathcal{M}_{2\times 2}(\mathbb{R})$, equipped with the following Leibniz product:
\begin{equation}\label{majora3}
\begin{pmatrix}
a_{11}& a_{12}\\
a_{21}&a_{22}\\
\end{pmatrix}\circ \begin{pmatrix}
b_{11}& b_{12}\\
b_{21}&b_{22}\\
\end{pmatrix}=\begin{pmatrix}
0& (a_{11}-a_{22})b_{12}\\
-(a_{11}-a_{22})b_{21}&0\\
\end{pmatrix}
\end{equation}The action of $\mathfrak{g}$ on $V$ occurs through the matrix commutator.
The embedding tensor is given by:
\begin{equation}
D:	\begin{pmatrix}
a_{11}& a_{12}\\
a_{21}&a_{22}\\
\end{pmatrix}\xmapsto{\hspace*{1.2cm}}\begin{pmatrix}
a_{11}& 0\\
0&a_{22}\\
\end{pmatrix}
\end{equation}
Using Equation \eqref{eq:actiong2}, one deduces the action of $\mathfrak{g}$ on the embedding tensor $D$.
This action was originally given by Equation \eqref{mixmiya} in \cite{lavauLieAlgebraCrossed2023} as:
\begin{equation}\label{mixmiya}
\eta(a;D)(b)=[a,Db]_A-D([a,b])=\begin{pmatrix}
b_{12}a_{21}-a_{12}b_{21}& -(b_{11}-b_{22})a_{12}\\
(b_{11}-b_{22})a_{21}&b_{21}a_{12}-a_{21}b_{12}\\
\end{pmatrix}
\end{equation}
where $[a,b]$ denotes the commutator of matrices.
 %We notice from Equation \eqref{mixmiya} that the embedding tensor is $\mathrm{Im}(D)$-equivariant, as expected.

%Using Equation \eqref{eq:actiong2}, one deduces the action of $\mathfrak{g}$ on the embedding tensor $D$.
Equation \eqref{mixmiya} implies two things. On the one hand, $D$ is invariant under the action of the two matrices $E_{11}$ and $E_{22}$, which translates the fact that $D$ is $\mathrm{Im}(D)$-equivariant, as it should be. On the other hand,
the actions of $E_{12}$ and $E_{21}$ on $D$ define two linear maps $D^+:V\to \mathfrak{g}$ and $D^-:V\to \mathfrak{g}$ whose actions % on a matrix $a$ %$a=\begin{pmatrix}a_{11}&a_{12}\\a_{21}&a_{22}\end{pmatrix}$ 
are:
\begin{equation}
D^+(a)=\begin{pmatrix}-a_{21}&a_{22}-a_{11}\\0&a_{21}\end{pmatrix}\quad \text{and}\quad D^-(a)=\begin{pmatrix}a_{12}&0\\a_{11}-a_{22}&-a_{12}\end{pmatrix}
\end{equation}
%where $\alpha=a_{22}-a_11$.
Using  Equation \eqref{eq:actiong2}, one observes that %the  $V^+$ and $V^-$ are stable under 
the actions of $E_{12}$ and $E_{21}$ generate
three other linear operators $D^{++}, D^{--}$ and $D^0$ from $V$ to $\mathfrak{g}$:
\begin{align}
&D^{++}(a)=E_{12}\cdot D^+(a)=4\begin{pmatrix}0&a_{21}\\0&0\end{pmatrix},\quad D^{--}(a)=E_{21}\cdot D^-(a)=4\begin{pmatrix}0&0\\a_{12}&0\end{pmatrix}\\
&D^0(a)=E_{21}\cdot D^+(a)=E_{12}\cdot D^-(a)=2\begin{pmatrix}a_{11}-a_{22}&-a_{12}\\-a_{21}&a_{22}-a_{11}\end{pmatrix}
\end{align}
Each one of these linear maps from $V$ to $\mathfrak{g}$ are eigenvectors of the actions of the diagonal matrices (the image of $D$).
 %\begin{equation}
 %E_{ii}\cdot e=(-1)^{i-1}e\quad \text{and}\quad  E_{22}\cdot f=(-1)^{i}f
 %\end{equation}
  %On the other hand, the actions of $E_{12}$ (resp. $E_{21}$) on $f$ (resp. on $e$) give a linear map $h:V\to \mathfrak{g}$ defined by:
 %\begin{equation}
 %h(a)=\begin{pmatrix}-\alpha&0\\0&\alpha\end{pmatrix}
 %\end{equation}
%This map is invariant under $E_{11}$ and $E_{22}$, while $E_{12}\cdot h=2e$ and $E_{21}\cdot h=-2f$.
The $\mathfrak{g}$-orbit $R_D\subset\mathrm{Hom}(V,\mathfrak{g})$ is then spanned by the linear maps $D, D^+, D^-, D^{++}, D^{--}$ and $D^0$, so it is 6-dimensional.

By computing the actions of every matrix of $\mathfrak{g}=\mathfrak{gl}_2(\mathbb{R})$ on the generators $D^+, D^-, D^{++}, D^{--}$ and $D^0$ via Equation \eqref{eq:actiong2}, one deduces that they form an irreducible representation of $\mathfrak{sl}_2(\mathbb{R})$. Indeed, the vectors:
\begin{equation}
v_0=D^{++}, \quad v_1=4D^+, \quad v_2=2D^0, \quad v_3=4D^-\quad \text{and}\quad v_4=D^{--},
\end{equation}
satisfy the usual properties of the 5-dimensional irreducible representation of $\mathfrak{sl}_2(\mathbb{R})$, denoted $V_4$. One can check that the operator $\Omega=6D-D^0$ generates the one-dimensional trivial representation of $\mathfrak{sl}_2(\mathbb{R})$ -- denoted $V_0$ --  so the $\mathfrak{g}$-orbit $R_D$ is fully reducible: $R_D=V_0\oplus V_4$. In particular, the embedding tensor $D$ has one leg in each of the representations.

Then, by denoting $\llbracket \chi,a\rrbracket=\chi(a)$ for every $\chi\in R_D$, one observes that: 
\begin{equation}\label{eq:rhside}
\llbracket D^+,\llbracket D^-,a\rrbracket\rrbracket+\llbracket D^-,\llbracket D^+,a\rrbracket\rrbracket=-2a_{12}D^+-2a_{21}D^-
\end{equation} %Notice as well that we obtain the same situation as in $D=$
The right-hand side being non-zero for a general matrix $a\in V$, one deduces that $\llbracket D^+,D^-\rrbracket$ cannot be zero, and rather should be set to act on matrices as the right-hand side of Equation \eqref{eq:rhside}. % follows: $\llbracket\llbracket D^+,D^-\rrbracket,a\rrbracket=-2a_{12}D^+-2a_{21}D^-$.
 The space $R_D$ being 6-dimensional and the bracket of two elements of $R_D$ being symmetric, one should compute 21 brackets in total to exhaust all the possibilities and find a basis for the space $T_{2}$. The dimension of the latter is strictly lower than 21 however, since some of the brackets vanish, as can be seen from the two following examples:
\begin{equation}
\llbracket D^{++},\llbracket D^{++},a\rrbracket\rrbracket=-4a_{21}E_{12}\cdot D^{++}=0\qquad\text{and}\qquad\llbracket D^{--},\llbracket D^{--},a\rrbracket\rrbracket=-4a_{12}E_{21}\cdot D^{--}=0
\end{equation}

%Upon observing that $D(a)=\frac{1}{2}(h(a)+\mathrm{tr}(a)\mathrm{I_2})$, one can check that $R_D$ is the adjoint representation of $\mathfrak{g}$, with the linear maps $e,f,h$ corresponding to $\mathfrak{sl}_2(\mathbb{R})$, while $2D-h$ is the singlet.

\end{example}

 \bigskip

 \section*{Acknowledgements:} 
The first author wants to thank the University of Pennsylvania, where the project started, for hosting him during one month. He is also grateful to Georges Skandalis and St\'ephane Vassout for having funded most of his research under the ANR project `Analysis on singular and non compact spaces. A $C^*$-algebra approach'. The final stage of this research could not have happened without the financial sustainability of the French social welfare system, as well as  scholarship No. 075-15-2019-1620 from the government of the Russian Federation. Many thanks to  Christian Saemann and Henning Samtleben for their helpful comments during the preparation of the paper. We are indebted to Jakob Palmkvist for his insightful ideas on taking a quotient of the free graded Lie algebra and to the anonymous referees who encouraged us to ponder several questions that led us to improve the first version of this paper. Thanks to Marco Manetti and Mario Trigiante whose comments benefited the current version. We thank Juan Moreno and Jonathan Wise for providing the elegant Lemma~\ref{snake}, which is essential to the construction of the tensor hierarchy. We finally would like to thank Johannes Huebschmann for his careful reading and remarks in the final stage of redaction. Concerning the appendix, Sylvain Lavau and Jakob Palmkvist are grateful to the Mainz Institute for Theoretical Physics (MITP) of the DFG Cluster of Excellence PRISMA* (Project ID 39083149), for its hospitality and support during the workshop "Higher Structures, Gravity and Fields" in January 2023, when this work was completed. We would also like to thank the organizers Martin Cederwall and Henning Samtleben for insightful comments and discussions.

\newpage

\bibliography{mybig} %commenter cette ligne sur overleaf
%\input{Main.bbl}
%\printbibliography %decommenter cette ligne sur overleaf

\end{document}